\documentclass[
    letterpaper,
    twoside,
    aps,
    prx,
    nofootinbib,
    showpacs,
    showkeys, 
    superscriptaddress, 
    longbibliography,
    twocolumn,
]{revtex4-2}

\usepackage{tabularx}
\usepackage{hyperref}
\usepackage{amsmath, mathtools, amssymb, dsfont, amsthm}
\usepackage{physics}
\usepackage{enumitem}
\usepackage[T1]{fontenc}
\usepackage{makecell}

\newtheorem{lemma}{Lemma}

\newtheorem{theorem}{Theorem}

\newtheorem{definition}{Definition}
\newtheorem{conjecture}{Conjecture}
\newtheorem{criterion}{Criterion}

\theoremstyle{definition}
\usepackage{xcolor}

\usepackage{placeins}

\usepackage{float}
\makeatletter
\let\newfloat\newfloat@ltx
\makeatother

\usepackage{algorithm}
\usepackage{algpseudocode}

\usepackage[normalem]{ulem}
\newcommand{\colorred}{\color{black}}
\newcommand{\added}[1]{{\colorred #1}}

\usepackage[capitalise]{cleveref}

\Crefname{criterion}{Criterion}{Criteria}
\crefname{criterion}{Crit.}{Crit.}
\Crefname{definition}{Definition}{Definitions}
\crefname{definition}{Def.}{Defs.}
\Crefname{theorem}{Theorem}{Theorems}
\crefname{theorem}{Thm.}{Thms.}
\Crefname{proposition}{Proposition}{Propositions}
\crefname{proposition}{Prop.}{Props.}
\Crefname{conjecture}{Conjecture}{Conjectures}
\crefname{conjecture}{Conj.}{Conjs.}

\usepackage{orcidlink}

\begin{document}

\title{Decoding across transversal Clifford gates in the surface code}

\author{Marc Serra-Peralta\,\orcidlink{0000-0002-8000-8701}}
\altaffiliation{These authors contributed equally to this work. Contact author: m.serraperalta@tudelft.nl}
\author{Mackenzie H. Shaw\,\orcidlink{0000-0002-0776-886X}}
\altaffiliation{These authors contributed equally to this work. Contact author: m.serraperalta@tudelft.nl}
\author{Barbara M. Terhal\,\orcidlink{0000-0003-0218-6614}}
\affiliation{QuTech, Delft University of Technology, Lorentzweg 1, 2628 CJ Delft, The Netherlands, and Delft Institute of Applied Mathematics, Delft University of Technology, Mekelweg 4, 2628 CD Delft, The Netherlands}

\date{\today}

\begin{abstract}
Transversal logical gates offer the opportunity for fast and low-noise logic, particularly when interspersed by a single round of parity check measurements of the underlying code. Using such circuits for the surface code requires decoding across logical gates, complicating the decoding task. We show how one can decode across an arbitrary sequence of transversal gates for the unrotated surface code, using a fast ``logical observable'' minimum-weight-perfect-matching (\textsc{mwpm}) based decoder, and benchmark its performance in Clifford circuits under circuit-level noise. We propose {\em windowed} logical observable matching decoders to address the problem of fully efficient decoding: our basic windowed decoder is computationally efficient under the restriction of quiescent (slow) resets. Our `advanced' two-step windowed decoder can be computationally inefficient but allows fast resets. For both windowed decoders we identify errors which scale sublinearly in $d$---depending on the structure of the circuit---which can lead to logical failure, and we propose methods to adapt the decoding to remove such failures. Our work highlights the complexity and interest in efficient decoding of fast logic for the surface code.
\end{abstract}

\maketitle
\tableofcontents
\floatstyle{ruled}

\section{Introduction}

One of the many challenges in performing fault-tolerant quantum computing is that of implementing a fault-tolerant gate set in a quantum error-correction (QEC) code. Transversal gates are an attractive way of implementing gates because they are both fast and inherently low noise. Although it is impossible to construct a \textit{universal} transversal gate set in a single QEC code~\cite{ZCC,EK}, the Clifford group can be implemented transversally in both the 2D color and surface codes~\cite{Bombin_2006,kubica2015unfolding,moussa2016transversal,Breuckmann_2024}, leaving the $T$ gate to be done by some fault-tolerant state preparation and injection. In the surface code some of these gates have been referred to as ``fold''-transversal because they involve two-qubit gates between qubits in the same code block that are adjacent to each other when the surface code is ``folded'' along its diagonal, see \cref{fig:fold_trans_gates}. Because of this beyond-strictly-2D connectivity, transversal gates are particularly attractive in platforms with mobile qubits such as neutral atom qubits~\cite{Bluvstein_2023} or trapped-ion qubit chips (see e.g.~Ref.~\cite{Quantinuum2}). The fault-tolerance of transversal gates has been known for a long time, but it has only recently been proven that transversal gates and magic state injections can be performed \textit{quickly} (that is, with only a constant number of QEC rounds between each gate) without compromising the fault-tolerance of the algorithm~\cite{cain2024,algo-FT}. This stands in contrast with the purely-2D procedure of lattice surgery in the surface code, which can only perform a parity measurement on a subset of logical qubits once every $\Theta(d)$ QEC rounds where $d$ is the distance of the code.

Transversal Clifford gates therefore offer the opportunity for faster fault-tolerant logic---referred by some as ``constant-time'' logical gates~\cite{zhang2025}---but comes at the price of more complex decoding. In particular, in order for measurement errors to be correctly diagnosed, any surface code decoder must look at a syndrome record which has a temporal width that grows at least linearly with $d$~\cite{Dennis_2002, Campbell_2019},~\footnote{Note that this does not contradict the known existence of hierarchical cellular automaton decoders for the surface code which use local feedback corrections based on noisy parity check data, see Ref.~\cite{bala:automaton} and references therein.}. In other words, we need to decode \textit{across} transversal logical gates. Because transversal logical gates can map stabilizer generators to \textit{products} of generators, error detectors~\cite{McEwen_2023} may need to be set by the joint parity of more than two measurement outcomes. As a consequence, single measurement errors can flip more than two detectors, corresponding to so-called hyperedges in a decoding hypergraph~\cite{Sundaresan_2023}. Even in the surface code, the fast and simple minimum-weight perfect matching (\textsc{mwpm}) algorithm therefore cannot be directly applied to the full decoding problem.

To overcome this issue, Ref.~\cite{algo-FT} used a minimum-weight \textit{hypergraph} decoder to prove their results. Such a decoder is inefficient both because minimum-weight decoding of a hypergraph is inefficient in general, and because the cost of the decoder scales with the depth of the algorithm being run.

In this paper we make progress towards designing an efficient \textsc{mwpm}-based decoder across fast transversal gates in the surface code. After reviewing previous work in \cref{sec:prev}, we describe our main set-up and concepts in \cref{sec:prelim}. We avoid minimum-weight hypergraph decoding with a logical observing matching (\textsc{lom}) decoder (\cref{sec:logical_observable_decoder}), which applies \textsc{mwpm} only to matchable subgraphs of the decoding hypergraph. 
Under a ``basic’’ error model (defined in \cref{sec:noise_detectors_graph}), the \textsc{lom} decoder can correct any error of weight $<d/2$ in the unrotated surface code in arbitrary circuits with fast transversal Clifford gates and $T$ gate injections which use fault-tolerantly prepared $\ket{T}$ magic states. 
This stands in contrast to existing hypergraph decomposition methods in the literature, including both ``splitting’’ decoders and ``hierarchical’’ matching decoders~\cite{sahay2024, wan2024, guernut2024, chen2024transversal, beverland2021cost}, whose failures we discuss in \cref{sec:fail_split,sec:hier} respectively. 
One key detail in the \textsc{lom} decoder is how to handle so-called fragile logical measurement observables, that we discuss pedagogically in \cref{sec:fragile}. 
\added{Our fault-tolerant results are not specific to \textsc{mwpm} and hold for any (distance-preserving) graph-based decoder, such as Union Find~\cite{delfosse2021almost}. }
Numerical benchmarks of the \textsc{lom} decoder show high thresholds and accuracy compared to minimum-weight decoding under both phenomenological and circuit-level depolarizing noise and for both repeated and random logical two-qubit Clifford gate circuits (\cref{sec:num}).

However, the \textsc{lom} decoder is still inefficient because it has a decoding volume that scales with the depth of the circuit. To this end, we propose a \textit{windowed} version of the \textsc{lom} decoder in \cref{sec:window}. It turns out to be non-trivial to design a windowed matching decoder for circuits with fast gates, resets and measurements, that is both computationally efficient and able to correct all basic errors of weight $<d/2$. Instead, we introduce two variants of the \textit{windowed}-\textsc{lom} decoder: a ``basic'' variant that is computationally efficient but cannot decode circuits with fast resets, and a ``two-step'' variant that can decode circuits with fast resets but is no longer guaranteed to be efficient. Both of these variants require further modifications to be able to correct all low-weight error patterns, one of which requires synchronized resets and measurements (except those in $T$ injection circuits) to occur only every $\Omega(d)$ steps, but we conjecture that with these modifications both variants of the windowed-\textsc{lom} decoder can correct up to $d/2$ basic errors and have a noise threshold.
We end the paper with a discussion in \cref{sec:discuss} on possible future work, including the use of independent observer-decoders more generally in decoding.

\subsection{Previous and simultaneous work}
\label{sec:prev}
In Ref.~\cite{algo-FT}, it was first proved and numerically demonstrated that one can implement fault-tolerant quantum logic by having $O(1)$ QEC rounds after each transversal gate for, say, the unrotated surface code, supplemented by $\ket{T}$ magic states. The constructions in that paper and the earlier paper on correlated decoding, Ref.~\cite{cain2024}, invoked minimum-weight decoders or used decoders that can explicitly handle hyperedges, thus not necessarily being very efficient. Our goal here is to show how to construct an efficient fault-tolerant {\em matching}-based decoder.

Ref.~\cite{delfosse2023} proposed an algorithm to decompose hyperedges into edges (see also~\cite{gidney2021stim}), but we show that it does not perform well when decoding across transversal gates in Section~\ref{sec:fail_split}. \added{Refs.~\cite{wan2024, turner2025} decompose the hyperedges and iteratively run \textsc{mwpm} to achieve good logical performance. However, no formal guarantees are proven for the number of iterations and the logical performance for an arbitrary circuit and increasing code distance.} 
Furthermore, the hierarchical \textsc{mwpm} decoding methods from Refs.~\cite{sahay2024, guernut2024, chen2024transversal, beverland2021cost} cannot be used to fault-tolerantly decode an {\em arbitrary} sequence of transversal logical Clifford gates, as we argue in Section~\ref{sec:hier}.

Ref.~\cite{zhang2025} addresses similar questions of decoding as our work, with numerics decoding transversal logical gates such as the CNOT and $H$ for the surface code, interspersed by single QEC rounds. Ref.~\cite{zhang2025} provides a windowed decoder, but it does not reduce decoding to a matching problem, instead it uses heuristic approximate solvers to deal with hyperedges. \added{A limitation for this type of approximate solvers is their poor embedding into specific hardware for running real-time decoding. On the other hand, Refs.~\cite{caune2024demonstrating, google2024_belowthreshold} have already performed real-time decoding on experimental data with graph-based decoders, such as MWPM. }
Windowed, also called sliding-window, matching decoders for memory experiments in the surface code have previously been considered and used in e.g. Refs.~\cite{Dennis_2002,Tomita_2014, O_brien_2017,skoric2023parallel, window_tan2023, viszlai2024}.

At the time of writing this paper and presenting our work at the APS March meeting 2025, we became aware of similar work on fast matching-based decoding of transversal gates for the surface code~\cite{cain+:upcoming}.

\begin{figure*}[tb]
    \centering
    \includegraphics[width=0.8\linewidth]{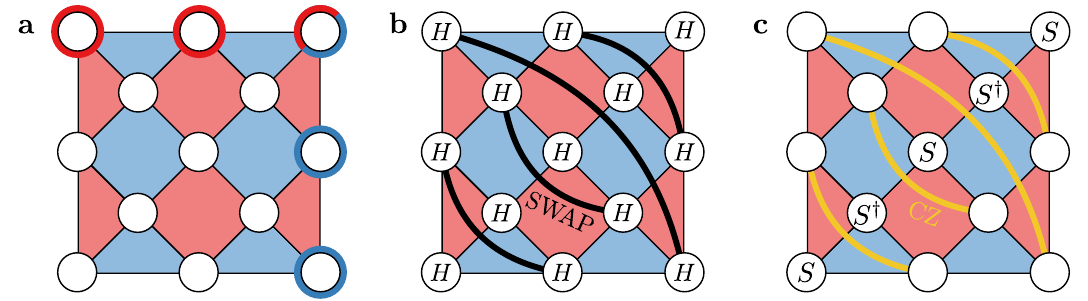}
    \caption{(a) The $d=3$ unrotated surface with data qubits as white circles and $X$- (resp. $Z$-) stabilizers as shaded red (resp. blue) regions. Representatives of the Pauli $\overline{X}$ (reps. $\overline{Z}$) are shown by the red (resp. blue) circles around the data qubits on which it has support. (b, c) The fold-transversal $\overline{H}$ and $\overline{S}$ gates, respectively. The logical Hadamard gate is implemented as a layer of Hadamard gates followed by a layer of SWAP gates. The long-range SWAPs (resp. CZs) are shown as thick black (resp. yellow) lines. }
    \label{fig:fold_trans_gates}
\end{figure*}

\section{Preliminaries}
\label{sec:prelim}

We begin this work with a pedagogical introduction to the concepts we will use throughout the paper.

\subsection{Circuits, (fragile) observables, and observing regions}
\label{sec:coo}

We consider an arbitrary (unencoded) quantum circuit $\mathcal{C}$ compiled into $\ket{0}$, $\ket{+}$ and $\ket{T}$ state preparation, destructive $Z$- and $X$-basis measurements, and Clifford gates that may be conditioned on the outcome of those measurements. 
If a measurement is used to condition a Clifford gate, we call it a {\em conditioning} measurement; otherwise, we call it a {\em final} measurement. Each different set of outcomes of conditioning measurements leads to a different \textit{realization} of the circuit. The outcome of the computation is then determined from some subset of the final measurements.

We consider each layer of resets, gates and/or measurements to take place at a time labeled by a half-integer, so $t+1/2$ with $t\in\mathbb{Z}$. With this, we say that a \textit{circuit location} is a point in the circuit in between layers of gates described by a set of coordinates $(t,j)$, where $t\in\mathbb{Z}$ represents time and $j$ is a qubit index. 

We make use of \textit{Pauli regions} of the circuit, which is formally a list of Pauli operators $P=[P_{(t,j)}]$, with each element of the list labeled by a circuit location. Each element is itself a single-qubit Pauli operator $I$, $X$, $Y$ or $Z$ (ignoring phases). Two Pauli regions $P_{1}$ and $P_{2}$ can be multiplied to each other element-wise by multiplying the Pauli operators at each circuit location. We say that two Pauli regions anticommute if their components anticommute at an odd number of circuit locations; otherwise we say they commute. We also consider time-slices of Pauli regions denoted $P_{t}$, which is a Pauli operator acting on all of the qubits that are active at time $t$. Note that a Pauli region would be a Pauli operator if we associated a qubit with each circuit location as is done for space-time codes~\cite{BFHS:spacetime}.

We define an \textit{observable} $O=\{M_{1},M_2,\dots\}$ as any non-empty subset of measurements in the circuit $\mathcal{C}$ where the outcome of the observable is the {\em parity} of the outcomes of each measurement $M_i$ in the observable (the measurements $M_i$ can be conditioning and/or final measurements and either $X$ or $Z$-measurements).  In a slight abuse of notation we will also use $O$ to refer to the \textit{Pauli representation} of the observable, which is the Pauli region that has a Pauli $Z$ at the circuit locations immediately preceding each of the $Z$-measurements in the observable (and likewise a Pauli $X$ preceding an $X$-measurement), and identity operators elsewhere. This way two observables $O_{1}$ and $O_{2}$ can be multiplied by each other, constructing $O_1 O_2$ by multiplying their Pauli representations. Later in \cref{sec:window}, we will additionally refer to observables that do not correspond to measurements in the circuit, and instead are just a Pauli operator in a circuit that hypothetically \textit{could} be measured.

Given a realization of the circuit, we define the \textit{observing region} $O^{\leftarrow}$, which is itself a Pauli region, of the observable $O$ as follows. This object has been called the spackle of a Pauli string in Ref.~\cite{BFHS:spacetime}, the back-cumulant of $O$ in Ref.~\cite{delfosse2023:spacetime}, and has also been used in Refs.~\cite{gottesman2022,McEwen_2023}. We consider all the measurement locations $M_i$ of the observable and apply the following rules:
\begin{enumerate}
    \item \label{rule1} Begin by setting $O^{\leftarrow}=O$, the Pauli representation of the observable; that is, for all circuit locations $(t,j)$, set $O_{(t,j)}=X$ (respectively, $Z$) if there is a $X$-measurement ($Z$-measurement) at $(t+1/2,j)$ which is in the observable $O$, otherwise set $O_{(t,j)}=I$.
    \item \label{rule2} Starting at the last time-step on which the observable has support and proceeding iteratively backwards in time, apply the following rule. If the circuit locations $\{(t,j): j\in J\}$ occur before a layer of Clifford gates $C$ at time $t+1/2$ for some subset of qubits $J$, set $O^{\leftarrow}_{(t,J)}=C^{\dag}{O}^{\leftarrow}_{(t+1,J)} C$, where $O^{\leftarrow}_{(t+1,J)}$ is the $|J|$-qubit Pauli operator acting on each qubit according to $O^{\leftarrow}_{(t+1,j)}$ for $j\in J$.
\end{enumerate}
Note that when the backpropagation hits a reset or a $T$-ancilla preparation, nothing backpropagates further. Note also that the observing region can only be calculated after all the conditional gates that occur before the last measurement included in $O$ have been fixed. Loosely speaking, the observing region $O^{\leftarrow}$ is obtained as the ``backpropagation'' of the Pauli representation of the observable $O$. Importantly, the observing region also corresponds to the region of the circuit where an error in the circuit flips the outcome of the observable: if $O^{\leftarrow}_{(t,j)}=Z$, then a Pauli $X$ or $Y$ error at the circuit location $(t,j)$ flips the value of the observable. 

We also consider \textit{reset stabilizers} of the circuit. Analogous to our definition of an observable $O$ as a set of measurements, we define a reset stabilizer $S$ as a set $\{R_{1},\dots\}$ of $\ket{+}$- or $\ket{0}$-state resets in the circuit $\mathcal{C}$. 
The Pauli representation of a reset stabilizer $S$ is the Pauli region which has $X$ (respectively, $Z$) operators at the circuit locations immediately following each of the $\ket{+}$-resets (respectively, $\ket{0}$-resets) in $S$, while the reset stabilizing region $S^{\rightarrow}$ is the ``forward-propagation'' (cumulant in Ref.~\cite{delfosse2023:spacetime}) of $S$ using a time-reversed version of Steps \ref{rule1}-\ref{rule2} used to define the observing region. An important fact that we prove for completeness (see e.g. Ref.~\cite{gottesman2022} and Proposition 3 in Ref.~\cite{delfosse2023:spacetime}) in Section~I of the Supplemental Material~(SM) is the following:
\begin{lemma}\label{lem:propa}
    $S^{\rightarrow}$ and $O$ anticommute if and only if $S$ and $O^{\leftarrow}$ anticommute.
\end{lemma}

Lastly, we define fragile, or unreliable, observables. The relevance of fragile observables relates to the encoded execution of the circuit $\mathcal{C}$ and the existence of fragile time-boundaries, see Section \ref{sec:frag-time}, but its definition depends only on $\mathcal{C}$ itself, as follows:

\begin{definition}[Fragile observable]
     We say that an observable $O$ is \textit{fragile} if its observing region $O^{\leftarrow}$ anticommutes with at least one reset stabilizer $S$ of the circuit; or, equivalently via Lemma \ref{lem:propa}, if the observable $O$ anticommutes with at least one reset stabilizing region $S^{\rightarrow}$ of the circuit.
\label{def:fo}
\end{definition}

In the absence of conditional Clifford gates, a fragile observable always has a 50-50 measurement outcome, since applying the reset stabilizer does not change the state while it changes the measurement outcome of the observable. In the presence of conditional-Clifford gates, an observable can be fragile or not depending on which Clifford gates are applied, and can even have a {\em deterministic} outcome (an example is circuit \ref{eq:2_T_gate_circuit} where the observable $O_4$ is fragile if both conditional $S$ gates are applied). Fragile observables have been called ``non-deterministic'' observables in Ref.~\cite{algo-FT}. We will call non-fragile observables {\em reliable} observables. Note that non-fragile observables are not necessarily deterministic, due to the presence of $\ket{T}$-state injections.

\subsection{Fault-tolerant circuit implementation in the surface code}
\label{sec:ft-surface}

To implement $\mathcal{C}$ fault-tolerantly, we encode each qubit into a distance-$d$ unrotated surface code, using $n=2d^{2}-2d+1$ data qubits. Throughout the logical computation we perform one QEC round of parity check measurements after each circuit layer. 
We compile each circuit element in $\mathcal{C}$ using the following logical operations on the unrotated surface code:
\begin{enumerate}
    \item Transversal logical $\ket{\overline{+}}$ (respectively, $\ket{\overline{0}}$) state preparation; that is, the preparation of $\ket{+}^{\otimes n}$ ($\ket{0}^{\otimes n}$) on each data qubit.
    \item Fault-tolerant magic state preparation, which could be achieved by, for example, magic state distillation. Importantly, we assume that the values of all $X$- and $Z$-stabilizers are known after the preparation of the state with an error rate less than $\sim p^{d/2}$ for some small $p \ll 1$.
    \item Transversal logical $\overline{X}$ (respectively, $\overline{Z}$) measurement; that is, an $X$-basis ($Z$-basis) measurement of each data qubit in the $n$-qubit code block, which is processed, together with other measurement data, by the decoder, see Section \ref{sec:noise_detectors_graph}.
    \item Fold-transversal $\overline{S}$ and $\overline{H}$ gates as depicted in Fig.~\ref{fig:fold_trans_gates} and the
    transversal $\overline{\text{CNOT}}$ gate, possibly conditioned on transversal logical measurements.
\end{enumerate}
All Pauli operations---both those arising from decoding and those directly in $\mathcal{C}$---are tracked in software. We denote the implementation of $\mathcal{C}$ encoded in the surface code as $\overline{\mathcal{C}}$, and when required for clarity we refer to $\mathcal{C}$ as the \textit{bare} circuit and $\overline{\mathcal{C}}$ as the \textit{encoded} circuit.

\subsection{Error model, detectors in the pre-gate frame, and decoding graph} \label{sec:noise_detectors_graph}

For all results in \cref{sec:logical_observable_decoder} and \cref{sec:window} we assume a \emph{basic} error model in the circuit $\overline{\mathcal{C}}$ which is as follows:
\begin{itemize}
	\item independently apply Pauli $X$ and Pauli $Z$, each with probability $p$, to each data qubit before each logical gate (followed by the single QEC round),
	\item apply an $X$ error (resp.~$Z$ error) with probability $p$ before a physical qubit $Z$-measurement (resp.~$X$-measurement).
\end{itemize}
Thus in this error model, no errors are inserted prior to the QEC round, but each parity check measurement in the QEC round can be faulty due to the insertion of an error prior to the ancilla qubit measurement. \added{We will show below that decoding efficiently across transversal gates is a non-trivial problem to solve even with this basic error model. Therefore, our strategy in this paper is to design a fast decoder that can efficiently decode errors in the basic model and prove its fault-tolerance (\cref{thm:single-LOM_FT})} However, these results directly bear on circuit-level noise since circuit-level errors have the same effect as products of a {\em constant} number of errors in the basic model, and we test this in the numerics in Section~\ref{sec:num} where we model more realistic phenomenological and circuit-level depolarizing noise models. 
\added{One could also apply known heuristic decoding improvements to \textsc{mwpm} to optimize decoding for circuit-level noise, however this is not the focus of this paper}.



We define a set of \textit{detectors}---sets of QEC measurement outcomes whose parity is deterministic in the absence of errors~\cite{McEwen_2023}. In the absence of logical operations, each detector simply consists of two measurements of the same ancilla qubit in consecutive QEC rounds, as is standard for decoding memory experiments on the surface code. If the measured outcome of a detector differs from its ideal outcome, we call this a \textit{defect}, and the full set of defects is the error \textit{syndrome}. Given the definition of a detector, one can define its \textit{detecting region}~\cite{McEwen_2023}, namely the Pauli region obtained by backpropagating the Pauli representation of the detector as observable, defined in Section \ref{sec:coo}. 


In the presence of transversal logical gates, we wish to find a set of detectors such that the detecting regions are still local in space-time. Such locality is desirable because it implies that each detector can only be flipped by a constant number of errors in the circuit. 
Because parity checks are spread non-trivially by the logical gates, the detector definitions that achieve this, depend on the logical gates that are applied. 
In this paper we choose a so-called \textit{pre-gate} frame for the detectors. In this frame, each $Z$- (resp.~$X$-) detector only involves a single $Z$ (resp.~$X$) parity measurement outcome before the logical gate and a possibly non-trivial linear combination of binary measurement outcomes after the logical gate, see its precise definition in Section~II in the SM~\cite{supp}. A very similar frame choice is the \textit{post-gate} frame where each detector only involves a single parity measurement outcome after the logical gate. These detector frames have been used in other previous works on decoding logical gates, see e.g.~Refs.~\cite{cain2024, wan2024}.

\begin{figure*}[tb]
    \centering
    \includegraphics[width=0.99\textwidth]{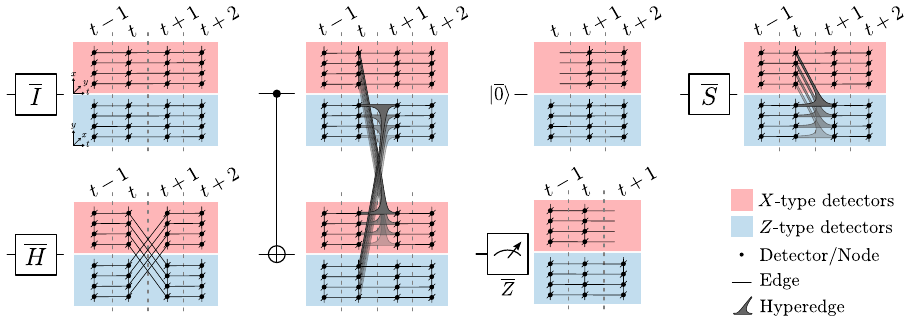}
    \caption{Decoding hypergraphs $\mathcal{G}$ when executing a logical operation assuming a basic error model, using the pre-gate frame. Time flows from left to right and the vertical gray dashed lines denote half-integer layers when a logical gate followed by a QEC round happens, so that detectors are labeled by integer times $t\in \mathbb{Z}$. In the subfigures, the logical gate+QEC round happens at layer $t+1/2$ while for all other layers we perform the identity gate. Logical reset and measurement happen at $t+1/2$ and are only shown for $\overline{Z}$, as the decoding hypergraphs for the $\overline{X}$-basis are the same but with the $X$ and $Z$ detector types swapped. The decoding hypergraphs are three-dimensional for the surface code, but we have drawn a two-dimensional slice for easier visualization. Note that for the logical $\ket{\overline{0}}$ there are no weight-1 time-like edges---the time-boundary is closed---while for the $X$-detectors there are weight-one time-like edges---the time-boundary is open.}
    \label{fig:hyperedges}
\end{figure*}

Given an observable $O=\{M_{1},M_2,\dots\}$ of the bare circuit $\mathcal{C}$, we can lift this to a \textit{logical} observable $\overline{O}$ of the encoded circuit $\overline{\mathcal{C}}$ by replacing each bare measurement $M_{i}$ in $O$ with a set of physical measurements that are in the support of the logical operator in $\mathcal{C}$. Throughout the manuscript we choose a logical representative that consists of a string of $X$ (resp.~$Z$) operators along a spatial boundary of the surface code, as shown in \cref{fig:fold_trans_gates}(a). The encoding gives rise to an observing region $\overline{O}^{\leftarrow}$ of the fault-tolerant circuit $\overline{\mathcal{C}}$ that corresponds to all the physical circuit locations where a Pauli flip would flip the recorded outcome of the observable $\overline{O}$. In a more general circuit-level noise model, $\overline{O}^{\leftarrow}$ can include locations on both data and ancilla qubits.

The product of the outcomes of each of the physical measurements that make up $\overline{O}$ gives the \textit{raw} outcome of $\overline{O}$. The task of the decoder is to take the outcomes of all the detectors, predict the error that occurred, and then to infer which observables are flipped by this error. Combining the decoder's prediction with the raw outcome of $\overline{O}$ gives the \textit{error-corrected} outcome of $\overline{O}$.

Using the basic error model and the definition of detectors, we construct a decoding hypergraph $\mathcal{G}=(\mathcal{V},\mathcal{H})$ consisting of vertices $v \in \mathcal{V}$ for each detector in $\overline{\mathcal{C}}$. There is a hyperedge $h\in \mathcal{H}$ for each $X$ or $Z$ error in the basic error model, such that the endpoints of each hyperedge correspond to the detector(s) that are flipped due to the error \footnote{For brevity we use the term \textit{hyperedge} to refer to hyperedges of any weight, meaning vertex support---including weight two---while we reserve the term \textit{edge} to only refer to weight-2 hyperedges.}. For each observable $\overline{O}$ of the circuit, there is a subset of the hyperedges $\mathcal{H}_{O}$ that cause a flip to the observable, we call this subset of hyperedges collectively the \textit{observing hyperedge set}, namely $\mathcal{H}_{O}$ is the set of hyperedges that correspond to an error that anticommutes with the logical observing region $\overline{O}^{\leftarrow}$. 
In our basic error model, for any observable $\overline{O}$, the observing hyperedge set $\mathcal{H}_{O}$ is a set of edges, i.e. it consists only of a set of space-like edges on one of the spatial boundaries of $\mathcal{G}$, see an example in Fig.~\ref{fig:decoding_subgraphs}. Hence we will refer to $\mathcal{H}_{O}$ as the observing edge set in the remainder of this paper. For circuit-level noise $\mathcal{H}_O$ can be a set of hyperedges. \added{Note that the observing edge set is closely related to the observing region: the observing region is defined at the level of unencoded circuits, but it has a corresponding expression at the physical qubit and decoding graph level (depending on the error model) that exactly coincides with the observing edge set\footnote{\added{Observing edge sets have been previously depicted in 3D decoding hypergraphs as logical sheets, see Refs.~\cite{chamberland2022universal, caune2024demonstrating}.}}.} The structure of $\mathcal{H}_{O}$ will be important in the following section when we define our single logical observable matching decoder.

Let's discuss the structure of the decoding hypergraph $\mathcal{G}$ for our basic error model. It consists of both space-like hyperedges---those arising from data-qubit errors---that connect detectors at the same time-step, and time-like hyperedges---those arising from measurement errors---that connect detectors in two adjacent time-steps, see Fig.~\ref{fig:hyperedges}. At each circuit location $(t,j)$, the space-like hyperedges form an identical subgraph consisting of a disconnected $X$- and $Z$-component corresponding to the $X$- and $Z$-detectors. Some space-like hyperedges have weight one, i.e.~they are connected to only one vertex, and these we say are \textit{connected to the boundary}. When decoding with a minimum-weight decoder, one would add a boundary vertex $v_{\text{bdy}}$ to the graph and connect each weight-one edge to it, but we do not consider such a modification here.

\subsubsection{Time-like hyperedges and fragile time-boundaries}
\label{sec:frag-time}

The time-like hyperedges may touch one, two \textit{or} three vertices, as shown in Fig.~\ref{fig:hyperedges}. In particular, weight-one time-like edges correspond to measurement errors in $Z$- (resp.~$X$-) detectors immediately following a $\ket{+}^{\otimes n}$ (resp. $\ket{0}^{\otimes n}$) reset, or immediately preceding an $X$- (resp.~$Z$-) measurement, and they will also later appear when we consider windowed decoding in \cref{sec:window}. In general, we refer to these weight-one edges collectively as constituting an \emph{open} time-boundary (as a single defect can be matched to it), as opposed to a closed time-boundary to which a single defect cannot be matched. More specifically, if the open time-boundary arises from a reset or measurement (instead of from windowed decoding), we refer to it as a \textit{fragile} time-boundary. Open time boundaries are potentially dangerous for decoding since only one type of stabilizer has a fixed eigenvalue. Moreover, the observing edge set of any fragile observable, see Definition \ref{def:fo}, ``borders'' at least one fragile time boundary, in the sense that at least one hyperedge in the observing edge set shares a vertex with at least one edge in the fragile time boundary. How to avoid decoding fragile observables will be discussed in Section \ref{sec:fragile}.

A weight-three time-like edge corresponds to a measurement error on an ancillary qubit immediately following a logical $\overline{S}$ or $\overline{\mathrm{CNOT}}$ gate. For example, a measurement error on an ancilla qubit measuring a $Z$-check immediately following an $\overline{S}$ gate triggers three detectors: two $Z$-detectors before and after the measurement \textit{and} an $X$-detector before the measurement whose detecting region includes the ancilla measurement by the definition of the pre-gate frame, see Tab.~I from SM~\cite{supp} and the gray hyperedge in Fig.~\ref{fig:bad_decomposition_plain_MWPM}(b).

\begin{figure*}[tb]
    \centering
    \includegraphics[width=0.9\textwidth]{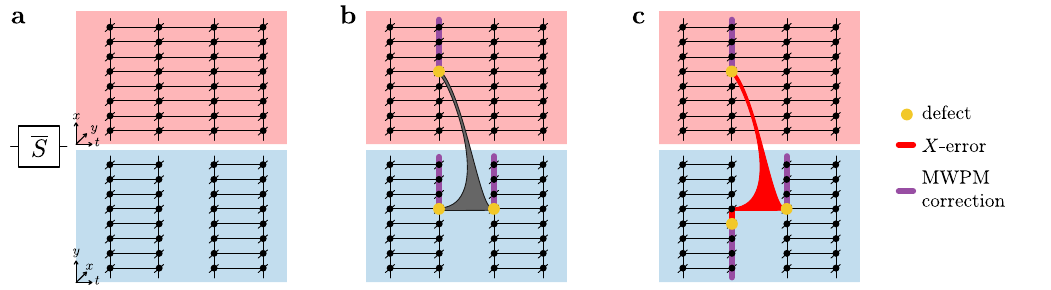}
    \caption{Hyperedge decomposition and error combination that explains the bad performance of a ``splitting-hyperedge'' matching decoder in the pre-gate frame. (a) $\mathcal{G}_{\text{matchable}}$ of the $\overline{S}$ gate, which is built by only taking the edges from the decoding hypergraph in Fig.~\ref{fig:hyperedges}. (b) The procedure from Ref.~\cite{delfosse2023} to decompose the gray hyperedge goes as follows: (1) the detectors involved in the hyperedge are triggered (yellow dots), (2) such syndrome is decoded with \textsc{mwpm} using $\mathcal{G}_{\text{matchable}}$, and (3) the correction (purple edges) corresponds to the hyperedge decomposition. Note that the number of purple edges in the decomposition can badly grow with the code distance. (c) The two errors marked in red lead to a logical error when decoded using $\mathcal{G}_{\text{matchable}}$, thus the logical scaling is at most $O(p^2)$ for any distance.}
    \label{fig:bad_decomposition_plain_MWPM}
\end{figure*}

Clearly, in the absence of the weight-3 time-like hyperedges, $\mathcal{G}$ would be matchable and thus decodable using \textsc{mwpm}. A simple approach to decoding $\mathcal{G}$ is therefore to follow the ``splitting-hyperedge'' methods of Ref.~\cite{delfosse2023} and Stim~\cite{gidney2021stim}; that is, to perform decoding only using the subgraph $\mathcal{G}_{\text{matchable}}$ of $\mathcal{G}$ that contains weight-2 edges such that $\mathcal{G}_{\text{matchable}}$ is matchable. We could then decompose each weight-3 hyperedge into a set of edges in $\mathcal{G}_{\text{matchable}}$ whose endpoints match the endpoints of the hyperedge, and use this decomposition to update the weights of the edges in $\mathcal{G}_{\text{matchable}}$. 
For example, this method can be applied to $Y$ data-qubit errors which produce four defects in standard circuit-level noise decoding for the surface code.

However, this approach is not fault-tolerant for our circuits: there are weight-2 error patterns that lead to a logical error, as shown in Fig.~\ref{fig:bad_decomposition_plain_MWPM}(c). The key issue is that the weight-3 time-like hyperedge cannot be decomposed into a set of weight-2 time-like edges, and instead must use a set of space-like edges. If the hyperedge is located $\sim d/2$ away from each of the spatial boundaries, the number of space-like errors that must be used to decompose the hyperedge is $\sim 3d/2$, see Fig.~\ref{fig:bad_decomposition_plain_MWPM}(b). Consequently, with just the weight-2 error pattern shown in Fig.~\ref{fig:bad_decomposition_plain_MWPM}(c), the matching decoder on $\mathcal{G}_{\text{matchable}}$ will make a logical error. We demonstrate this numerically in Section~\ref{sec:fail_split}. We therefore require a different approach to decode $\mathcal{G}$ using \textsc{mwpm}, which we present in the following section.

\section{Logical observable matching (\textsc{lom}) decoding} \label{sec:logical_observable_decoder}

In this section we design a matching-based decoder for the hypergraph $\mathcal{G}$. 
The key observation enabling our decoder is that for any single observable $O$, it is possible to define a matchable subgraph that ``follows'' the \textit{observing region} $O^{\leftarrow}$ of the observable. Therefore, by performing matching on this subgraph of the decoding graph, we are able to determine if a logical error occurred that flips the observable $O$. We call this decoder the \textit{Single Logical Observable Matching} decoder, or single-\textsc{lom} for short. In \cref{sec:windowless}, we turn the single-\textsc{lom} decoder into the \textsc{lom} decoder: a windowless matching decoder for the full circuit $\mathcal{C}$ that works by running a \textit{separate} iteration of the single-\textsc{lom} decoder for each of the independent reliable observables in the circuit. 



\subsection{Single logical observable matching decoding}\label{sec:slom}

\begin{figure*}
    \includegraphics[width=0.9\linewidth]{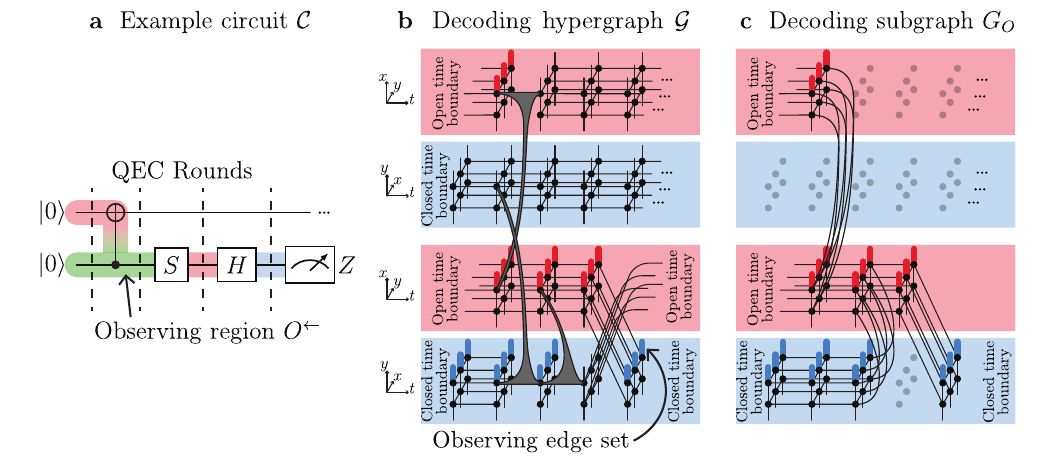}
    \caption{The decoding hypergraph and subgraph $G_O$ for an observable $O$ for an example circuit \footnote{Actually, the $Z$ measurement observable which is decoded here happens to be fragile, and decoding fragile observables can be avoided for reasons discussed in Section \ref{sec:fragile}.}. (a) The observing region $O^{\leftarrow}$ of the observable $O$ being the $Z$-measurement is shown, where red, green and blue shading indicates $X$, $Y$ and $Z$ support respectively. (b) The decoding hypergraph $\mathcal{G}$ of the encoded circuit $\overline{\mathcal{C}}$ for the $d=3$ unrotated surface code. The vertices of $\mathcal{G}$ are partitioned based on the logical qubit to which they pertain and whether they represent $X$- (red background) or $Z$- (blue background) detectors. Note that the $x$ and $y$ spatial coordinates have been swapped between the $X$- and $Z$-detectors for ease of visualization. \added{Also note that we have only shown the weight-3 hyperedges that appear on the top-front qubit to reduce clutter; weight-3 hyperedges with the same pattern also exist at all the other spatial coordinates in $\mathcal{G}$ but are not shown here.} We have also labeled all the edges that make up the observing edge set $\mathcal{H}_{O}$ (in blue and red), and the open and closed time boundaries that arise from resets and measurements in (a). (c) The decoding subgraph $G_{O}$ for the observable $O$. As explained in the text, $G_{O}$ contains all the vertices $V_{O}$ of $\mathcal{G}$ that pertain to the same logical qubit and represent a detector of the same Pauli type ($X$ or $Z$) as an edge in the observing edge set. It clearly has the property that it is a graph with no hyperedges. Vertices that are excluded from $V_{O}$ are shaded in gray.}\label{fig:decoding_subgraphs}
\end{figure*}


Given an observable $O$, we define the
subgraph $G_O=(V_O,E_O)$ of $\mathcal{G}$ as follows, see the example circuit in Fig.~\ref{fig:decoding_subgraphs}. To define the vertex subset $V_{O}\subseteq\mathcal{V}$, we take the observing edge set $\mathcal{H}_{O}$ of $O$, which consists of a set of space-like edges on the spatial boundary of the surface code. Each edge $e\in\mathcal{H}_{O}$ is connected to the boundary, with its only endpoint a vertex that occurs at some logical circuit location $(t,j)$ and corresponds to an $X$- or a $Z$-detector. To construct the vertex subset $V_{O}$, we include all vertices that occur at the same logical circuit location $(t,j)$ and correspond to a detector of the same Pauli type as an edge in $\mathcal{H}_{O}$, as shown in \cref{fig:decoding_subgraphs}. Given the vertex subset $V_{O}$ we now define the edge subset $E_{O}=\{h\cap V_{O}:h\in\mathcal{H},h\cap V_{O}\neq \emptyset\}$, which is intuitively the ``projection'' of the hyperedges in $\mathcal{H}$ onto the vertex subset $V_{O}$, \added{and $\mathcal{H}_0\subseteq E_0$}. It is a straightforward exercise to verify that $E_O$ is indeed a set of edges for all possible choices of observables for the $\{\overline{H},\overline{S},\overline{\mathrm{CNOT}}\}$ gates, by visual inspection using \cref{fig:decoding_subgraphs}. A general mathematical argument can be found in Section~III in the SM~\cite{supp}.

The single-\textsc{lom} decoder then proceeds by applying \textsc{mwpm} to the observed defects in $V_{O}$ using the edges in $E_{O}$. If the returned minimum-weight solution contains an odd number of edges in the observing edge set $\mathcal{H}_{O}$, then the single-\textsc{lom} decoder predicts that the observable has been flipped by the error and must be corrected; otherwise, it predicts that the observable has not been flipped by the error. Given the definition of $\mathcal{G}_{O}$, it is clear that defects in the entire observing edge set are contained in $E_{O}$ and therefore the single-\textsc{lom} decoder is looking at all relevant parts of the circuit that could cause a flip in $O$. Note that one could also apply other graph-based decoders, such as Union Find~\cite{delfosse2021almost}, to decode the observed defects in $V_{O}$. 

Finally, we discuss the fault-tolerance of this single-\textsc{lom} decoder applied to any reliable observable, which is summarized as follows.
\added{
\begin{theorem}[Fault-tolerance of single-\textsc{lom}]\label{thm:single-LOM_FT}
    Given a reliable observable $O$, the single-\textsc{lom} decoder applied to the decoding subgraph $G_{O}=(V_{O},E_{O})$ can correctly predict whether any basic error of weight $w<d/2$ has flipped the observable $O$.
\end{theorem}
\begin{proof}
    The proof is straightforward and follows familiar arguments for the surface code, but for completeness we repeat them here.
    
    The errors that occur on the hypergraph can be projected onto the decoding subgraph, let us write $\vec{e}\subseteq E_{O}$ as the set of edges in $G_{O}$ that have been flipped. The defects that are visible to the decoder are then $\partial\vec{e}\subseteq V_{O}$. The single-\textsc{lom} decoder then works by applying the minimum weight correction $\vec{c}\subseteq E_{O}$ to the defects such that $\partial\vec{c}=\partial\vec{e}$. The total error-plus-correction string $\vec{e}\oplus\vec{c}$ does not trigger any detectors, and therefore is either a logical error or logically trivial, depending on whether $\vec{e}\oplus\vec{c}$ intersects the observing edge set $\mathcal{H}_{O}$ an odd or an even number of times (respectively).
    
    We will now show that if the weight $w(\vec{e})$ of $\vec{e}$ is less than $d/2$, then $\vec{e}\oplus\vec{c}$ intersects $\mathcal{H}_{O}$ an even number of times. First, note that because $\vec{c}$ is a minimum-weight correction to $\partial\vec{e}$, then we have $w(\vec{e}\oplus\vec{c})\leq w(\vec{e})+w(\vec{c})< d$.
    Because the observable is reliable, all the time boundaries in $G_{O}$ are closed and the only weight-one edges in $G_{O}$ are the spatial boundaries on the surface code, with the observing edge set $\mathcal{H}_{O}$ containing all the weight-one edges on one of these boundaries. If $\vec{e}\oplus\vec{c}$ were to intersect $\mathcal{H}_{O}$ an odd number of times, it therefore must traverse from the ``top'' to the ``bottom'' spatial boundary. However, the minimum-weight path from the ``top'' spatial boundary to the ``bottom'' has weight $d$; this is even the case in the presence of time-like edges arising from weight-3 hyperedges because these edges connect vertices with the same spatial coordinates. Therefore, any possible string $\vec{e}\oplus\vec{c}$ that intersects $\mathcal{H}_{O}$ an odd number of times must use at least $d$ edges. This is a contradiction since $w(\vec{e}\oplus\vec{c})\geq d$, and hence $\vec{e}\oplus\vec{c}$ is logically trivial, and the proof is complete.
\end{proof}

Interestingly, there exist error patterns of weight at least $d$ that are undetectable errors in the single-\textsc{lom}, but are detectable in the full decoding hypergraph and therefore may be correctable by a minimum weight hypergraph decoder. For example, if the error string uses an edge that corresponds to a hyperedge in the full decoding hypergraph, the string may be detectable in the full hypergraph but not in the single-\textsc{lom}.
}


\subsection{Multiple logical observable matching decoding}
\label{sec:windowless}

In this section, we explain how to turn the single-\textsc{lom} decoder into a decoder for an entire quantum circuit $\mathcal{C}$ which we call the \textsc{lom} decoder. 
We start by discussing how to decode \textit{reliable} circuits---that is, circuits in which all observables are reliable---before arguing that fragile observables do not need to be decoded in Section \ref{sec:fragile}.

\subsubsection{Decoding reliable circuits}\label{sec:reliable}

In the case where all observables in the circuit are reliable, the \textsc{lom} decoder is simple: decode every logical measurement individually---both conditioning as well as final measurements---using an independent instance of the single-\textsc{lom} decoder. Since we only decode reliable observables, each instance of the single-\textsc{lom} decoder can correct every basic error of weight $<d/2$. Therefore, the \textsc{lom} decoder also has this property for the full reliable circuit $\mathcal{C}$. 

It is worth discussing some properties of the \textsc{lom} decoder. First, the \textsc{lom} decoder does \textit{not} pass any information between each instance of the single-\textsc{lom} decoder. Thus, measurements that occur at the same time can be decoded in parallel. A disadvantage however is that the same location in the circuit is typically decoded multiple times by independent instances of the single-\textsc{lom} decoder, increasing the total space-time cost of the decoder in classical hardware. Another important property of the \textsc{lom} decoder is that it does \textit{not} explicitly provide a physical correction, i.e.~a set of errors consistent with the syndrome, as corrections returned by each single-\textsc{lom} decoder may not be consistent with each other. This independence also poses issues for decoding near fragile time-boundaries as we will see when handling fragile observables in \cref{sec:fragile}, and when developing a windowed decoder in Section \ref{sec:window}.

\begin{figure*}
\centering
    \includegraphics[width=0.9\textwidth]{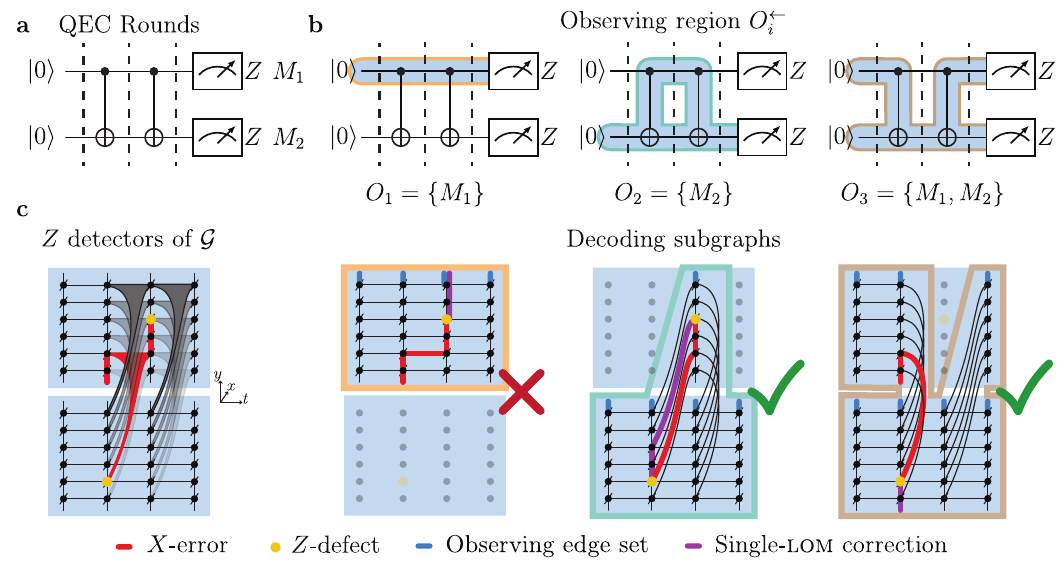}
    \caption{An example of a weight-5 error in the $d=7$ surface code in which the decoded outcomes for the reliable observables $O_{1}$, $O_{2}$ and $O_{3}=O_{1}O_{2}$ are inconsistent. (a) The circuit consists of two consecutive CNOT gates, with both qubits initialized and measured in the $Z$-basis. The logical action of the circuit is trivial here and only serves as a simple example. (b) The observing regions of the three observables. (c) The $Z$-detectors in the decoding hypergraph $\mathcal{G}$. For simplicity, only a slice of $\mathcal{G}$ is shown for a fixed value of the spatial coordinate $x$. When each of the observables is decoded by the single-\textsc{lom} decoder, it predicts a flip in the observable $O_{1}$ but not in $O_{2}$ or $O_{3}$, despite the fact that $O_{3}=O_{1}O_{2}$. This is only possible because the weight of the error---five---is greater than $d/2=7/2$. Note that in this figure and throughout, we will refer to measurement errors of $Z$-stabilizers as $X$-errors (and vice versa for $X$-stabilizer errors) for simplicity.}\label{fig:inconsistent_SLOM}
\end{figure*}

One can easily generalize the preceding definition of the $\textsc{lom}$ decoder, and run single-\textsc{lom} decoder instances over any independent generating set of reliable observables---that is, any minimal set of reliable observables that generates all possible \added{reliable} observables in the circuit. The independence here is desirable since the single-\textsc{lom} decoder may not return consistent results for observables that are dependent on each other \added{if there are more than $d/2$ errors in the $\textsc{lom}$ decoding volume}. For example, in \cref{fig:inconsistent_SLOM}, the single-\textsc{lom} decoder predicts a flip to $O_{1}$ but not to $O_{2}$ or $O_{3}$, despite the fact that $O_{3}=O_{1}O_{2}$. Hence, if in the final measurements of some algorithm, the final outcome is some \textsc{xor} of final measurements, it is advisable to include this as an observable in the generating set to minimize its logical error rate, see the numerical evidence for this choice of decoding in Section \ref{sec:pt}.
This freedom to decode a circuit using any set of (independent) reliable observables is also relevant for handling fragile observables as discussed in the next section.

Before going into this, we note that more realistic circuit-level noise will introduce new hyperedges in the decoding subgraph $G_{O}$ for the observable $O$ and can make $\mathcal{H}_O$ an observing hyperedge set. However, such hyperedges are decomposable via splitting into a \textit{constant} number of edges that are already present in the decoding subgraph $G_{O}$, see \cref{sec:circuit_level_slom}, \added{or decoding $G_O$ can be done with other hypergraph decoding techniques, see Section \ref{sec:discuss}}.


\subsubsection{Handling fragile observables}\label{sec:fragile}

Let's first understand what the problem is when we decode fragile observables using the single-\textsc{lom} decoder, before we argue that for any circuit we only need to decode an independent generating set of reliable observables.
The issue arises due to a combination of two factors. The first of these is a transversal logical reset providing a fragile (open) time-boundary (defined in \cref{sec:frag-time}) that is followed by a non-trivial logical gate after only $O(1)$ QEC rounds. Second, we decode the corresponding fragile time-boundary multiple times by \textit{independent} instances of the single $\textsc{lom}$ decoder. If either of these factors are removed, no issues arise for fragile observables. The issue is closely related to those explored in other recent works~\cite{algo-FT,zhang2025}. 


\begin{figure}
    \includegraphics[width=0.99\linewidth]{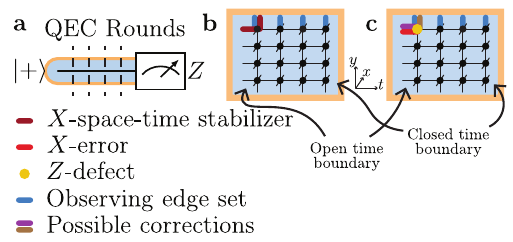}
    \caption{A simple example of a fragile observable. (a) A $\ket{+}$ reset followed by a $Z$-measurement, with four QEC rounds in between for ease of visualization. (b) The $Z$-detectors of the decoding hypergraph $\mathcal{G}$, with an $X$-space-time stabilizer and the observing edge set of the observable containing the $Z$-measurement. Because the observable is fragile, the observing edge set has odd overlap with the edges contained in the $X$-space-time stabilizer. (c) There are two equal-weight corrections to the weight-one measurement error shown, one of which flips the observable while the other one does not. In this simple circuit, it does not matter which decision the decoder makes, because the observable has 50-50 measurement statistics anyway.}\label{fig:fragile_observable}
\end{figure}

Recall that a fragile observable $O$ is such that $O^{\leftarrow}$ anticommutes with a reset stabilizer $S$, see Definition~\ref{def:fo} and Lemma~\ref{lem:propa}. Because $O^{\leftarrow}$ anticommutes with $S$, this also means that the corresponding logical observing region $\overline{O}^{\leftarrow}$ anticommutes with a weight-2 space-time stabilizer at the start of the circuit, at the fragile time-boundary, as shown in \cref{fig:fragile_observable}(a, b)\footnote{A different definition of the observing edge set, for example, one that includes space-like edges in the bulk of the surface code rather than just on the spatial boundary, can increase the minimum weight of the space-time stabilizer (say to three), but it will always be of constant weight.}. This leads to a somewhat paradoxical statement: applying the weight-2 space-time stabilizer causes an undetectable flip in the decoded outcome of the fragile logical observable $\overline{O}$. Since it is a space-time stabilizer, applying it does not change the state at that time-step in the circuit, and therefore it cannot change the measurement statistics of any observable. As a result this operator does not constitute a \textit{logical} fault. It does, however, still flip the decoded outcome of $\overline{O}$---this is possible since the statistics of $O$ are 50-50. To illustrate this point, suppose that a weight-one error occurred on one of the two edges contained in the space-time stabilizer, as shown in \cref{fig:fragile_observable}(c). Then, a decoder has two equally-likely corrections it can make, one of which flips the observable and the other of which does not.

At this point, such an ambiguity is not a problem, since the decision the decoder makes does not change the measurement statistics of the observable. A problem does arise, however, when the same fragile time boundary is decoded multiple times by independent instances of a decoder, and the decoder makes a \textit{different} correction in each of these iterations, while the outcomes of the logical observables are supposed to be correlated. 
Care must therefore be taken to ensure that the decoder decodes all fragile observables in a consistent way.

\begin{figure*}
    \includegraphics[width=0.9\linewidth]{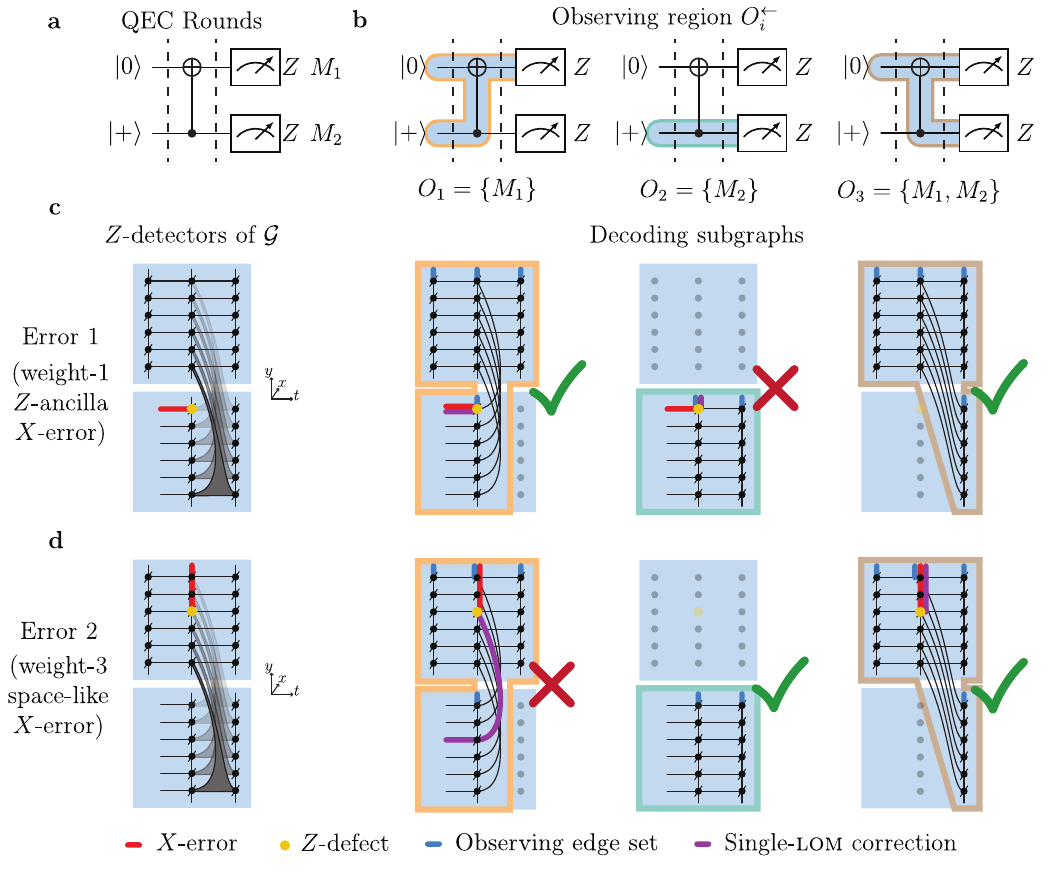}
    \caption{Examples in the $d=7$ surface code of the single-\textsc{lom} decoder applied to potentially hazardous error configurations around a fragile time-boundary. In the bare circuit there are three observables that single-\textsc{lom} decoder can be ran on: $O_{1}$ and $O_{2}$ are fragile, while $O_{3}$ is not fragile and moreover has deterministic $+1$ statistics. 
    We show two examples of errors in the $d=7$ surface code that do \textit{not} lead to a logical error in $O_{3}$ when $O_{3}$ is decoded directly by the single-\textsc{lom} decoder, but \textit{can} lead to a logical error when $O_{1}$ and $O_{2}$ are decoded separately using single-\textsc{lom} decoders to infer $O_{3}$. 
    Error 1 is a weight-1 measurement error that causes an error in $O_{3}$ if the decoders for $O_{1}$ and $O_{2}$ make different choices about which correction to make---which will occur 50\% of the time if minimum weight corrections are chosen with uniform probability. Meanwhile Error 2 is a weight-3 space-like $X$-error that will deterministically cause a decoding error in $O_{1}$ and $O_{2}$, this logical error is even constant weight as the distance of the code increases.}\label{fig:Bell_state_fragile_observables}
\end{figure*}

To illustrate the problem, consider a pair of fragile observables $O_{1}$ and $O_{2}$ whose observing regions both border the same fragile time boundary but whose product is \emph{not} fragile. A simple example is given by the circuit
\begin{equation}\label{eq:Bell_state_fragile_observables}
\vcenter{\hbox{\includegraphics{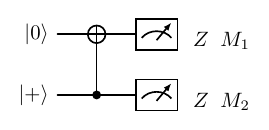}}}
\end{equation}
in which the observables $O_{1}=\{M_{1}\}$ and $O_{2}=\{M_{2}\}$ are both fragile, but their product $O_{3}=\{M_{1},M_{2}\}$ is not. Moreover, in this circuit, $O_{3}$ has a deterministic $+1$ measurement outcome. 
Now, suppose that we use two instances of single-\textsc{lom} decoders to decode $O_{1}$ and $O_{2}$, and take their product to infer the outcome of $O_{3}$. Because $O_{1}$ and $O_{2}$ are fragile, a low-weight error ($<d/2$) can cause them to flip. For example, the decoder may correctly infer $O_{1}$ but incorrectly infer $O_{2}$. This then leads to a flip---and this time a \textit{genuine} logical error---to the product operator $O_{3}=O_{1}O_{2}$, even though $O_{3}$ is not fragile. In \cref{fig:Bell_state_fragile_observables}, we show two examples of error patterns that can lead to a logical error in $O_{3}$ if $O_{1}$ and $O_{2}$ are decoded using the single-\textsc{lom} decoder. We can get around this issue by running the \textsc{lom} decoder as follows. It first randomly selects an outcome for $O_{1}$, then decodes the reliable observable $O_{3}$ using the single-\textsc{lom} decoder, then uses the product $O_{2}=O_{1}O_{3}$ to infer the value of $O_{2}$. This is possible here because both $O_{1}$ and $O_{2}$ are 50-50 random, but correlated through the outcome of $O_{3}$.

In the circuit in \cref{eq:Bell_state_fragile_observables} the solution to the problem is fairly straight-forward because the pair of fragile observables was measured at the same time; however this need not be the case. For example, the first fragile observable $O_{1}$ could condition a Clifford gate that needs to be committed to before continuing the circuit, as shown in the circuit
\begin{equation}\label{eq:2_T_gate_circuit}
    \vcenter{\hbox{\includegraphics{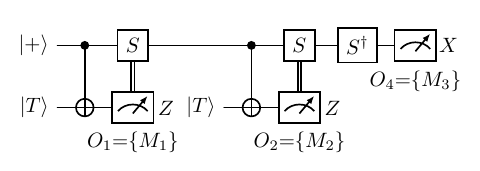}}}
    .
\end{equation}
In this simple example, the observables $O_{1}$ and $O_{2}$ are fragile, while their product $O_{3}=\{M_{1},M_{2}\}$ is reliable (note that the outcome of $O_{3}$ is not deterministic due to the presence of $\ket{T}$-state injections). Note also that the circuit in \cref{eq:2_T_gate_circuit} simply amounts to preparing $S^{\dag}TT\ket{+}=\ket{+}$ before measuring in the $X$-basis, so that the final observable $O_{4}$ is deterministically $+1$. \added{Note that $O_{4}$ may or may not be fragile depending on how many of the conditional $S$-gates have been applied.}

\cref{eq:2_T_gate_circuit} causes an issue both for the \textsc{lom} and a minimum-weight decoder, because the decoder needs to be run, first to decide on an outcome for the fragile observable $O_{1}$, and then a second time to decide on an outcome for the fragile observable $O_{2}$. If the two instances of the decoder give different corrections in the vicinity of the fragile time-boundary---for example, due to a weight-1 error string similar to the one in \cref{fig:bad_decomposition_plain_MWPM}(c)---then this will cause a logical error in the product observable $O_{3}$. This is inherently a problem because a logical error to $O_{3}$ means that an error has occurred somewhere in the observing region $O_{3}^{\leftarrow}$, and every possible location has a non-trivial effect on the state of the system. Indeed, in \cref{eq:2_T_gate_circuit}, an error to $O_{3}$ is consistent with an $X$-error after the first $T$-injection, meaning that the final state is $S^{\dag}TXT\ket{+}=S^{\dag}\ket{+}$, which has 50-50 measurement statistics when measured in the $X$-basis.

There are multiple solutions to this problem. Perhaps the simplest is to wait $\Theta(d)$ rounds after the transversal logical reset, allowing the state to reach the ``quiescent'' state~\cite{Fowler_2012}, and perform an initial window of decoding that commits to a correction in the first $\Theta(d)$ rounds. This works because all of the subsequent instances of the decoder will use the same correction in the vicinity of the fragile time-boundary. In the case of the minimum-weight decoder, Ref.~\cite{algo-FT} proposes a more efficient solution: an additional ``consistency check'' in the decoder, in which the second instance of the decoder checks whether its correction is consistent with the first and makes an adjustment if necessary before continuing. Related issues around $T$-injections in the context of a windowed decoder can also be solved by using an additional ancilla per $T$-injection as proposed in Ref.~\cite{zhang2025}. Here, we provide yet another solution that is tailored to the \textsc{lom} decoder.


To decode the circuit in \cref{eq:2_T_gate_circuit} using the \textsc{lom} decoder, we use essentially the same procedure as for \cref{eq:Bell_state_fragile_observables}. We begin by applying the circuit up to and including the measurement $M_{1}$. Conventional wisdom would now say that we need to decode the observable $O_{1}$ and commit to applying the conditional $S$ gate or not. However, because $O_{1}$ is fragile, it can be flipped by a space-time stabilizer without causing a logical error. There is therefore no need to apply a decoder to $O_{1}$, and we can uniformly at random select an outcome for $O_{1}$ without decoding it, regardless of what the ``true'' outcome of the measurement was. Then, we continue executing the circuit up to and including the second measurement $M_{2}$. The observable $O_{2}=\{M_{2}\}$ is also fragile and has a 50-50 measurement outcome. However, the product $O_{3}=O_{1}O_{2}$ is not fragile and must be decoded. We therefore decode the observable $O_{3}=\{M_{1},M_{2}\}$, which is ``peeled'' off the fragile time-boundary, and infer the value of $O_{2}$ from the product $O_{2}=O_{1}O_{3}$. 
We can then finish running the circuit, but we now have to interpret the final observable $O_{4}$. As mentioned, $O_{4}$ may or may not be fragile depending on how many of the conditional $S$-gates have been applied. If it is not fragile, it can be decoded directly; if it is fragile, we first decode the product observable $O_{1} O_{4}=\{M_{1},M_{3}\}$, and then multiply its outcome with the outcome of $O_{1}$ to obtain the outcome of $O_{4}$. In all cases, only a basic error of weight $>d/2$ can cause a flip in the outcome of $O_{4}$.

We can generalise the procedure as follows.
Each time a measurement $M$ occurs in the circuit, we first check whether the observable $O=\{M\}$ is fragile or not. If $O$ is reliable, then we simply decode it using the single-\textsc{lom} decoder and continue. If on the other hand $O$ is fragile, we check whether there exists a product with another, earlier fragile observable $O'$ \added{(which can itself be a product of individually measured observables)} such that $OO'$ is reliable. If no such $O'$ exists, we uniformly at random select an outcome for $O$ and continue. If such an $O'$ does exist, then we decode $OO'$ and use the outcome of $O'$ to infer the outcome of $O$. Continuing this way through the circuit allows us to reproduce the measurement statistics of the original circuit fault-tolerantly, without ever decoding a fragile observable. Note that there is no guarantee on the minimum weight of such an operator $O'$, which we will find later to cause problems with the efficiency of the windowed-\textsc{lom} decoder in \cref{sec:window_efficiency}.

In many cases, we do not need to \textit{actually} randomly sample the measurement outcomes of the fragile measurement. Take, for example, the fragile observable $O_{1}$ in \cref{eq:2_T_gate_circuit}. We are not interested directly in the measurement statistics of $O_{1}$, and only use it to decide whether the conditional $S$ gate is applied or not. However, the $T$-injection gadget is \textit{deterministic} after the conditional $S$ gate is applied: it will always apply a $T$ gate to the input state. Therefore, it doesn't matter whether we choose the outcome to be $+1$ or $-1$, and we can therefore choose it to be $+1$ without randomly sampling the result. Note that such a trick only works for the first few $T$-injections in a circuit whose corresponding observables are fragile.

\added{
What we have described above is only one way of explaining why fragile observables do not need to be decoded. Here we briefly outline an alternative, entirely equivalent way of understanding this procedure. It turns out that for a fragile observable $O$ that is modeled using a coin toss above, any $S$ gates that are conditioned on $O$ can actually be replaced by a Pauli gate conditioned on $O$. A simple example of this is the $S$ gate conditioned on $O_{1}$ in \cref{eq:2_T_gate_circuit}: because the $T$ gate is being applied to the $\ket{+}$ state, the conditional $S$ gate can be replaced by a conditional $X$ gate. One can generalize this argument to general circuits and show that this replacement can be done for precisely the same observables that we modeled with a coin toss above. One can then track this conditional $X$ gate in software through the rest of the algorithm, possibly updating future observables to include $M_{1}$ when they would be flipped by the conditional gate (for example, $O_{2}$ in \cref{eq:2_T_gate_circuit}). This last process turns out to be exactly equivalent to the procedure described above of multiplying fragile observables together to make new reliable observables.
}

\subsection{Issues with hierarchical matching decoding}
\label{sec:hier}

In contrast to the \textsc{lom} decoder introduced above, recent work~\cite{sahay2024, guernut2024, chen2024transversal} has used a different strategy to decode transversal Clifford gates in the surface code. We refer to this decoder as the \textit{hierarchical} matching decoder (it has also been called a sequential or ordered matching decoder). In this section, we explain why such an approach is \textit{not} fault-tolerant in general, hence justifying our design of the \textsc{lom} decoder. It is worth noting that we do not rule out the possibility that these issues with a hierarchical matching decoder could be overcome; the goal here is simply to point out what issues need to be solved. 

We describe the hierarchical matching decoder in the simplified setting of non-conditional Clifford circuit under the basic error model, as this is sufficient to break its fault-tolerance. The hierarchical matching decoder works by splitting the decoding graph $\mathcal{G}$ into a number of \textit{track} subgraphs $T_{\ell}=(V_{\ell},E_{\ell})$ for $\ell=1,2,\dots,L$, each of which is a graph and is hence matchable. The idea is to decode each track $T_{\ell}$ sequentially, starting at $\ell=1$, and committing to a subset of hyperedges in $\mathcal{G}$ at each step. 
More precisely, we define each track by specifying a vertex subset $V_{\ell}\subseteq\mathcal{V}$. This vertex subset must satisfy the following criterion:
\begin{criterion}[Hierarchical Track Vertex Assignment]\label{crit:hier}\,
    \begin{enumerate}[label=(\alph*)]
        \item Every vertex $v\in\mathcal{V}$ is contained in exactly one track vertex subset $V_{\ell}$,
        \item For every edge $e\in\mathcal{H}$ with weight $|e|=2$, both endpoints $v,v'$ of $e$ are contained in the same track $v,v'\in T_{\ell}$, 
        \item For every hyperedge $h\in\mathcal{H}$ with weight $|h|=3$, there exists $\ell,\ell'$ such that $\ell<\ell'$, $|h\cap V_{\ell}|=2$, and $|h\cap V_{\ell'}|=1$.
    \end{enumerate} 
\end{criterion}
With \cref{crit:hier}, the tracks are defined as
\begin{definition}[Hierarchical Tracks]\label{def:hier}
    For a set of vertex subsets $\{V_{\ell}\}_{\ell=1,\dots,L}$ satisfying \cref{crit:hier}, the track subgraph $T_{\ell}$ is given by $T_{\ell}=(V_{\ell},E_{\ell})$ where the track edges are given by
    \begin{equation}
        E_{\ell}=\Big\{h\cap V_{\ell}:h\in\mathcal{H}, |h\cap V_{\ell}|\leq2\Big\}.
    \end{equation}
\end{definition}
We also define the hyperedge subset $H_{\ell}$ as the set of hyperedges that correspond to the edges in $E_{\ell}$. 

\begin{figure}
    \includegraphics[width=0.9\linewidth]{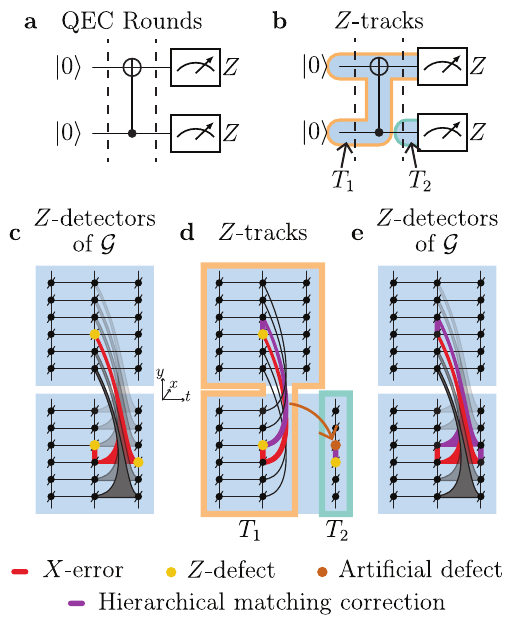}
    \caption{An example of successful correction by a hierarchical matching decoder. (a) The circuit that we decode. For this example, we only consider $X$-errors and the $Z$-detectors of the decoding graph. (b) The circuit locations corresponding to the tracks $T_{1}$ and $T_{2}$. (c) The $Z$-detectors of the full decoding hypergraph $\mathcal{G}$, with an example weight-2 $X$-error pattern. Again, only a slice of the decoding hypergraph is shown for simplicity. (d) The decoding hypergraph is split into the tracks $T_{1}$ and $T_{2}$. After matching is run on $T_{1}$, the decoder makes the correction in purple, which is not exactly the same as the error pattern that occurred. When the correction from $T_{1}$ is committed to in the full hypergraph, this introduces an ``artificial'' defect in $T_{2}$. This is then matched in $T_{2}$ with the left-over original defect from $\mathcal{G}$. (e) The corrections from $T_{1}$ and $T_{2}$ are lifted to a correction in the full hypergraph $\mathcal{G}$. Here we can see that the combination of the error and the correction give a space-time stabilizer, similar to Fig.~1 from SM~\cite{supp}, but this time in the bulk of the decoding graph.}\label{fig:hierarchical_example}
\end{figure}

With a set of tracks defined by \cref{def:hier}, the hierarchical decoder works by building a correction to the full hypergraph $\mathcal{G}$ ``hierarchically''; that is, by decoding each track $T_{\ell}$ sequentially, starting at $\ell=1$, and committing a subset of hyperedges in $\mathcal{G}$ at each step. A simple example is shown in \cref{fig:hierarchical_example}. By the end of the algorithm we will have built a correction $\vec{h}\subseteq\mathcal{H}$ of the full hypergraph. Note that \cref{crit:hier}(c) guarantees that the first track $T_{1}$ is always matchable. Therefore, we perform matching on $T_{1}$, returning a track correction $\vec{e}_{1}\subseteq E_{1}$. We then commit this correction to the full hypergraph $\mathcal{G}$; that is, for every hyperedge $h_{1}\in H_{1}$ we commit to its inclusion in the correction $\vec{h}\subseteq\mathcal{H}$ depending on whether the corresponding edge is in the track correction $\vec{e}_{1}$. Then, we update the remaining defects in the graph, flipping them if they are the endpoint of a hyperedge in the correction $\vec{h}$. By definition, this means that there will be no more defects in the first track, but the update could introduce or remove defects in the later tracks (we call these ``artificial'' defects). We then continue to the second track $T_{2}$. Now, note that all hyperedges $h$ with $|h\cap V_{2}|=1$ must satisfy $|h\cap V_{1}|=2$ and therefore must already be committed to. We can then decode the real and artificial defects in $T_{2}$ using matching on the subgraph with vertices $V_{2}$ and edges $E_{2}$. Repeating this process for all $\ell$ provides a correction to the hypergraph $\mathcal{G}$, from which the logical action can be inferred. 

At this point it is useful to compare the hierarchical matching decoder to the \textsc{lom} decoder defined in Section~\ref{sec:logical_observable_decoder}. Indeed, the two decoders are closely related: the region of the hypergraph $\mathcal{G}$ decoded in each iteration of the single-\textsc{lom} decoder always forms a valid choice for the \textit{first} track $T_{1}$ of a hierarchical matching decoder because the matchable subgraph $G_O$ for a logical operator $O$ is a valid first track choice. The difference is that the hierarchical matching decoder then attempts to use this correction from $T_{1}$ as a starting point for a correction to the full hypergraph $\mathcal{G}$, while the \textsc{lom} decoder does not use the result of $T_{1}$ at all when decoding other regions of the hypergraph. In fact, the reverse is also (almost) true: any valid choice of first track $T_{1}$ corresponds to the matchable subgraph $G_{O}$ corresponding to the backpropagation of some operator $O$, although this operator may not be an actual observable (for example, the operator to be backpropagated could have an $X$ operator immediately preceding $Z$-measurement). This point can be seen by noting that after fixing a choice of which time-like edges immediately following a logical gate $\overline{L}$ are included in $T_{1}$, there is only one valid choice of time-like edges immediately preceding $\overline{L}$ that can be included in $T_{1}$. By inspection of \cref{fig:hyperedges}, this unique valid choice corresponds to the backpropagation of the corresponding logical Pauli operator through $\overline{L}$.

We note that some previous work~\cite{sahay2024,chen2024transversal} only considered situations where the first track only traverses a single logical qubit; this is however not sufficient for decoding arbitrary Clifford circuits. Take, for example, a pair of alternating CNOT gates anywhere in a circuit
\begin{equation}
    \vcenter{\hbox{\includegraphics{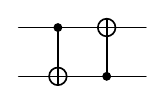}}}
    .
\end{equation}
Here it is not possible to choose a track $T_{1}$ that traverses only a single logical qubit. It is therefore necessary to consider more general choices of tracks \footnote{One can show that it is necessary in some circuits to have more than $L=2$ tracks, and that it is sufficient for arbitrary circuits to have $L=2k$ tracks where $k$ is the number of logical qubits in the circuit, leaving aside issues with decoding near fragile time-boundaries which can be solved. Because the focus here is to show that hierarchical matching does \textit{not} work in general, we do not provide a proof of these statements.}, and this brings in the possibility of the tracks containing loops that we now discuss.


\subsubsection{Time-like loops}\label{sec:hier_loops}

The most serious problem with the hierarchical decoder that we have found concerns the presence of time-like loops. 
Examples of time-like loops are shown in \cref{fig:hierarchical_loop}(a)(i--iii) in an alternating-CNOT experiment. There are three choices of first $Z$-track shown in \cref{fig:hierarchical_loop}(a) corresponding to the three single-\textsc{lom}s one can run in the circuit, all of which contain a time-like loop.

More precisely, we say that there is a loop in a track $T_{\ell}$ if there exists a set of time-like edges in $E_{\ell}$ such that together they do not trigger any defects in $T_{\ell}$. If a non-final track $T_{\ell}$ with $\ell<L$ contains such a loop, then it can produce artificial defects in later tracks that do not appear in pairs, causing the hierarchical decoder to perform high-weight corrections to match the defects to spatial boundaries. In the wrong configuration, this means that constant-weight errors can lead to logical errors in the later tracks when decoded.

\begin{figure*}
   \includegraphics[width=0.9\textwidth]{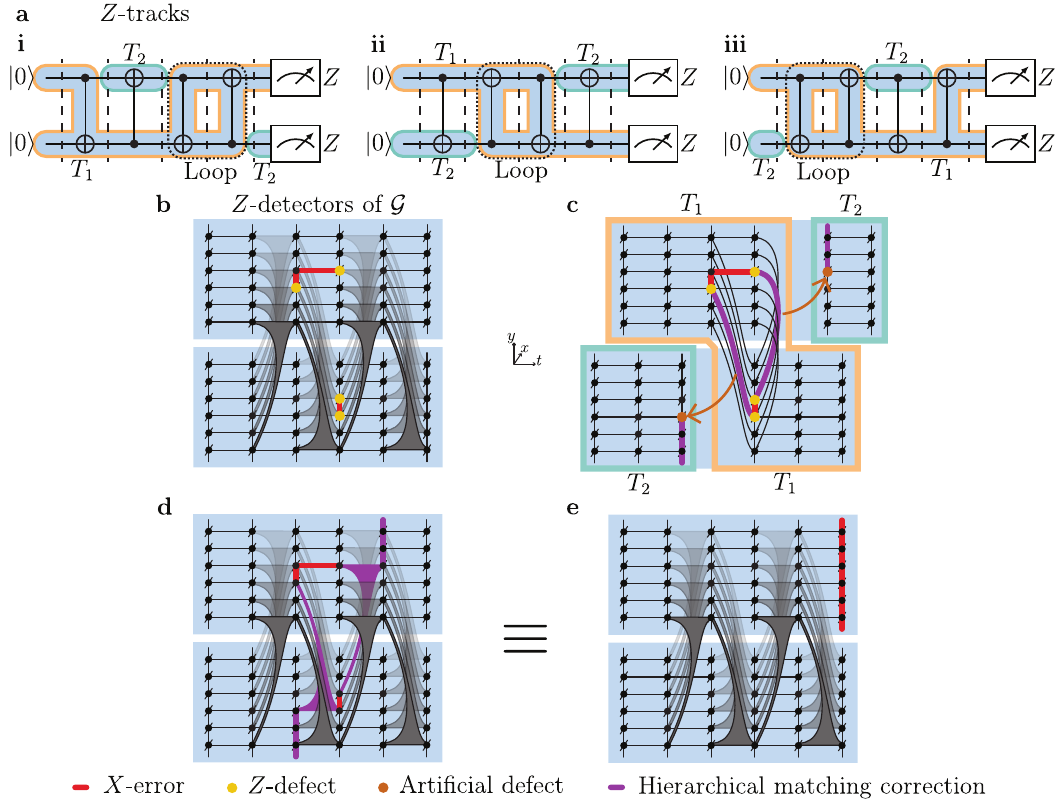}
    \caption{An example of a failure of hierarchical matching due to the presence of a loop in the first track. (a) The circuit consists of four alternating CNOT gates. Looking only at the $Z$-stabilizers, all three possible choices of tracks (i--iii) contain a loop in the first track. 
    (b--e) For the choice of tracks in (a)(ii), a weight-3 physical error in the $d=7$ surface code resulting in a logical $X$-error immediately before the first measurement.}\label{fig:hierarchical_loop}
\end{figure*}

The presence of loops depends both on the circuit and the choice of tracks within the circuit. In the example in \cref{fig:hierarchical_loop}, a weight-3 error in the $d=7$ surface code leads to a logical error, although in this circuit the error is of constant weight and would cause a logical error no matter the distance. 
Similar examples can also be constructed in other non-trivial circuits by, for example, including a $T$-injection between two of the CNOT gates.
In general, the larger the loop, the larger the weight of the error required to cause a logical error. Any solution to using hierarchical matching will have to address the decoding of loops---perhaps by avoiding them using circuit identities or clever track design, or else by decoding them in some other way---in order to be fault-tolerant.

We briefly remark that time-like loops are not the \textit{only} impediment to fault-tolerance with a hierarchical decoder: low-weight errors arising from fragile time-boundaries and time-like snakes similar to what we will discuss later in \cref{sec:window} can also occur. An example of the issue of fragile time-boundaries is shown in \cref{fig:hierarchical_fragile} in Appendix \ref{sec:frag-hier}.



\section{Numerical simulations}
\label{sec:num}

\subsection{Experiments performed}

Here we describe the numerical simulations conducted to characterize the performance of our logical observable decoding strategy on sequences of transversal Clifford gates. We use Stim~\cite{gidney2021stim} to perform the stabilizer simulations and PyMatching~\cite{Higgott_2025} for the \textsc{mwpm} optimization. We have written a wrapper package in Python that builds the encoded circuit for each experiment under the specified noise model, available in~\cite{surface_sim}, and the decoder software can be found in~\cite{lomatching}. 

The physical implementation of the fold-transversal $\overline{H}$ and $\overline{S}$ gates is shown in Fig.~\ref{fig:fold_trans_gates}. For the QEC round of the surface code, we use the parallel schedule of Ref.~\cite{versluis2017scalable} with an ancilla reset after each measurement, see the circuits in Section~V in the SM~\cite{supp}. We have used a longer schedule than the ``standard'' one~\cite{Fowler_2012, Tomita_2014} to use the same primitive gate set---namely CZ and $R_Y(\pi/2)$ gates---as the current superconducting qubit and neutral atom platforms~\cite{wintersperger2023neutral}. The detectors (or checks) correspond to the vertices in the hypergraph and are constructed by tracking the code stabilizers through the fold-transversal operations~\cite{algo-FT}, according to what we call the pre-gate frame, see Section~\ref{sec:noise_detectors_graph} and Section~II in the SM~\cite{supp}.

The logical circuits used to benchmark the decoder are divided into two classes. Both classes correspond to Clifford circuits without conditional gates because Stim is a stabilizer simulator and it does not have out-of-the-box support for conditional gates (except for Pauli gates, which can be tracked in software). The experiments only include reliable observables because all fragile ones are ignored by \textsc{lom} decoding, see Section~\ref{sec:fragile}. 

The circuits in the first class focus on a single type of gate $G \in \{I, S, H, \text{CNOT}\}$ and are depicted in Fig.~\ref{fig:experiment_circuits}(a). 
For each experiment, the considered logical gate is repeated $d+1$ times on an unrotated surface code of distance $d$, and each time it is followed by one QEC round. This class of circuits is denoted as \textit{repeated-gate experiments}. In the case $G=I$, these circuits correspond to memory experiments, as realized in recent experiments~\cite{google2024_belowthreshold,Krinner_2022}. Such similarity allows us to extract a ``threshold'' for each transversal gate and compare it to the idling or memory experiment. 

The input states to the circuits are $\ket{0}^{\otimes n}$ ($Z$-basis experiment) and $\ket{+}^{\otimes n}$ ($X$-basis experiment), with $n=2d^2-2d+1$ for a distance $d$ code, followed by a single QEC round. Since the distance $d$ will be odd, $G^{d+1}=I$, so the ideal final state will be the same as the initial state. We measure the final logical state in the corresponding logical basis, by measuring all data qubits in such basis. Figure~\ref{fig:experiment_circuits}(a) shows the experiments done in the $Z$-basis, although we simulate both $Z$- and $X$-bases. 

\subsubsection{Arbitrary two-qubit Clifford circuits} \label{sec:tq_clifford_exp_description}

The second class of experiments aims to benchmark our decoder on one- and two-qubit circuits built using fold-transversal Clifford gates. The motivation for running this test is to verify that the decoder's performance does not depend on any particular structure of the decoding hypergraph, and that our decoder should be useful for experiments running logical randomized benchmarking on one or two logical qubits~\cite{combes2017, lacroix}. 

Logical randomized benchmarking would use random sequences of Clifford gates $C_m \circ C_{m-1} \ldots \circ C_1$ followed by their inverse $(C_m \circ \ldots \circ C_1)^{-1}$, and compile each $C_i$ into $H,S$ and CNOT. Instead, here we consider a Clifford gate $C$, and run $C^{-1} \circ C$ where we compile $C$ into $H$, $S$ and CNOT, and separately compile $C^{-1}$ into $H$, $S$ and CNOT. 
In principle, one could do this for the whole two-qubit Clifford group, $C\in \mathcal{C}_2$, but since $|\mathcal{C}_2| = 11520$, it would be needlessly computationally expensive. 

Instead, as logical Paulis are tracked in software, we consider the two-qubit Clifford group modulo Paulis, i.e.~$C \in {\rm Sp}(4, \mathbb{F}_2)$~\cite{calderbank1998quantum} which contains 720 elements. In order to further save time, we remove unnecessary SWAP gates from $C$, as follows. We iterate through all $C \in {\rm Sp}(4, \mathbb{F}_2)$ to create a list, and let Stim compile each gate $C$ into $H$, $S$ and CNOT using Gaussian elimination~\cite{gidney2021stim} \footnote{Note that we do not use the decomposition from Ref.~\cite{grier2022classification} as the hypergraphs would have a particular structure across all circuits and thus this would ruin the purpose of the two-qubit Clifford experiments.}. Then, if some $C'=\mathrm{SWAP} \cdot C$ where $C$ is already on the list, we remove $C'$ from the list, only keeping 360 Clifford circuits in total. For each $C$ on the shortened list, we compute $C^{-1}$ and compile this again into $H$, $S$, and CNOT. The maximum CNOT-$H$-$S$ depth of the compilation of $C^{-1} \circ C$ is $\ell=14$. For a fair comparison, if the depth falls short of 14, it is increased to 14 by adding logical idling gates. For most of the circuits $C^{-1} \circ C$, the compilation of $C$ and $C^{-1}$ into $H$, $S$ and CNOT is such that $C^{-1}$ does not undo each $H$, $S$ or CNOT in the order in which they are applied in $C$, ensuring that the gate sequence and hence the decoding hypergraph does not have a special structure. 
 
Finally, for each choice of Clifford $C$, we have two associated circuits: one in the $Z$-basis and one in the $X$-basis. In the discussion of the numerical results, we index the top (bottom) logical qubit as 1 (2). 

\added{We want to note that, although we could have used the deep logical circuits from Ref.~\cite{cain2024} (as done in Refs.~\cite{turner2025, cain+:upcoming}), their logical measurements are applied as noiseless multi-qubit Pauli product measurements, which are different than standard transversal $X$- and $Z$-basis measurements. Due to the importance of open-time boundaries in the \textsc{lom} decoder, we have preferred to use our arbitrary two-qubit Clifford circuits for benchmarking our decoder. }

\begin{figure*}[tb]
    \centering
    \includegraphics[width=0.99\textwidth]{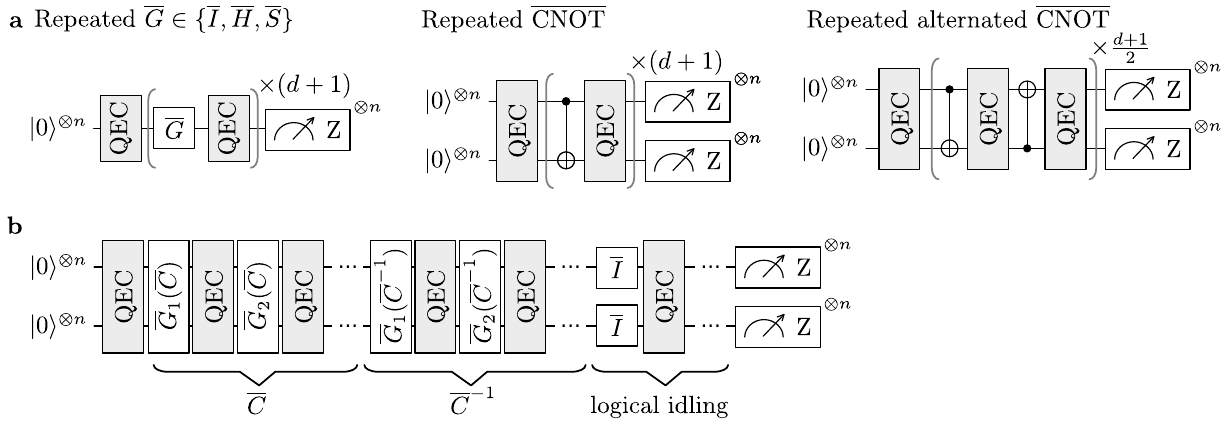}
    \caption{Logical circuits used to benchmark the decoder in the $Z$-basis: (a) repeated-gate experiments, (b) two-qubit Clifford experiments. The QEC box corresponds to a single QEC round. The gates $G_i(C)\in \{I, H, S, \text{CNOT}\}$ for $i=1,2,\ldots $ are the compilation of the Clifford gate $C$. The logical idling at the end of the circuit in (b) is used to fix the length to be independent of the Clifford $C$. The physical implementation of each block is described in Section~V in the SM~\cite{supp}.}
    \label{fig:experiment_circuits}
\end{figure*}

\subsection{Decoding phenomenological and circuit-level noise}\label{sec:circuit_level_slom}

Quantum device noise is more complex than the basic error model defined in Section~\ref{sec:noise_detectors_graph}; in particular, realistic errors can trigger more than three detectors, leading to new hyperedges. In this section, we show how these new hyperedges can be pre-processed before running the \textsc{lom} decoder. Unavoidably, the logical performance is diminished because correlation information is discarded when decomposing these new hyperedges into multiple edges in order to run a matching decoder. However, these hyperedges can be decomposed into a constant number of edges and therefore do not pose a fundamental fault-tolerance problem, unlike the weight-3 time-like hyperedges in Section~\ref{sec:noise_detectors_graph}.
We use two noise models in our numerical simulations, phenomenological depolarizing noise and SI1000 circuit-level depolarizing noise, defined as follows. In the phenomenological depolarizing noise model:
\begin{itemize}
    \item one applies a Pauli error drawn from $\{X, Y, Z\}$ with probability $p/3$ to each data qubit independently, before a logical gate, before a logical idling step, and before each single QEC round, 
    \item one applies an $X$ error with probability $p$ before a physical qubit $Z$ measurement. 
\end{itemize}
The reason to both insert Pauli noise before the logical gate and after the logical gate (i.e.~directly before the QEC round) is that these Pauli errors have different effects in the decoding hypergraph, see Section~II in the SM~\cite{supp}. 
We choose depolarizing noise in the phenomenological setting instead of independent $X$ and $Z$ errors~\cite{Dennis_2002}, since it is symmetric with respect to $X, Y, Z$ and our logical Clifford gates transform these Paulis into each other.

We use the SI1000 circuit-level noise model~\cite{gidney2021fault} to benchmark our decoder in conditions similar to future experimental realizations of logical circuits. This model is inspired by the superconducting qubit error probabilities, with two-qubit gate infidelities ten times higher than single-qubit gates and five times lower than measurement errors. Although the superconducting qubit platform does not currently have the long-range connectivity required for fold-transversal gates, the mentioned infidelity factors are also similar for neutral atoms, with infidelities of $\sim 0.03\%$ for single-qubit gates, $\sim 0.5\%$ for two-qubit gates, and $1-5\%$ for measurement operations~\cite{evered2023high, rodriguez2024experimental, graham2023midcircuit, deist2022mid}. Therefore, we believe this noise model is a good proxy for near-term realizations of logical experiments with transversal gates. 

A natural method of pre-processing the new hyperedges follows the standard procedure for memory experiments~\cite{gidney2021stim}. It decomposes the new hyperedges in terms of errors from the basic error model and updates the error probabilities of these basic errors accordingly. We note that there always exists a decomposition into $O(1)$ basic errors for any error in our phenomenological and circuit-level noise models since any circuit-level error can be propagated to a set of $O(1)$ basic errors. Since the \textsc{lom} decoder can correct any error of weight less than $d/2$, it follows that the \textsc{lom} decoder can decode any error of weight less than $d/c$ with $c\geq 2$, with the constant $c$ determined by the decomposition of circuit-level errors into basic errors. 

In particular, we have applied Algorithm~3 from Ref.~\cite{delfosse2023} to the decoding subgraph to obtain the hyperedge decompositions, although other methods are available~\cite{gidney2021stim}. An alternative method is to simply remove the new hyperedges from the decoding hypergraph. It is a reasonable option if the probability update from the hyperedge decomposition is small, or if it is similar for all edges so it leads to just an overall weight rescaling not impacting the matching preference. We have compared this option with the standard method for the repeated-gate experiments and see a relative difference of $<10\%$ in the logical error probability. In our numerical results, we are then using this alternative method for its simplicity. 

\subsubsection{Circuit distance of the repeated-gate experiments} \label{sec:d_circ_rep_exp}

Prior to decoding any experiments, as a sanity check, we have also computed the circuit distance $d_{\mathrm{circ}}$ for the entire repeated-gate experiment circuits for $d=3$ and $d=5$, assuming circuit-level noise, by solving the equivalent MaxSAT problem of finding the lowest-weight undetectable logical error, see Section~IV in the SM~\cite{supp} and Ref.~\cite{qec-util}. We numerically found that $d_{\mathrm{circ}}=d$ suggesting that a minimum-weight decoder would correct up to $< d/2$ errors, hence the circuits implementing the logical gates are fault-tolerant. \added{The $d=7$ computation involves more than 1600 variables and 1700 conditions, and} was stopped because it did not finish even after more than 20 hours. Since we are computing the circuit distance of {\em the entire circuit}, the decoding volume of the repeated-gate experiments grows as $\Theta(d^3)$, making this problem computationally challenging. Similarly, for the two-qubit experiments, the decoding volume is large because the considered circuits are deep, thus computing their circuit distance is computationally expensive.

In addition to the circuit distance, we have also computed the lowest-weight undetectable error pattern---we call it $\vec{e}_{\rm min}$---which leads to a decoding failure in the repeated-gate experiments {\em when decoding with the single-\textsc{lom} decoder}. Note that this error can have weight less than $d$ since we are now using a circuit-level error model instead of the basic error model, and because the single-\textsc{lom} decoder discards syndrome information since it only considers the decoding subgraphs $G_O$. As mentioned, the circuit-distance of the \textsc{lom} decoder can be reduced by a constant factor, since we know that the \textsc{lom} decoder can correct any basic error of weight less than $d/2$ and any circuit-level error can be expressed as $O(1)$ basic errors.
For $d=3,5,7$, with phenomenological noise, we have numerically obtained $|\vec{e}_{\rm min}| = 3,5,7$ respectively across all repeated-gate experiments, verifying that our decoder is circuit-distance preserving. For circuit-level noise and $d=3,5$, we get $|\vec{e}_{\rm min}| = 3,5$ except for the repeated-$\overline{S}$ experiment in the $X$-basis where $|\vec{e}_{\rm min}| = 2, 4$. 
The weight of this error $\vec{e}_{\rm min}$ is one lower due to a ``hook'' error and its error propagation from the long-range CZs in the physical implementation of the $\overline{S}$ gate, see Fig.~2 from SM~\cite{supp}. We discuss the effect of this distance reduction in the numerics.

\subsubsection{Logical failures} \label{sec:logical_failure}


We consider a logical failure to have occurred whenever any of the error-corrected logical outcomes from the experiment do not match the prepared logical input state ($+1$ eigenstate of logical $Z$ or $X$). The logical error probability, $\overline{P}$, is the probability of having a logical decoding failure. 
We estimate $\overline{P}$ using Monte-Carlo sampling, and we stop when reaching $5 \cdot 10^4$ logical decoding failures or $5\cdot 10^9$ samples, whichever is achieved first. The uncertainty in the logical error probability is calculated using the Wilson score interval with a 95\% confidence level~\cite{Wilson01061927, heussen2024dynamical}, which is more reliable at $\overline{P} \rightarrow 0$ than using a normal approximation or Wald interval~\cite{brown2001interval}. The threshold values have been estimated using the method from Ref.~\cite{hillmann2024single}, which is robust when the logical error probabilities around the threshold are close to their saturation value of $1 - 0.5^k$ with $k$ the number of logical qubits in the experiment.

\subsection{Numerical failure of splitting-hyperedge decoding}
\label{sec:fail_split}

To test the performance of the splitting hyperedge decoder discussed in \cref{sec:frag-time}, we run the repeated-gate experiments for $G \in \{I, S, \mathrm{CNOT}\}$ using the fold-transversal implementation of these gates and the $d=3,5,7$ unrotated surface codes, under circuit-level depolarizing noise. We decode these numerical experiments in both pre-gate and post-gate frames using a splitting-hyperedge \textsc{mwpm} decoder, in which the hyperedges in the decoding graph are attempted to be decomposed into edges using Stim with 
\begin{verbatim}stim.Circuit.detector_error_model(
    decompose_errors=True, 
    ignore_decomposition_failures=True,
)
\end{verbatim}
The logical error probability is shown in Fig.~\ref{fig:threshold_repeated_exp_matching}. 
We have also used the ``splitting-hyperedge'' methods from Ref.~\cite{delfosse2023} and obtained similar logical performance. 

As discussed in Section~\ref{sec:noise_detectors_graph}, this decoding strategy is not fault-tolerant because there exist weight-2 error patterns that lead to a logical error when decoded, for any code distance. Such bad performance is observed in Fig.~\ref{fig:threshold_repeated_exp_matching}, with the logical error probability scaling as $p$ for small $p$ for distances 3, 5, and 7, and saturating for large $p$ at $\overline{P} \approx 1/2$ for experiments on a single logical qubit and $\overline{P} \approx 3/4$ for experiments on two logical qubits. Note that this is worse than the $O(p^2)$ scaling predicted in \cref{sec:noise_detectors_graph} due to the presence of circuit-level noise. This $O(p)$ scaling is observed for both pre-gate and post-gate frames. Therefore, the splitting-hyperedge \textsc{mwpm} decoder cannot handle the hyperedges present when performing only 1 QEC round after each logical gate and has motivated us to develop the decoding method in this paper to decode across transversal gates. 

We emphasize that if logical gates are separated by $\Theta(d)$ QEC rounds, it could be possible to decode each gate in an individual decoding window using the decoding algorithms from Ref.~\cite{chen2024transversal} for the $\overline{S}$ gate and from Ref.~\cite{sahay2024, wan2024, guernut2024} for the $\overline{\mathrm{CNOT}}$ gate. Or, alternatively, Ref.~\cite{cain2024} shows how the decoding hypergraph can be modified in this situation so that it can be decoded with a splitting-hyperedge \textsc{mwpm} decoder. However, these decoding methods slow down the logical circuit execution by a factor $\Theta(d)$, thus not exploiting the advantage of fast logic that comes from performing transversal gates.

\begin{figure*}[htb]
    \centering
    \includegraphics[width=0.9\textwidth]{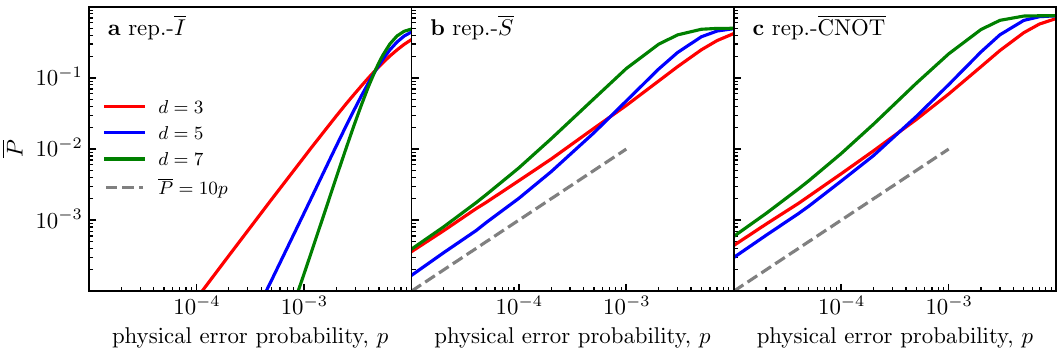}
    \caption{Logical error probability of the `splitting-hyperedge' \textsc{mwpm} decoder as a function of the physical error probability for some repeated-gate experiments, showing a scaling of $\overline{P}$ with $p$ and no evidence of a threshold. The experiments are run on unrotated surface codes with circuit-level noise, in the $\overline{Z}$-basis, using the pre-gate frame. The results for the post-gate frame are statistically identical. To decode, we use Stim to decompose hyperedges into edges and then run PyMatching, which we call a splitting-hyperedge \textsc{mwpm} decoder. \added{The 95\% confidence intervals are shown as shaded regions but they are too small to be visible. }}
    \label{fig:threshold_repeated_exp_matching}
\end{figure*}


\subsection{Decoding across logical two-qubit Clifford gates}

This section contains the numerical results of our \textsc{lom} decoder for the two classes of experiments. The circuits have been decoded using both the pre-gate and post-gate detector frames. We have found that the logical error probabilities are statistically identical for both frames in all considered scenarios, thus the discussion and figures only report the results in the pre-gate frame. The reason for such similar performance is that the decoding hypergraphs of both frames differ neither in the number nor the weight of the hyperedges, but only in which vertices participate in each hyperedge,
i.e.~they both have weight-3 as in Fig.~\ref{fig:hyperedges}, but just differ in which three vertices form a hyperedge. 

We have not implemented the windowed-\textsc{lom} decoders which will be defined in \cref{sec:window} and we leave their numerical benchmarking as future work. 


\subsubsection{Repeated-gate experiments} \label{sec:results_rep_exp}

The repeated-gate experiments correspond to memory experiments with logical gates instead of idling between QEC rounds. Such a similarity allows us to benchmark our decoder on each logical gate individually by comparison to the memory, repeated-$\overline{I}$,  experiment. 
Our decoder for the non-trivial repeated-gate experiments obtains a very similar performance to the repeated-$\overline{I}$ experiments, see Fig.~\ref{fig:threshold_repeated_exp} for both phenomenological and circuit-level noise (and Fig.~4 from SM~\cite{supp} for additional results). We observe threshold behavior and overall good logical error suppression.
We also run the same experiments with 2 and 3 QEC rounds between the $d+1$ logical gates, but we see worse performance (similar to results in Ref.~\cite{cain2024}): this is not surprising as the total sequence of physical gates is simply longer in these experiments with more interspersed QEC rounds.

The experiments with worse performance are the repeated-$\overline{S}$ in the $\overline{X}$-basis and the two repeated-$\overline{\mathrm{CNOT}}$ experiments. Regarding the repeated-$\overline{S}$ experiment, a contribution to the difference between the $\overline{X}$- and $\overline{Z}$-bases is that the number of circuit-level errors that the \textsc{lom} decoder can correct is smaller by 1 in the former, as described in Section~IV in the SM~\cite{supp}. Another factor is the difference in the probability of a logical Pauli error that can flip the measurement outcome for both bases. For surface codes, the minimum weight of Pauli $\overline{Y}$ is larger than that of Pauli $\overline{X}$ and $\overline{Z}$ ($2d-1 > d$), hence it is more probable to have $\overline{X}$ or $\overline{Z}$ errors than $\overline{Y}$ errors. For the repeated-$\overline{S}$ experiment in the $\overline{Z}$-basis, the observable to protect throughout the circuit is always $\overline{Z}$, thus we are susceptible to only $\overline{X}$ and $\overline{Y}$ errors. However, in the $\overline{X}$-basis, the $\overline{S}$ gates make the observable change from Pauli $\overline{X}$ to Pauli $\overline{Y}$ and vice-versa. Since $\overline{Y}$ is susceptible to both $\overline{X}$ and $\overline{Z}$ errors, this leads to a worse performance when carrying out the experiment in the $\overline{X}$-basis. In the other experiments, the logical gates do not map $\overline{X}$ nor $\overline{Z}$ to $\overline{Y}$, thus this effect is not present, see Section~VI in the SM~\cite{supp}. 


Regarding the repeated-gate experiments with $\overline{\mathrm{CNOT}}$ gates, there are two observables to decode: $\overline{Z}_1$ and $\overline{Z}_2$ for the $\overline{Z}$-basis experiment and $\overline{X}_1$ and $\overline{X}_2$ for the $\overline{X}$-basis experiment. 
Due to the Pauli propagation through the $\overline{\mathrm{CNOT}}$ gates, one of the two observables will have support on both logical qubits at some point in the circuit. This observable is going to be susceptible to logical Pauli errors happening on both logical qubits, compared to just one of them for the memory experiment. Therefore, the repeated-gate experiments with $\overline{\mathrm{CNOT}}$ gates perform worse than the memory experiment. The probability of incorrectly inferring each one of the individual observables is consistent with these arguments and shown in Fig.~\ref{fig:threshold_repeated_exp_cnot}. 

We note that the repeated-$\overline{I}$ and $\overline{\mathrm{CNOT}}$ experiments in the $\overline{Z}$-basis have slightly lower performance than in those in the $\overline{X}$-basis under circuit-level noise, see Fig.~\ref{fig:threshold_repeated_exp}(f),(i)--(j). Although one might expect the opposite due to the extra layer of $H$ gates during the initialization of $\ket{+}^{\otimes n}$ (see the circuits in Section~V in the SM~\cite{supp}). Refs.~\cite{zheruithesis, o2024compare} explain this discrepancy with the order of the CZs during the QEC round which affects $\overline{X}$ differently than $\overline{Z}$. This is not observed in the repeated-$\overline{H}$ experiment as it symmetrizes $\overline{X}$ and $\overline{Z}$. 

\begin{figure*}[htb]
    \centering
    \includegraphics[width=0.4\textwidth]{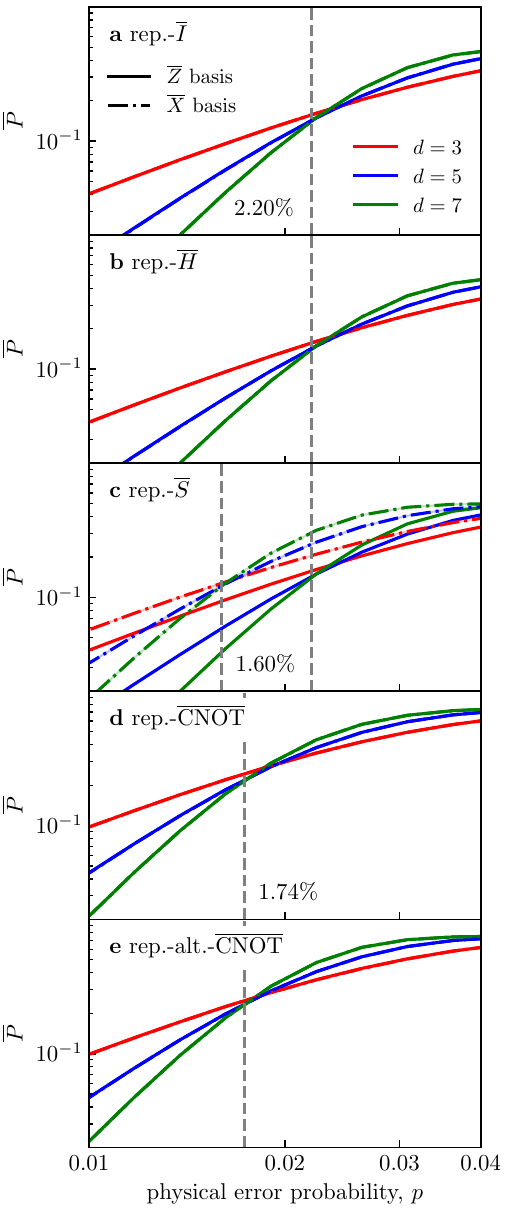}
    \includegraphics[width=0.4\textwidth]{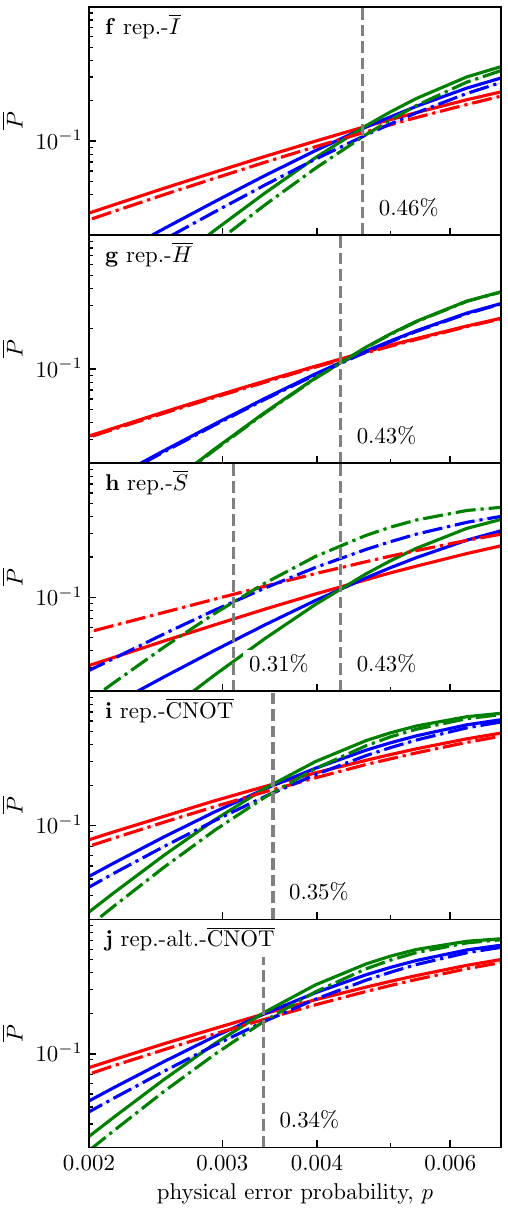}
    \caption{Logical error probability of the \textsc{lom} decoder as a function of the physical error probability for the repeated-gate experiments under (a--e) phenomenological depolarizing noise and (f--j) circuit-level noise for the pre-gate frame. In the experiments, each logical operation is followed by one QEC round, thus $d+2$ rounds are performed, see Fig.~\ref{fig:experiment_circuits}(a). The vertical dashed gray lines indicate the crossing of $\overline{P}$ for the $d=5$ and 7 surface codes. \added{The 95\% confidence intervals are shown as shaded regions but they are too small to be visible. }}
    \label{fig:threshold_repeated_exp}
\end{figure*}

\begin{figure*}[htb]
    \centering
    \includegraphics[width=0.9\linewidth]{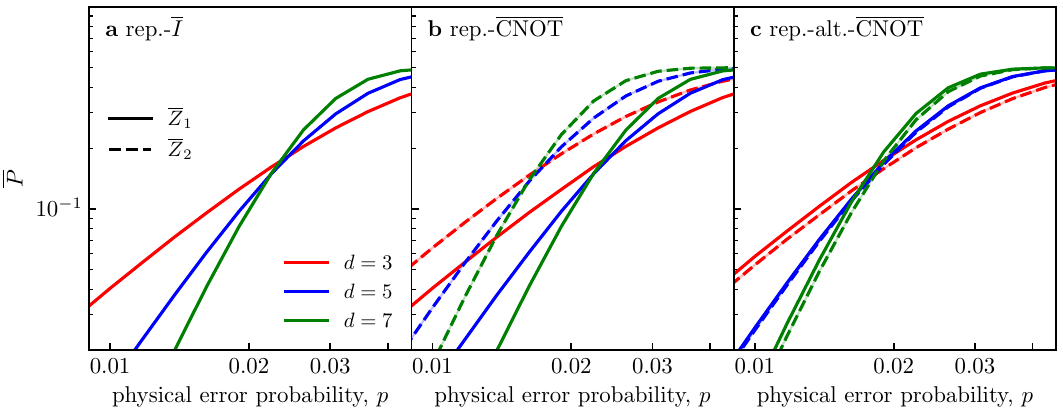}
    \caption{Focusing on single logical observables for the two-qubit repeated-gate experiments under phenomenological noise. As one may expect, the performance of $\overline{Z}_1$ in the repeated-$\overline{\mathrm{CNOT}}$ experiment in (b) matches the performance of the repeated-$\overline{I}$ experiment in (a) while $\overline{Z}_2$, being sensitive to logical $X$ errors on both qubits, performs worse. For the alternating-CNOT experiment, the logical error rates of $\overline{Z}_1$ and $\overline{Z}_2$ are symmetrized. \added{The 95\% confidence intervals are shown as shaded regions but they are too small to be visible. }}
    \label{fig:threshold_repeated_exp_cnot}
\end{figure*}

\subsubsection{Preferential treatment}
\label{sec:pt}
We show numerical evidence for the recommendation mentioned at the end of Section \ref{sec:reliable} about how to build the generating set of final observables. In particular, we mentioned that, if the final outcome of an algorithm is some \textsc{xor} of final measurements, it is advisable to include this as a (reliable) observable in the generating set to minimize its logical error rate. We consider the case of the repeated-alternated-$\overline{\mathrm{CNOT}}$ experiment in the $\overline{Z}$ basis (whose initial logical state is $\ket{\overline{00}}$) under phenomenological depolarizing noise and use two methods for obtaining the error-corrected $\overline{Z}_1\overline{Z}_2$ outcome, see Fig.~\ref{fig:threshold_cnot_z1z2}. The first one corresponds to running a single-\textsc{lom} decoder for the $\overline{Z}_1\overline{Z}_2$ observable, while in the second one the $\overline{Z}_1$ and $\overline{Z}_2$ observables are decoded independently with single-\textsc{lom} decoders and then the error-corrected $\overline{Z}_1\overline{Z}_2$ outcome is obtained by taking their XOR. The first method has a slightly better performance. 

\begin{figure}[htb]
    \centering
    \includegraphics[width=0.85\linewidth]{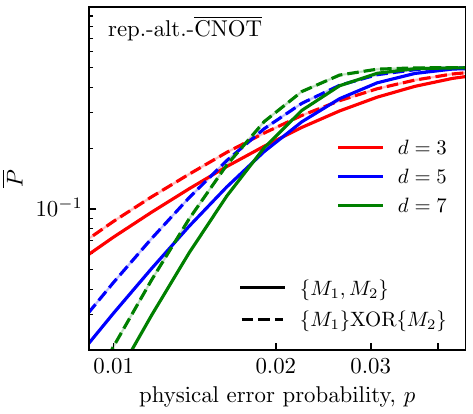}
    \caption{Performance of the \textsc{lom} decoder for the $\overline{Z}_1\overline{Z}_2$ observable of the repeated-alternated-$\overline{\mathrm{CNOT}}$ when using two different methods of obtaining the logical correction. The ``$\{M_1, M_2\}$'' method corresponds to running a single-\textsc{lom} decoder for the $\overline{Z}_1\overline{Z}_2$ observable. The ``$\{M_1\} \mathrm{XOR} \{M_2\}$'' method corresponds to running single-\textsc{lom} decoders for the $\overline{Z}_1$ and $\overline{Z}_2$ observables separately and then obtaining the error-corrected $\overline{Z}_1\overline{Z}_2$ outcome by taking their XOR. The circuit is run with phenomenological depolarizing noise. \added{The 95\% confidence intervals are shown as shaded regions but they are too small to be visible. }}
    \label{fig:threshold_cnot_z1z2}
\end{figure}

\subsubsection{Arbitrary two-qubit Clifford experiments}

In this section, we show that our decoder maintains its performance when decoding all two-qubit Clifford experiments under phenomenological and circuit-level depolarizing noise. In Fig.~\ref{fig:threshold_clifford_exp}, we report the average (solid line) logical error probabilities across all circuits, as well as the best and worst performing ones (shaded region). For phenomenological noise, we observe the same scaling as for the repeated-gate experiments (i.e. $p^{\lceil (d+1)/2 \rceil}$), while for circuit-level noise, the average and worst cases seem to scale as $p^{\lceil (d+1)/2 \rceil - 1}$ due to the presence of $\overline{S}$ gates for which the circuit distance can be reduced by 1 as shown in Section~IV in the SM~\cite{supp}. The exponential suppression of the physical errors shows that our decoder can decode any two-qubit Clifford unitary and does not exploit any particular structure of the decoding hypergraph. 

The crossings of the average $\overline{P}$ [Fig.~\ref{fig:threshold_clifford_exp}(a),(c)] are hard to see because of the large depth of the two-qubit Clifford experiments. The logical error probability is very close to $\overline{P}_{\mathrm{max}} = 3/4$ at the crossings, which has also been reported by Ref.~\cite{cain2024} for these types of (deep) logical circuits. We compute the thresholds using the method described in Ref.~\cite{hillmann2024single}, which has been designed for this situation of $\overline{P}$ saturating to its maximum close to the crossings. The threshold values for the average $\overline{P}$ across all circuits under phenomenological noise are 2.9576\added{(2)}\% and 2.8212\added{(2)}\% for the $\overline{Z}$ and $\overline{X}$ basis (resp.), while for circuit-level noise they are 0.5036\added{(1)}\% and 0.4963\added{(1)}\% (resp.).
We note that the crossing between $d=5$ and 7 is at a higher physical error probability than the repeated-$\overline{I}$ (memory) experiment for both noise models: these results are not comparable as in the repeated-gate experiments the depth of the logical circuit increases with $d+1$ while for the two-qubit Clifford experiments it is simply fixed at 14. 

Apart from the average logical error probability, we are also interested in the performance difference between the best and worst circuits. We focus the discussion on phenomenological noise, because the location and number of noise channels are the same for all two-qubit Clifford experiments, ensuring a fair comparison across circuits. As expected, the best-performing circuit corresponds to the memory experiment while the worst one is
\begin{widetext}
\begin{equation} \label{eq:worse_clifford_circuit}
    \vcenter{\hbox{\includegraphics{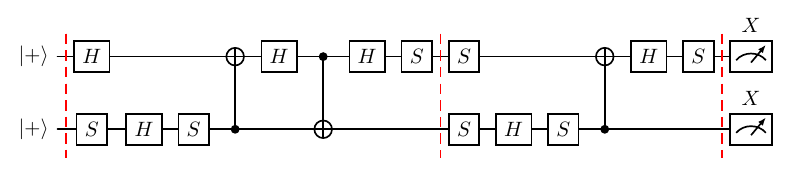}}}
    .
\end{equation}
\end{widetext}
In this latter circuit, the observables are propagated to both qubits and transformed into Pauli $\overline{Y}$, making them more prone to errors, as discussed in Section~\ref{sec:results_rep_exp}. The factor $\overline{P}_{\mathrm{worst}} / \overline{P}_{\mathrm{best}}$ at a fixed physical probability grows with the code distance, with values 5, 12, and 22 for $d=3, 5, 7$ respectively at $p = 2 \cdot 10^{-3}$. At such a low physical error probability, $\overline{P}$ is dominated by $C \cdot p^{\lceil \frac{d + 1}{2} \rceil}$. The entropic prefactor $C$ depends on the number of lowest-weight error patterns leading to a logical failure. The larger the size of the observing region, the faster $C$ grows with the distance. Therefore, the factor $\overline{P}_{\mathrm{worst}} / \overline{P}_{\mathrm{best}} \sim C_{\mathrm{worst}} / C_{\mathrm{best}}$ increases with the distance at a fixed physical probability.

\begin{figure*}[p]
    \centering
    \includegraphics[width=0.4\textwidth]{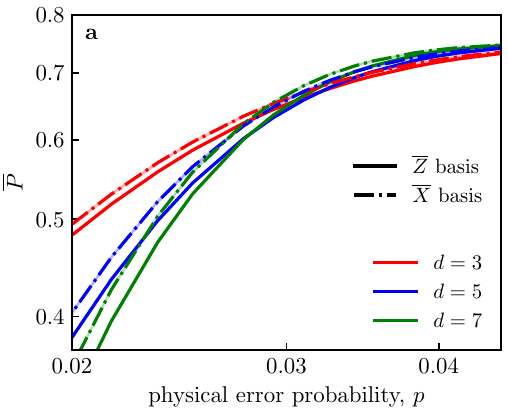}
    \includegraphics[width=0.4\textwidth]{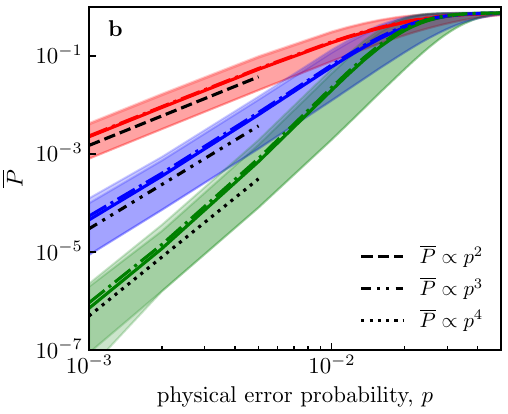}
    \includegraphics[width=0.4\textwidth]{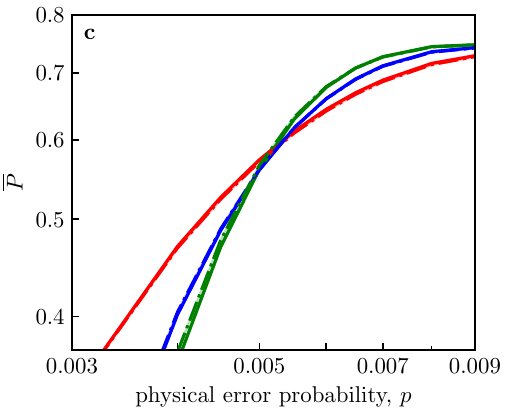}
    \includegraphics[width=0.4\textwidth]{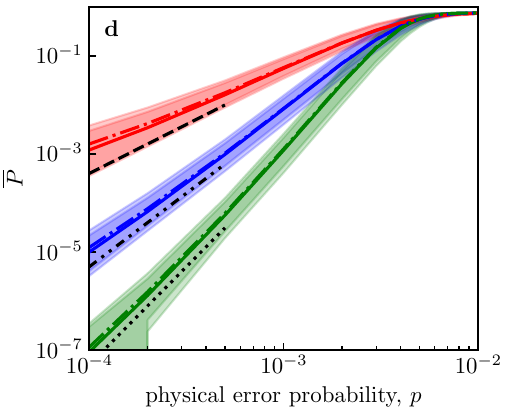}
    \caption{Logical error probability of the \textsc{lom} decoder as a function of the physical error probability for the two-qubit Clifford experiments in the pre-gate frame under (a, b) phenomenological and (c, d) circuit-level noise. 
    The solid (resp. dashed) line corresponds to the average over all experiments in the $\overline{Z}$- (resp. $\overline{X}$-) basis. \added{The shaded regions in (a, c) are the 95\% confidence intervals, but they are too small to be visible.} The shaded regions in (b, d) show the best and worst performance across all circuits. }
\label{fig:threshold_clifford_exp}
\end{figure*}

\subsubsection{Comparison with minimum-weight decoding}

We compare the performance of our decoder to minimum-weight decoding, which we implemented in~\cite{mle_decoder} by solving the equivalent mixed-integer (linear) programming problem~\cite{landahl2011, cain2024} with the Gurobi solver~\cite{gurobi}. 
Because minimum-weight decoding is computationally expensive, we only decode the best and worst performing two-qubit Clifford circuit, that is, a two-qubit memory experiment and circuit~\ref{eq:worse_clifford_circuit} respectively, under circuit-level noise, see Fig.~\ref{fig:mle_comparison}. For the memory experiment, both decoders perform similarly with just a factor of two difference because minimum-weight decoding does not have to decompose the hyperedges and can process them directly. For the worst-performing circuit, the performance difference increases at lower physical error probabilities due to the worse scaling of the \textsc{lom} decoder in the presence of $\overline{S}$ gates. We expect the performance of the \textsc{lom} decoder to be closer to minimum-weight decoding when running belief propagation as a pre-step to update the decoding hypergraph probabilities.

\begin{figure*}[p]
    \centering
    \includegraphics[width=0.75\linewidth]{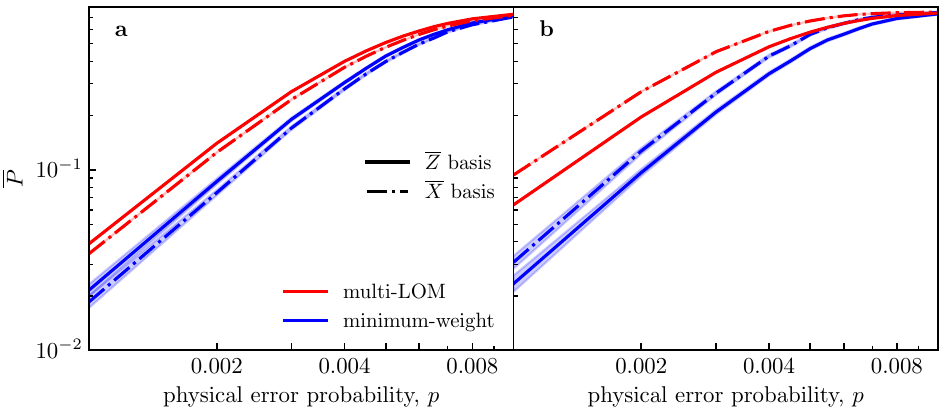}
    \caption{Performance comparison between \textsc{lom} and (approximate) minimum-weight decoding for (a) the two-qubit memory experiment and (b) circuit~\ref{eq:worse_clifford_circuit} for the distance-3 unrotated surface code under circuit-level noise. The solid (dashed) line corresponds to the experiment in the $\overline{Z}$- (resp. $\overline{X}$-) basis. \added{The shaded regions show the confidence interval with 95\% confidence level. Note that for the \textsc{lom} decoder, the interval is too small to be visible. } }
    \label{fig:mle_comparison}
\end{figure*}

\subsubsection{\added{Deep Clifford circuits}}

\added{This section contains a benchmark of the \textsc{lom} decoder applied to deep Clifford circuits~\cite{cain2024}, giving a direct numerical comparison with decoding via the Ghost Protocol in~\cite{turner2025}. The goal of this section is to show that one can easily enhance the performance of the \textsc{lom} decoder by taking into account correlations (e.g. due to $Y$ errors) by running belief propagation (BP) in the \textit{full hypergraph} as a pre-processing step to optimize error probabilities in the \textsc{lom} decoding graphs. The Ghost Protocol decoder already tries to account for these correlations directly. In our \textsc{bp+lom} decoder, if the belief propagation does not converge after $n_{\rm BP}$ iterations, then \textsc{lom} is executed using the updated error probabilities from the last belief-propagation iteration. }

\added{A preparation step involves initializing four surface code logical qubits in $\ket{\overline{+}}^{\otimes 4}$ followed by 32 logical gate layers, where each layer is composed of random transversal $\{H, X, Y, Z\}$ gates followed by random transversal CNOTs, see Ref.~\cite{cain2024} for more information. One then decodes the multi-Pauli (logical) stabilizers of the corresponding four (logical) qubit state, which should ideally have outcome +1. We consider a decoding failure to have occurred whenever any of the error-corrected logical outcomes from the circuit do not match the ideal outcome +1. 
The noise model used in the simulations of these deep Clifford circuits is the one from Ref.~\cite{turner2025}, with the same physical error rate $p=10^{-3}$ across all code distances, and the random circuit that is implemented is identical to the one in \cite{turner2025}.}

\added{The performance of the \textsc{lom}, \textsc{bp+lom}, and Ghost Protocol decoders on deep Clifford circuits with a single QEC round per gate layer is shown in Fig.~\ref{fig:deep_logical_circuits}. As expected, \textsc{bp+lom} is more accurate than \textsc{lom} due to belief propagation taking into account correlations, with its performance increasing with the number of iterations. 
In particular, very few belief-propagation iterations can reduce the decoder failure probability by an order of magnitude at high distances, showing that little correlation information can heavily improve the \textsc{lom} performance. Finding a faster and more efficient method for this purpose is an open question and beyond the scope of this work. 
}

\begin{figure}
    \centering
    \includegraphics[width=0.99\linewidth]{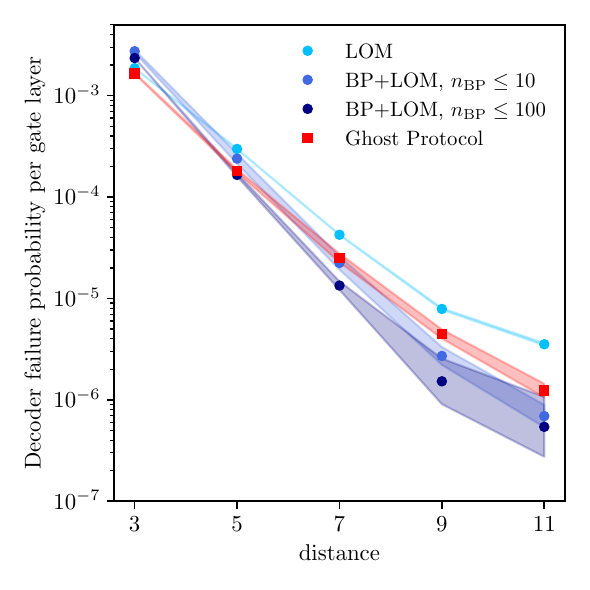}
    \caption{\added{Decoder failure probability per gate layer of the \textsc{lom}, \textsc{bp+lom}, and Ghost Protocol~\cite{turner2025} decoders for the deep Clifford circuit with a single QEC round per gate layer and noise model from Ref.~\cite{turner2025}. The circuit has 32 gate layers, where each layer is composed of random transversal $\{H, X, Y, Z\}$ gates followed by random transversal CNOTs on each qubit. All distances have been simulated with the same noise strength $p=10^{-3}$. The parameter $n_{\rm BP}$ in \textsc{bp+lom} corresponds to the maximum number of belief-propagation iterations that are run. The shaded regions show the confidence interval with 95\% confidence level. Note that, for $n_{\rm BP}=100$ at distances 9 and 11, the confidence intervals are quite large as only $\sim 10$ decoding failures have been sampled due to longer decoding runtimes. 
    }}
    \label{fig:deep_logical_circuits}
\end{figure}


\section{Windowed logical observable matching decoding}
\label{sec:window}

The \textsc{lom} decoder discussed so far is not a scalable approach to decoding because each instance of the single-\textsc{lom} decoder is independent and will therefore need to decode the circuit all the way back to the initial time step. To tackle this problem we introduce a \textit{sliding window} version of the \textsc{lom} decoder that we call windowed-\textsc{lom} decoder for short. This works by splitting the circuit into sliding windows in time, with each window decoded by splitting it into multiple single-\textsc{lom} decoding instances and committing to a part of the correction that they make. This way, later iterations of the single-\textsc{lom} decoder can use the previously-committed corrections to reduce the space-time volume of the decoding subgraph. 


\begin{table}[]
    \centering \colorred
    \begin{tabular}{c|c|c|c}
        Decoder & Efficient & Fault-tolerant & \makecell{Resets and\\measurements} \\\hline
        \makecell{Basic\\Sec.~\ref{sec:basic_window}} & \makecell{Yes*\\Sec.~\ref{sec:basic_windowed-LOM_efficiency}} & \makecell{No\\Sec.~\ref{sec:window-ft}} & Slow\\\hline
        \makecell{Two-step\\Sec.~\ref{sec:two-step_windowed-LOM}}& \makecell{No\\Sec.~\ref{sec:window_efficiency}} & \makecell{No\\Sec.~\ref{sec:window-ft}} & Fast\\\hline
        \makecell{Basic +\\ short-cut edges} & \makecell{Yes*\\Sec.~\ref{sec:basic_windowed-LOM_efficiency}}& \makecell{Conjectured\\App.~\ref{sec:short-cut_edges}} & Slow\\\hline
        \makecell{Two-step +\\ short-cut edges} & \makecell{No\\Sec.~\ref{sec:window_efficiency}} & \makecell{Conjectured\\App.~\ref{sec:short-cut_edges}} & Synchronized
    \end{tabular}
    \caption{Summary of the windowed-\textsc{lom} decoder variants introduced in this section. Here, ``efficient'' means that the decoder runs in polynomial time in $d$; the asterisk here is to indicate that the basic windowed-\textsc{lom} decoder is not efficient for some unphysical circuits. ``Fault-tolerant'' means that the decoder can always correct basic errors of weight less than $d/2$. All the variants can handle fast gates and $T$-injections, and the last column indicates whether they can also handle fast resets and measurements.}
    \color{black}
    \label{tab:windowed-LOM_variants}
\end{table}

We wish for the windowed-\textsc{lom} decoder to satisfy three properties. First, that it is computationally efficient to run, growing at most polynomially in the distance $d$ of the code independent of the number of logical qubits and the depth of the circuit. Second, that it is able to correct any basic error pattern with weight $<d/2$. And finally, we wish for this to happen while maintaining the ability to perform fast resets, gates, and measurements. Unfortunately, we are unable to give a construction where all these properties are satisfied simultaneously. Instead, we provide a number of versions of the windowed-\textsc{lom} decoder, each of which has different computational complexity, fault-tolerance and requirements on the circuit\added{; these are summarized in \cref{tab:windowed-LOM_variants}}.

We begin pedagogically with the simplest version of the windowed-\textsc{lom} decoder that we call the ``basic'' windowed-\textsc{lom} decoder. In this setting, we make some fairly strict assumptions about the circuit; namely, that resets are not fast, but slow, using $\Omega(d)$ rounds of QEC to reach the ``quiescent'' state before starting to interact with other logical qubits. With this, the basic windowed-\textsc{lom} decoder is computationally efficient, but it cannot always correct up to $d/2$ errors. 

We solve the issue of slow resets in \cref{sec:two-step_windowed-LOM} by introducing a second step to the windowed-\textsc{lom} decoder to give what we call the \textit{two-step} windowed-\textsc{lom} decoder. However, the second step is not computationally efficient in general, and the two-step windowed-\textsc{lom} decoder suffers from the same issues of fault-tolerance as the basic windowed-\textsc{lom} decoder.

In both the basic and two-step windowed-\textsc{lom} decoders, the same two modifications need to be made for it to be able to correct $d/2$ basic errors. First, it requires the ``synchronization'' of resets and measurements to occur just before or after each so-called window boundary, respectively. Notably, this synchronization is \textit{not} necessary for the measurements occurring in $T$-injections which can occur at any time in the circuit, but it does mean that algorithms with lots of other resets and measurements \added{(such as running parity-check circuits of some code using logical surface-code qubits)} will either be slowed down or will require additional idling logical qubits. And second, it requires the addition of extra edges in each single-\textsc{lom} decoder that we call ``short-cut'' edges. With these modifications, we conjecture that the windowed-\textsc{lom} decoder can always correct any number of basic errors $<d/2$.


This section is structured as follows. First, we briefly review the general structure of sliding window decoders, before explaining how the basic windowed-\textsc{lom} decoder works. Then, we discuss the problems with the basic windowed-\textsc{lom} decoder and their respective solutions, first discussing how fast resets compromise the computational efficiency of the windowed-\textsc{lom} decoder, and then why the basic windowed-\textsc{lom} decoder cannot correct up to $d/2$ basic errors. We leave a thorough optimization of each variant of the windowed-\textsc{lom} decoder to future work.

\subsection{Sliding window decoders}\label{sec:sliding_window}

We review how the sliding-window matching decoder works for a surface code memory experiment~\cite{Dennis_2002}, as shown in \cref{fig:memory_window} (or see Fig.~2 in~\cite{skoric2023parallel}). Each window consists of a {\em commit} region followed by a {\em buffer} region, with both regions having width $>d/2$ in time. Each window has a backwards time-boundary at the start of the commit region, and a forwards time-boundary at the end of the buffer region. The backwards time boundary of each window is closed---i.e. it has no time-like edges connected to the boundary node---because we know the past, while the forwards time-boundary is open---it has time-like edges connected to the boundary node, similar to a fragile time-boundary, as we don't know the future. We call the boundary between the commit and the buffer regions the \emph{center} of the window, which coincides with the backwards time-boundary of the next window. We label the times associated with the backwards time-boundary, the center of the window, and the forwards time-boundary $t_{\text{prev}}$, $t_{\text{center}}$, and $t_{\text{current}}$ (respectively).

The decoder is then run sequentially on overlapping windows to build a correction in the full hypergraph. 
Each time we commit to a correction, we update the defects in the decoding graph that still need to be corrected by flipping the vertices (defect $\leftrightarrow$ no-defect) that are the endpoints of edges in the committed correction. This guarantees that after committing to the correction, there are no more defects in the commit region of the window. This procedure can introduce new defects in the rest of the graph that we call \textit{artificial} defects which will be seen and decoded by the next window. 

\begin{figure}
    \includegraphics[width=0.9\linewidth]{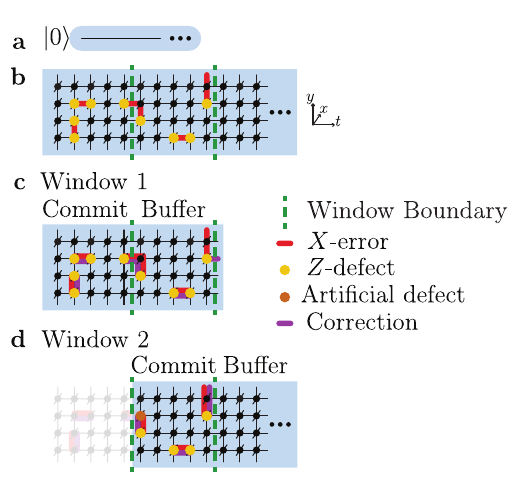}
    \caption{The operation of a sliding window matching decoder on a $Z$-basis memory experiment, shown in (a), in the $d=5$ surface code. (b) The $Z$-detectors in the first rounds of the full decoding graph with an arbitrarily chosen set of $X$-errors. The decoding problem is split into overlapping windows, the first two of which are shown in (c) and (d), including their respective buffer and commit regions and the correction that they apply. Note that only the correction in the commit region of window 1 is committed to, giving rise to an artificial defect in window 2.}\label{fig:memory_window}
\end{figure}

\subsection{The basic windowed-\textsc{lom} decoder}\label{sec:basic_window}


We now define the basic windowed-\textsc{lom} decoder for a circuit $\mathcal{C}$ with fast gates and measurements, but slow resets. In particular, we assume that each reset is followed by $\Omega(d)$ rounds of QEC to allow it to reach the ``quiescent'' state. During this $\Omega(d)$ delay, before the windowed-\textsc{lom} decoder is run, a round of normal sliding-window decoding is performed that commits to the correction in the vicinity of the fragile time-boundary. As a result, we assume that there are no open time-boundaries arising from resets in the windowed-\textsc{lom} decoder, allowing us to safely decode all observables regardless of whether they are fragile or not. Remember that we had to handle fragile observables separately in Section \ref{sec:fragile} for the \textsc{lom} decoder, thus making reliable observables potentially be of high-weight, supported on many logical qubits.

\begin{figure*}
    \includegraphics[width = 0.9\linewidth]{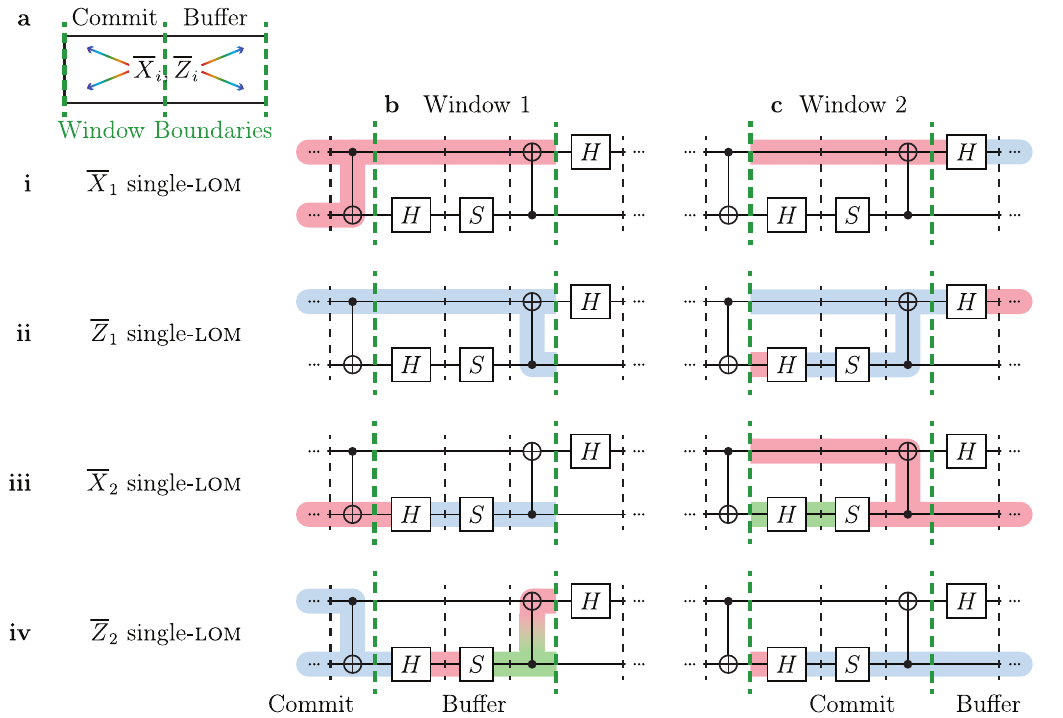}
    \caption{(a) A visual representation of the windowed-\textsc{lom} decoder, including both the buffer and commit regions. In each window, a single-\textsc{lom} decoding instance is run on the track obtained by forward- and backward-propagating a single-logical-qubit $\overline{X}$ or $\overline{Z}$ from the center of the window. (b,c) A detailed example of the single-\textsc{lom} tracks in two windows of the windowed-\textsc{lom} decoder implemented on the $d=3$ surface code, where red, green and blue regions represent $X$, $Y$ and $Z$ Pauli support respectively.}\label{fig:window_example}
\end{figure*}

The job of a windowed decoder is threefold:
\begin{enumerate}
    \item to infer the artificial defects in the following window. This only requires fixing the correction to the time-like edges \textit{at} the time $t_{\text{center}}$, but not before;
    \item to update the ``logical Pauli frame'' of the computation up to time $t_{\text{center}}$; that is, for each logical Pauli operator, we must commit to whether an error has occurred before time $t_{\text{center}}$ in its observing region; and
    \item to commit to all the logical measurement outcomes of the circuit in order to perform any gates that are conditioned on these outcomes. This must be done immediately to allow for fast $T$-injections.
\end{enumerate}

To execute these three jobs, we structure the windows as follows. Just like in the memory experiment, in the windowed-\textsc{lom} decoder we define a set of window boundaries separated by $>d/2$ rounds of QEC. As we execute the circuit, we run a window of the windowed-\textsc{lom} decoder each time we reach a window boundary, with the buffer and commit regions defined in the same way as in \cref{sec:sliding_window}.

In particular, the basic windowed-\textsc{lom} decoder must commit to the time-like edges and the logical Pauli frame. To do this, we run $2k$ independent instances of the single-\textsc{lom} decoder, where $k$ is the number of qubits that are active at the center of the window, see \cref{fig:window_example}. Each single-\textsc{lom} decoder is obtained by propagating a single-logical-qubit $\overline{X}$ or $\overline{Z}$ operator at $t_{\text{center}}$ \emph{forwards and backwards} through the window. Once a correction is obtained, we commit to the time-like edges \emph{at} the time $t_{\text{center}}$, but not before. This is well-defined because at each time $t_{\text{center}}$, each time-like edge is, by definition, contained in only one single-\textsc{lom}. Moreover, we can use the single-\textsc{lom} decoder to update the logical Pauli frame of the observable. That is, based on whether the correction intersects an even or odd number of times with the observing edge set in the commit region of the current window, we decide whether or not the logical Pauli operator has flipped between $t_{\text{prev}}$ and $t_{\text{center}}$. From this, we can determine whether the logical Pauli operator has been flipped through the \textit{whole} circuit up until $t=t_{\text{center}}$ by looking at whether the (backpropagation of) the Pauli operator was flipped in any of the previous windows. It is important that each single-\textsc{lom} has support only on a single-logical-qubit $\overline{X}$ or $\overline{Z}$ operator at the center of the window to ensure fault-tolerance, as will be discussed in \cref{sec:window-ft}.

However, in addition to running windows of decoders, we immediately run additional single-\textsc{lom} decoders for every logical measurement that occurs. This is important because the outcome may condition a Clifford gate. To do this, we simply run a single-\textsc{lom} decoder on the measurement back to the last point in the circuit where the correction has been committed to. To put this in context, the most-recent window boundary that we have passed represents the forwards time-boundary of the most-recently-decoded window, and the correction has been committed up to the second-most-recent window boundary. Note, of course, that because we have slow resets, even measurements that are fragile can be decoded in this way, because there are no more fragile time boundaries from the slow resets as described above. We then use this single-\textsc{lom} decoder to commit to the measurement outcome, without committing to any edges in the correction. This is advantageous in the presence of $T$ gates, where the subsequent conditional $S$ gate can thus be performed without delay. The windowed-\textsc{lom} decoder therefore does not require any additional ancilla qubits to perform fast $T$ gates, in contrast to the windowed minimum-weight decoder considered in Ref.~\cite{zhang2025}.

\subsubsection{\added{Efficiency of the basic windowed-\textsc{lom} decoder}}\label{sec:basic_windowed-LOM_efficiency}

The basic windowed-\textsc{lom} decoder is highly parallelizable but is limited by the maximum time it takes to decode a single logical observable. The maximum space-time decoding volume of each single-\textsc{lom} decoder is bounded by the number of logical qubits that a single-logical-qubit Pauli operator propagates to in the forwards- or backwards direction within the window. To quantify this, we define $f(t)$ as the maximum number of qubits that a single-qubit Pauli operator can spread to in a circuit of depth $t$. Each single-\textsc{lom} propagates from a single qubit through a circuit of depth $t=\Theta(d)$.  If there are no limitations on the circuit, then $f(t)=2^{t}$, and the volume of the decoding graph grows exponentially in $d$, as $O(d^{3}2^{\Theta(d)})$. Note that for scrambling circuits, which could be Clifford circuits, such spread is maximized~\cite{Nahum_2018}. 

However, in any physical implementation of the circuit, we must have $f(t)=O(t^{D})$ for some integer $D$. This is true even in platforms without fixed connectivity such as neutral atoms, because there is a limit to the distance that an atom can travel in preparation for a CNOT gate before it accumulates enough noise that a QEC round needs to be performed. In such physical cases, the volume of the decoding graph grows only polynomially as $O(d^{3+D})$.

We remark that just because the decoder is efficient does \textit{not} mean that it avoids the decoder back-log problem as described in~\cite{terhal:review}, which instead would require some kind of parallel windowed decoder~\cite{skoric2023parallel}.

\subsection{\added{The two-step windowed-\textsc{lom} decoder}}\label{sec:two-step_windowed-LOM}

In \cref{sec:basic_window}, we argued that the basic windowed-\textsc{lom} decoder is computationally efficient when the circuit contains only slow resets, whenever the circuit is implemented on a physical device. It is worth noting that slow resets may not be much of a hindrance for many applications. For example, if the desired algorithm is written in terms of Clifford and $T$ gates with resets and measurements only occurring at the start and end of the algorithm respectively, then slow resets and synchronized measurements only \textit{add} an $O(d)$ time cost to the start and the end of the algorithm. This cost is likely to be small compared to the time it takes to execute the gates in the circuit.

Nevertheless, in this section we explain how to modify the windowed-\textsc{lom} decoder to handle fast resets using a second step in each window, and why this comes at the cost of the efficiency of the decoder. The resulting two-step windowed-\textsc{lom} decoder can decode circuits with fast resets, measurements, and gates, but is computationally inefficient \added{in general}. We leave solving this efficiency issue as an open problem.

\added{
The key issue that the two-step windowed-\textsc{lom} decoder needs to solve is how to handle fragile observables. Recall that the basic windowed-\textsc{lom} decoder avoids this problem because the low resets mean that even fragile observables have no open time-boundaries on resets, and therefore can be safely decoded. Now, with fast resets, we must adapt the techniques developed in \cref{sec:fragile} to work for the windowed-\textsc{lom} decoder. In particular, we need to be careful about decoding any observables that are \textit{commit-region fragile}: that is, the observable anticommutes with a reset occurring in the commit region between $t_{\mathrm{prev}}$ and $t_{\mathrm{center}}$. Indeed, if the observable anticommutes with a reset that occurs before the backwards time-boundary $t_{\mathrm{prev}}$, then the corresponding single-\textsc{lom} will not contain the open time-boundary caused by that reset and it can therefore be decoded normally. We will refer to observables that are not commit-region fragile as commit-region reliable.

Dealing with commit-region fragile observables is complicated because of the multi-faced nature of the windowed-\textsc{lom} decoder: each window needs to commit both to the time-like edges at $t_{\mathrm{center}}$ and to the logical Pauli frame flips in the commit region (we will return to the third job of the windowed-\textsc{lom} decoder---committing to logical measurement outcomes---later).
}
Importantly, the single-\textsc{lom} decoding instances that are used to commit to the time-like edges must \textit{only} have support on a single-logical-qubit $\overline{X}$ or $\overline{Z}$ operator at the center of the window. This is necessary to provide a unique assignment of the time-like errors, and to avoid time-like loops as we will discuss in \cref{sec:window-ft}.
\added{
On the other hand, the single-\textsc{lom} decoding instances that are used to commit to the Pauli frame flips cannot be commit-region fragile. These two conditions cannot be satisfied simultaneously in general, necessitating the use of two steps in the decoder as we now explain.
}


\added{
In the first step of the two-step windowed-\textsc{lom} decoder, we decode observables that commit to the time-like edges at $t_{\text{center}}$ that have support only on a single-logical-qubit $\overline{X}$ or $\overline{Z}$ at $t_{\text{center}}$. This step is the same as in the basic windowed-\textsc{lom} decoder, except for the fact that we do not commit to the logical Pauli frame flips in the commit region. Then, in the second step, we decode observables that commit to the logical Pauli frame flips, in particular by picking an (independent) generating set of commit-region reliable observables.
}
For each generating commit-region reliable observable, we use a single-\textsc{lom} defined by the back-propagation of that observable that spans \textit{only} the commit region of the window; that is, from $t=t_{\text{prev}}$ to $t_{\text{center}}-1$. This single-\textsc{lom} has closed boundaries on both the forwards and backwards time-boundaries, where the artificial defects on the forwards time-boundary are inferred from the committed time-edges in the first step of the windowed-\textsc{lom} decoder. The outcome of this single-\textsc{lom} decoder can then be used to update the logical Pauli frame in a way that is consistent with the time-like edges committed to in the first step. An independent generating set of commit-region fragile observables can then be assigned random values based on a coin toss to complete the logical Pauli frame assignment.

\added{
This procedure can be simplified somewhat if some of the observables from the first step are already commit-region reliable, or can be made commit-region reliable by combining with measurements that occur in the commit region. In this case, the logical Pauli frame can already be committed to in the first step, and these observables do not need to be decoded again in the second step.
}

Finally, we return to the issue of the decoding of logical \textit{measurements} in the presence of fast resets. Recall that we decode these logical measurements \added{immediately after they occur} so we can apply any gates that are conditioned on the measurement outcome. \added{As a result, the single-\textsc{lom} spans from $t_{\text{prev}}$ (the time up to which the most-recent window has committed to) to the time of the measurement; in this context we call this the commit region of the single-\textsc{lom}, although recall that it only commits to the logical Pauli frame flips of the measured observable and not to any time-like edges.} Similarly to in the \textsc{lom} decoder, we need to take into account the possibility that the observable is commit-region fragile but ``correlated'' with another observable that is commit-region fragile.

To handle this, we follow a similar procedure to the one in \cref{sec:fragile}. In particular, each time a measurement $M$ occurs in the circuit, we first check whether the observable $O=\{M\}$ is commit-region fragile or not. If $O$ is commit-region reliable, then we simply decode it using the single-\textsc{lom} decoder and use the previously committed logical Pauli frame to obtain the outcome $O$. If $O$ instead is commit-region fragile, we check whether there exists a product with another, earlier or simultaneous observable $O'$ that is also commit-region fragile, such that $OO'$ is commit-region reliable. If no such $O'$ exists, we uniformly sample an outcome for $O$ and continue. If such an $O'$ does exist, then we decode $OO'$ and use the outcome of $O'$ and the previously committed logical Pauli frame to infer the outcome of $O$.

\subsubsection{\added{Efficiency of the two-step windowed-\textsc{lom}} decoder}\label{sec:window_efficiency}

Now that we have described the action of the two-step windowed-\textsc{lom} decoder, we explain why the decoder is not efficient in general. The issue is with the second-step single-\textsc{lom}s that are obtained by the back-propagation of a multi-logical-qubit Pauli operator. As foreshadowed in \cref{sec:fragile}, there is no general upper bound on the weight of this operator, even for low-depth local circuits. Take, for example, the following repetition-code style circuit
\begin{equation}
    \vcenter{\hbox{\includegraphics{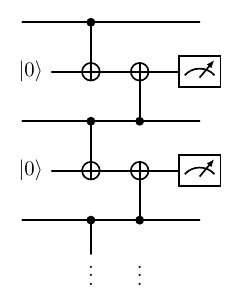}}}
    ,
\end{equation}
which is of constant depth and local in 1D. If we have a window boundary immediately after these logical measurements, the lowest-weight reliable observable that contains $\overline{X}_{1}$ is the product of \textit{all} remaining $\overline{X}$ operators. The volume of the corresponding second-step single-\textsc{lom} therefore scales with the total number of logical qubits in the circuit. \added{On the other hand, it may possible that in certain circuits, such as those obtained by concatenating a surface code with another QEC code, that the weight of these reliable observables can be bounded.} We leave it as an open problem how to solve this efficiency issue with the two-step windowed-\textsc{lom} decoder.

\subsection{Issues with fault-tolerance}
\label{sec:window-ft}

The basic and two-step windowed-\textsc{lom} decoders both suffer from a number of issues relating to its fault-tolerance that we now address in turn. In particular, each issue has the potential to cause the windowed-\textsc{lom} decoder to be unable to correctly decode an error of weight $<d/2$. Intuitively, these arise because we are using the outcomes of a single-\textsc{lom} decoder to infer artificial defects that are then seen by other single-\textsc{lom} decoders, and we bump into issues that are thus similar in nature to what we may encounter for a hierarchical decoder. However, for both the basic and two-step windowed-\textsc{lom} decoders, we are able to solve each of these problems with some appropriate modifications, and we \textit{conjecture} that with these modifications the windowed-\textsc{lom} decoder is able to correct $<d/2$ basic errors. We note that this discussion is \textit{not} the same as asking whether the windowed-\textsc{lom} decoder does, or does not, have a threshold either with or without the modifications, which we leave as open problems.

\added{\subsubsection{Time-like loops}}

First, \added{we explain that the windowed-\textsc{lom} decoder does \textit{not} fail due to the presence of time-like loops, unlike the hierarchical decoder \cref{sec:hier_loops}.}
It is true that time-like loops may be present in each single-\textsc{lom} track, \textit{and} it is also true that edges from the single-\textsc{lom} are used to infer artificial defects in the next window. However, by definition, in any single-\textsc{lom} that is used to commit to time-like edges, any such time-like loop cannot traverse the center of the window from which the artificial defects are inferred, since at this point each single-\textsc{lom} only has support on a single logical qubit $\overline{X}_i$ or $\overline{Z}_i$. Time-like loops therefore do not lead to constant-weight logical decoding errors in the windowed-\textsc{lom} decoder. This is important because the remaining issues that we discuss below are not as serious as that of time-like loops: the smallest number of basic errors required to cause a logical error in the windowed-\textsc{lom} decoder increases with $d$, even if it is less than $d/2$.

\begin{figure}[h!]
    \includegraphics[width=0.9\linewidth]{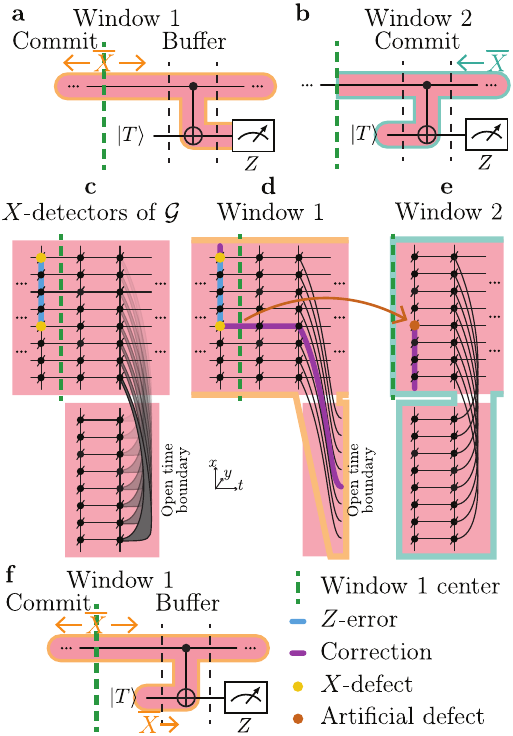}
    \caption{(a) An example of a fragile time-boundary in the forwards direction, in which a logical $X$ operator at the window boundary propagates forwards to the logical $Z$ measurement that it anti-commutes with. Since the window boundary of the windowed-\textsc{lom} decoder is too close (less than $d/2$ time-edges away) to the logical $Z$ measurement, a low-weight error can subsequently cause a logical error in a single-\textsc{lom} decoder in window~2, shown in (b). (c) An example of a weight-4 error in the $d=9$ surface code that will cause a logical error. In general the weight of the error is $d/4+O(1)$. (d) The correction that the single-\textsc{lom} decoder shown in (a) makes involves matching to the fragile time-boundary, causing an artificial defect in window~2. (e) The correction of this artificial defect in a single-\textsc{lom} in window~2 that leads to a logical error. Note that in window~2 we backpropagate the logical $X$ and therefore have no support on this fragile-time boundary. (f) A modification to the single-\textsc{lom} in (a) avoids this issue by ``peeling'' the single-\textsc{lom} off the fragile time-boundary and onto the closed time-boundary arising from the fault-tolerant state preparation of $\ket{\overline{T}}$.}\label{fig:window_fragile}
\end{figure}

\added{\subsubsection{Fragile time-boundaries}}

In contrast, the presence of fragile time-boundaries can cause an issue for the windowed-\textsc{lom} decoder because they can cause the time-like edges to be incorrectly inferred, as shown in \cref{fig:window_fragile}. Because we propagate operators forwards \textit{and} backwards in the windowed-\textsc{lom} decoder, this problem can even arise in the basic windowed-\textsc{lom} decoder \added{even though resets are slow} due to fragile time-boundaries on \textit{measurements}. In the worst case, this can allow errors of weight $d/4+O(1)$ to cause a logical error in the windowed-\textsc{lom} decoder. Note that this is an entirely separate issue to decoding fragile \textit{observables} (which are not present in the basic windowed-\textsc{lom} decoder due to the slow resets) because this problem relates to inferring the committed time-like edges instead of the committed logical Pauli frame.

Solving this for the windowed-\textsc{lom} decoder is simple: we simply need to ensure that every committed time-edge is $>d/2$ time-like edges away from any fragile boundary in the single-\textsc{lom}. One way to guarantee this is by synchronizing all resets to occur only at the first time-step after each window boundary, and all measurements at the last time step before each window boundary. That way, all resets and measurements will be $>d/2$ away from the center of every single-\textsc{lom} that contains it. Of course, this is not necessary for the slow resets in the basic windowed-\textsc{lom} decoder.

However, it is possible in some cases to be more flexible than this, in particular in the case of $T$-injection circuits. When a single-\textsc{lom} involves propagating an $\overline{X}$ operator forwards through the $T$-injection circuit, it can have support on the fragile time-boundary of the $Z$-measurement, as shown in \cref{fig:window_fragile}(a). \added{This can lead to errors of weight $d/4+O(1)$ causing a logical error, as shown in \cref{fig:window_fragile}(b--e). The solution to this is to include an $\overline{X}$ operator arising from the $\ket{T}$ state preparation in the observable that we are forward-propagating through the circuit, as shown in \cref{fig:window_fragile}(f).}
This way, the single-\textsc{lom} is ``peeled'' off the fragile time-boundary, instead depending on the non-fragile time-boundary resulting from the fault-tolerant preparation of the $\ket{\overline{T}}$ state. This peeling-off idea however does not work for arbitrary mid-circuit measurements, necessitating the synchronization of the remaining measurements and resets.

\added{\subsubsection{Time-like snakes}}

The final issue that we discuss is the circuit-dependent occurrence of time-like \textit{snakes}. Intuitively, the issue arises because a low-weight spatial error in one window can be decoded multiple times by different single-\textsc{lom}s, each of which leaves behind a pair of defects for the following window to decode---we call these \textit{left-over} defects as they may contain both artificial defects, which are due to the committed correction, as well as real defects that have not yet been corrected by a committed correction. 

If all these left-over defects appear in the decoding graph of the same single-\textsc{lom} in the second window, then the minimum-weight correction is no longer guaranteed to be low-weight. It is therefore possible that an error of weight $<d/2$ causes a logical decoding error. Actually realizing such a bad example also relies upon the pairs of left-over defects being ``far away'' in the second window; that is, the shortest number of time-like edges required to travel from one pair of left-over defects to another is long. For this to happen, the circuit needs to be such that there is a ``snake'' in the structure of the single-\textsc{lom}. We discuss these issues in more detail in \cref{sec:snakes}, with two explicit examples. In \cref{sec:simple_snake} we show a relatively small example in which the weight of the physical error scales as $2d/5+6$, although the smallest-distance example where the physical error is of weight strictly less than $d/2$ is only $d=65$. Meanwhile, in \cref{sec:complicated_snake}, we show an asymptotically worse example where the physical error weight scales sublinearly with $d$. Both of these examples motivate a solution to the problem of snakes.

Our solution to address the issue of snakes is to modify the decoding subgraph in each single-\textsc{lom} by adding ``short-cut'' edges between vertices that represent different logical qubits (defined precisely in \cref{sec:short-cut_edges}). Intuitively, the short-cut edges remove the problem of snakes because there is always a short time-like edge that links pairs of left-over defects. Moreover, each short-cut edge always connects vertices with the same spatial coordinates, and therefore does not inadvertently create other low-weight logical errors. We conjecture that the short-cut edges, together with the above ``synchronization'' of resets and measurements, are enough to guarantee that the windowed-\textsc{lom} decoder can correct all basic error patterns of weight $<d/2$. \added{It is worth noting, however, that these short-cut edges may reduce performance in practice, particularly if the issue of time-like snakes does not appear much in real circuits.} We leave theoretical and numerical investigations of thresholds to future work.

\section{Discussion}
\label{sec:discuss}

In this work we have made substantial progress towards an efficient, fault-tolerant matching-based decoder for algorithms implemented in the rotated surface code with fast transversal Clifford gates, fast $T$-injections, fast resets and fast measurements. The \textsc{lom} decoder constitutes a matching-based decoder that can correct all basic error patterns of weight $<d/2$, and numerically performs well under repeated- and random-Clifford gate circuits. The basic windowed-\textsc{lom} decoder adds to this by being efficient, but comes at the cost of slowing down the resets and synchronizing the non-$T$-injection measurements in the circuit. Although the basic and/or two-step windowed-\textsc{lom} decoders may perform well in practice, we still leave open the fundamental question: is it possible to design an efficient matching-based decoder that does not require these slow resets? 

Apart from this general open question, there are also a number of other possible directions to further study (windowed) \textsc{lom} decoders. 

We have shown that the logical performance of the \textsc{lom} decoder is improved by first running a belief-propagation algorithm on the full decoding hypergraph before executing the \textsc{lom} decoder. Belief propagation updates the weights of all hyperedges taking into account the correlations as it uses the \textit{full} decoding hypergraph, \added{such as correlations from $Y$ errors}. \added{However, there could be faster and more efficient methods than belief propagation for providing correlation information to the \textsc{lom} decoder.}

It is natural to consider whether the \textsc{lom} decoder applies when more than one $\overline{S}$ and/or $\overline{\mathrm{CNOT}}$ gate is performed between QEC rounds. In this more general setting, it is possible for weight-4 hyperedges to appear in the decoding hypergraph $\mathcal{G}$ that can remain weight-4 in the decoding sub(hyper)graph $G_O$ for an observable $O$, even for the basic error model, hence not reducing the problem to matching, as we discuss in Section~III from SM~\cite{supp}. 

Although we have focused on the unrotated surface code, the \textsc{lom} decoder can be also applied to rotated surface and color codes (with an appropriate matching-based color code decoder, see e.g.~Ref.~\cite{lee2025color}). An interesting question would be to determine the minimum number of circuit-level errors that cause a logical error for the \textsc{lom} decoder applied to these codes, to see if the exact circuit distance is achieved or not. Recent work~\cite{chen2024transversal} has shown a more efficient implementation of the $\overline{S}$ gate for the rotated surface code: the logical gate is implemented inside the syndrome extraction circuit, which changes the structure of the decoding hypergraph. 
Further work is thus required to numerically decode this implementation with the \textsc{lom} decoder.

Multiple future directions are necessary also to evaluate the practicality of the windowed-\textsc{lom} decoder. It is likely that one can extend the windowed-\textsc{lom} decoders to parallel windowed-\textsc{lom} decoders to improve the speed of real-time decoding and to provide a solution to the back-log decoding problem~\cite{skoric2023parallel,window_tan2023, riverlane:aps}. Numerical simulations are also needed to study the existence of a threshold and to compare the performance of the windowed decoder variants. Moreover, we have considered $T$-state preparation as a black box in this work; investigating the interplay of the windowed-\textsc{lom} decoder with an actual magic state distillation protocol is important future work.

On a fundamental level, the fault-tolerance issue of time-like snakes that we uncovered for the windowed-\textsc{lom} decoder, and for which we proposed the short-cut edge solution, has a somewhat fundamental character which could also come up for other fast-logic decoding strategies and different codes. The issue arises because multiple independent observers (i.e.~decoders) find and commit to causes (i.e.~corrections) for events (i.e.~defects) in overlapping bounded regions of space-time. The observers pass on information about these commitments to a following observer (i.e.~leaving left-over defects), with the following observer seeing a subsequent space-time region. Since the commitments of the previous observers were independent, and due to the causal structure in the subsequent space-time region, the next observer, given the previous commitments it has to accept, may deduce a cause that is less likely than the actual cause, leading to logical decoherence. More information thus needs to be passed from the independent multiple observers to the next observers about ``what could have happened before'' and our proposed short-cut edge solution is a blanket expression of this.
Phrased in this way, it is clear that the issue is not about the nature of the defects or edges versus hyperedges\added{, but of a more fundamental character, related to independent space-time bounded observers}.
However, notably, decoding strategies in which independent observers commit to causes in non-overlapping regions of space-time do not suffer from this issue: for example, a windowed parallel hyperedge decoder where space-time is split up into non-overlapping commit regions, that is followed by another parallel decoding step which completes the set of commitments to a full commitment~\cite{skoric2023parallel}. Since our aim is to construct a general windowed efficient {\em matching} decoder for fast logic, we are led (or perhaps even forced) into using independent observers which commit to causes in overlapping space-time regions. 



\section*{Data and software availability}
The data and code used to obtain the numerical results in this work are publicly available in~\cite{data-project}. As mentioned in the main text, the code uses two Python libraries developed during this project: \texttt{surface-sim}~\cite{surface_sim}, for building the encoded circuits, and \texttt{lomatching}~\cite{lomatching}, containing the \textsc{lom} decoders. 

\section*{Acknowledgments}
This work is supported by QuTech NWO funding 2020-2028 – Part I “Fundamental Research”, project number 601.QT.001-1, financed by the Dutch Research Council (NWO), and the OpenSuperQPlus100 project (no. 101113946) of the EU Flagship on Quantum Technology (HORIZON-CL4-2022-QUANTUM-01-SGA). We thank Boris Varbanov and Zherui Wang for initial discussions on decoding fold-transversal gates for the surface code. We thank Timo Hillmann, Alex Kubica, Mark Turner, and the authors of~\cite{cain+:upcoming} for interesting discussions. We acknowledge the use of the DelftBlue supercomputer~\cite{delftblue} for decoding the experiments. 

M.S.P. and M.H.S. contributed equally to this work, with M.S.P. focusing on numerics and M.H.S. focusing on theory, and B.M.T. provided motivation, writing, and supervision.

\appendix

\section{Hierarchical matching decoding with fragile time boundaries}\label{sec:frag-hier}

\begin{figure}[htb]
     \includegraphics[width=0.9\linewidth]{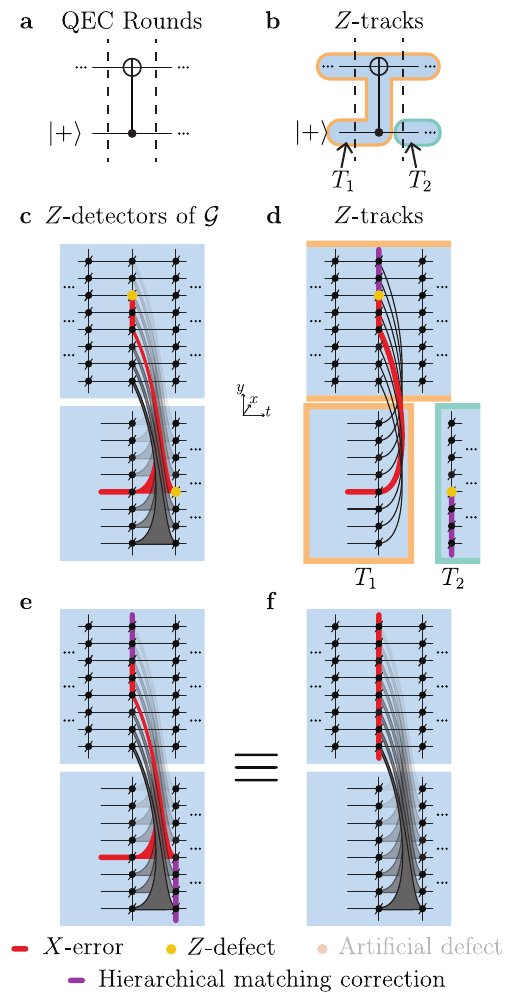}
     \caption{An example of a failure of hierarchical matching due to the presence of a fragile time-boundary in the first track. (a-e) The circuit, tracks, error, action of the hierarchical matching decoder, and the final correction for the error respectively, similar as in \cref{fig:hierarchical_example}. The error$+$correction in (e) can be multiplied by space-time stabilizers to the error string in (f), which corresponds to a logical $X$-error on the top qubit in (a). The physical error here has weight $4<9/2=d/2$. Note also that although we have suggested that the circuit extends in both time directions in the figure, this example still works if, say, the first qubit is initialized in $\ket{0}$.}\label{fig:hierarchical_fragile}
\end{figure}

The hierarchical matching decoder defined in Section \ref{sec:hier} can run into issues when decoding around fragile time boundaries that are {\em completely separate} to the issues discussed with the $\textsc{lom}$-decoder and fragile observables in \cref{sec:fragile}. The issue in the hierarchical matching decoder occurs when the fragile time boundary is \textit{not} decoded in the last track $T_L$. When this is the case, one can construct errors of weight $d/4+O(1)$ that cause a logical error on an adjacent logical qubit. This is shown in \cref{fig:hierarchical_fragile} for a weight-4 error in the $d=9$ surface code. Intuitively, this happens because a single defect is created close to the fragile time boundary in the full hypergraph $\mathcal{G}$ but is decoded only in a later track that cannot be matched to the fragile time boundary, causing a high-weight spatial correction instead. This problem has a simple solution: simply ensure that all fragile time-boundaries are contained in the last track $T_{L}$. Indeed, when the fragile time boundary in \cref{fig:hierarchical_fragile} is included in track $T_{2}$ instead of $T_{1}$, the hierarchical matching decoder indeed corrects the error \cref{fig:hierarchical_fragile}(c) without a logical error.

\begin{figure*}
    \includegraphics[width=0.9\textwidth]{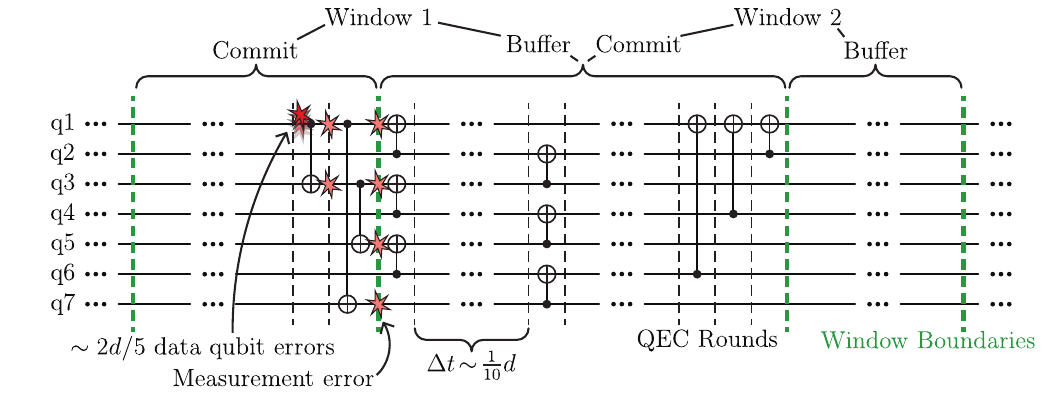}
    \caption{A circuit $\overline{\mathcal{C}}$ with 7 logical qubits, each encoded in the unrotated surface code of distance $d$. An error of weight $\sim 2d/5$ can cause a logical error when decoded by the windowed-\textsc{lom} decoder, in the absence of short-cut edges. The error occurs at seven locations in the circuit: the first set of errors consists of $\sim 2d/5$ $X$-errors on data qubits immediately before the first CNOT gate, while the remaining errors are single measurement errors on $Z$-stabilizer measurements during a QEC round. No other physical errors occur. The precise configuration of the $X$ errors is shown in~\cref{fig:window_snake_2}. The error will be decoded in two windows (1 and 2) of width $>d$, whose commit and buffer regions are labeled above the circuit. We will show in~\cref{fig:window_snake_2,fig:window_snake_3,fig:window_snake_4} that a logical error will be inferred for the $\overline{Z}_{1}$ operator at the end of the circuit.}\label{fig:window_snake_1}
\end{figure*}

\section{Time-like snakes in the windowed-\textsc{lom} decoder}
\label{sec:snakes}

In this appendix we go into detail about the issue of time-like snakes, short-cut edges, and the conjecture that the windowed-\textsc{lom} decoder can correct up to $d/2$ basic errors. The discussion applies both to the one-step and two-step windowed-\textsc{lom} decoders. We begin in \cref{sec:simple_snake} by showing a relatively simple example of how a time-like ``snake'' can allow a low-weight error (still of weight $\Theta(d)$) to lead to a logical decoding error. Then, in \cref{sec:complicated_snake} we show a more elaborate example where the weight of the error grows sublinearly in $d$, further motivating a need to solve this issue. We do so in \cref{sec:short-cut_edges}, where we describe more precisely the ``short-cut edges'' that we use to break-up the snake. The short-cut edges only connect vertices with the same spatial coordinates. Intuitively, these work by ensuring that any pair of defects within the same single-\textsc{lom} can be matched by a string of edges of weight depending only on the space-time coordinates of the defects, and not on the rest of the logical structure of the circuit, as shown in Lemma \ref{lem:short-cut_edge_metric}. The short-cut edges thus change the space-time metric, making it more logically-trivial like in a memory experiment.


Of course, this discussion does not immediately imply the existence of a threshold. Indeed, the windowed-\textsc{lom} decoder without short-cut edges may still have a threshold even if it is not able to correct some errors scaling sublinearly in $d$; and the windowed-\textsc{lom} decoder \textit{with} short-cut edges does not necessarily have a threshold even though it is able to correct all basic errors with weight $<d/2$. It is therefore an open question whether a threshold exists for the windowed-\textsc{lom} decoder in both variants, and which variant performs better in practical scenarios.

\subsection{Simple example of a time-like snake}\label{sec:simple_snake}

\begin{figure}
    \includegraphics[width=0.9\linewidth]{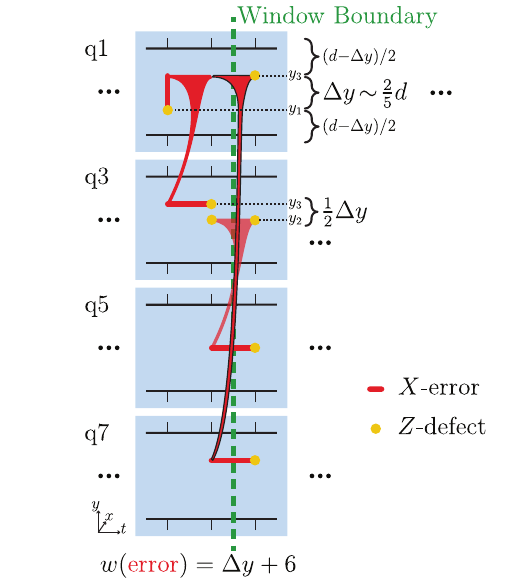}
    \caption{A simplified representation of the decoding hypergraph $\mathcal{G}$ showing the $X$-error pattern of weight $2d/5+6$ in the circuit in~\cref{fig:window_snake_1} that leads to a logical error when decoded by the windowed-\textsc{lom} decoder. Only the slice of the circuit in~\cref{fig:window_snake_1} where the physical errors occur, is shown in the figure. For ease of visualization the space and time coordinates are not to scale; each interval in time represents just one time-like edge.}\label{fig:window_snake_2}
\end{figure}

\begin{figure*}[htb!]
    \includegraphics[width=0.9\textwidth]{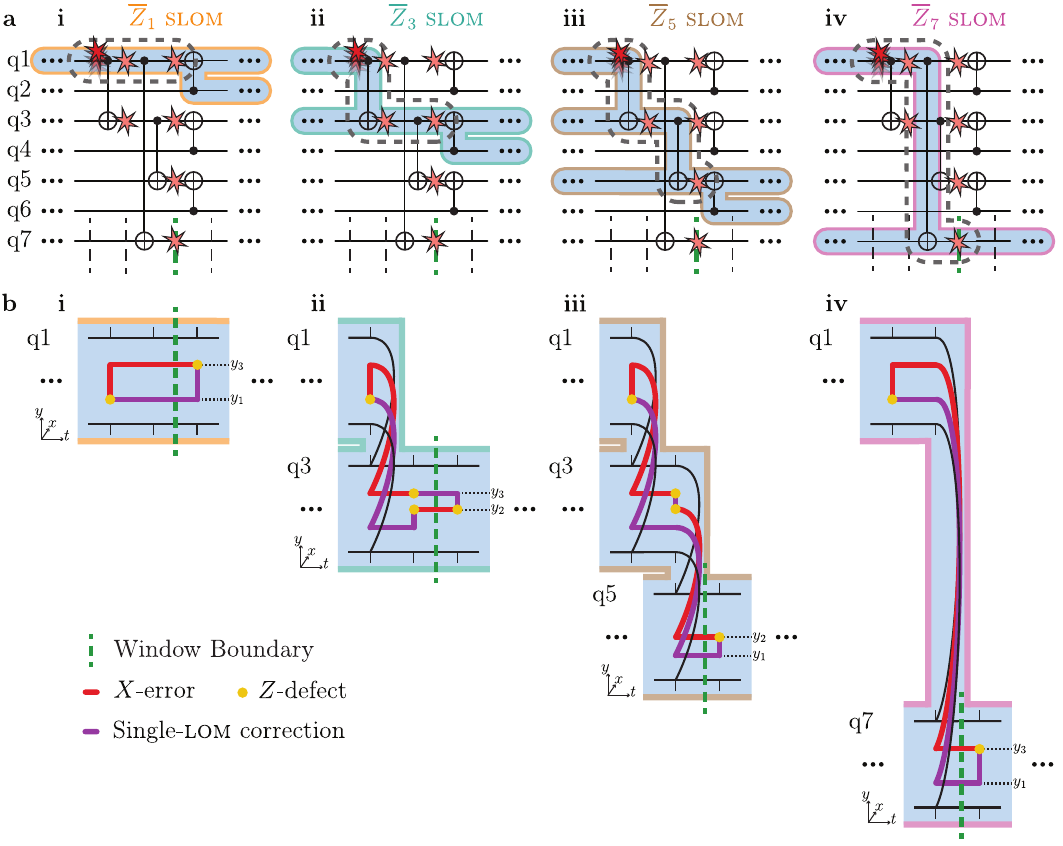}
    \caption{The first window of decoding the error from \cref{fig:window_snake_1,fig:window_snake_2}. We are only interested in four of the iterations of the single-\textsc{lom} (abbreviated \textsc{slom} in the figure), corresponding to the $\overline{Z}$ operators on logical qubits (i) 1, (ii) 3, (iii) 5 and (iv) 7 in the center of the window. (a) The propagation of each single-\textsc{lom} through the circuit. (b) The decoding subgraph of each single-\textsc{lom}, focusing only on the regions containing physical errors, as shown by the gray dotted regions in (a). In order to create a logical error in the second window, the correction in each single-\textsc{lom} decoder must follow the pattern shown in purple. Each correction is one of multiple minimum-weight corrections in each single-\textsc{lom}.}\label{fig:window_snake_3}
\end{figure*}

\begin{figure*}
    \includegraphics[width=0.9\textwidth]{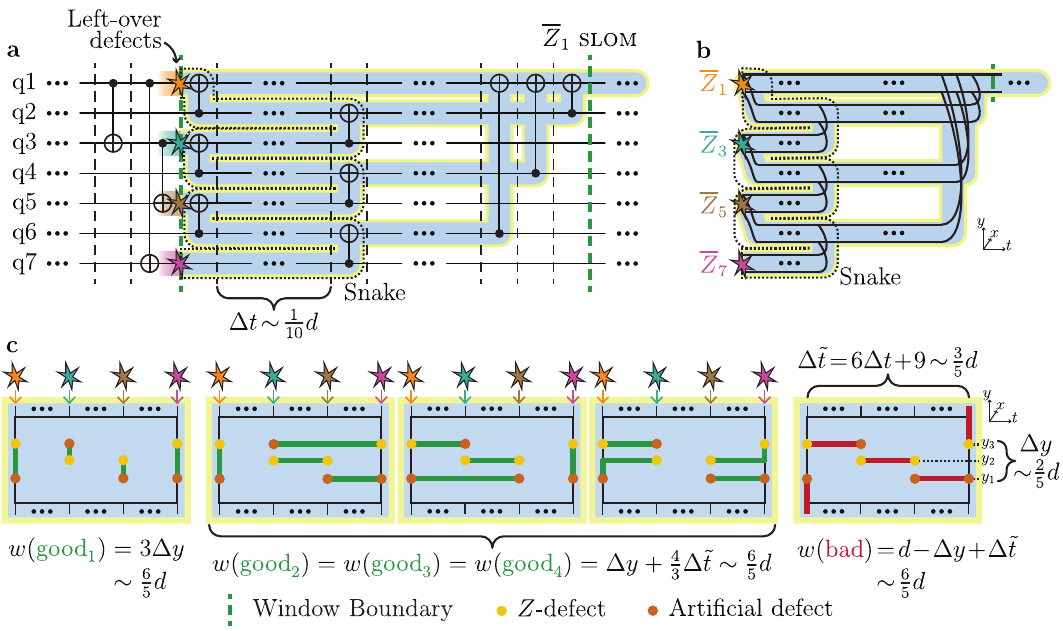}
    \caption{The second window of decoding the errors from \cref{fig:window_snake_1,fig:window_snake_2} along with the left-over defects from the first window of decoding in \cref{fig:window_snake_3}. We are only interested in the single-\textsc{lom} corresponding to $\overline{Z}_{1}$ in the center of the window. (a) The propagation of the $\overline{Z}_{1}$ single-\textsc{lom} through the circuit, with the locations of the left-over defects from the four single-\textsc{lom}s from the first window shown. (b) The decoding subgraph of the single-\textsc{lom}. Focusing only on the ``snake'' part of the subgraph and flattening it out gives the representation in (c), which shows five possible corrections to the left-over defects in the single-\textsc{lom}. The temporal length $\Delta \tilde{t}$ is the length of the snake; note that it includes a constant offset (nine in this example) due to the turns of the snake. The first four corrections shown are good because they correspond to corrections in the same logical class as the physical error that occurred, while the last one is bad because it causes a logical error equivalent to an $\overline{X}_{1}$ error in the middle of the second window. All five corrections have weight $\sim 6d/5$, despite being caused by a physical error of weight only $\sim 2d/5$.}\label{fig:window_snake_4}
\end{figure*}

\begin{figure*}
    \includegraphics[width=0.9\textwidth]{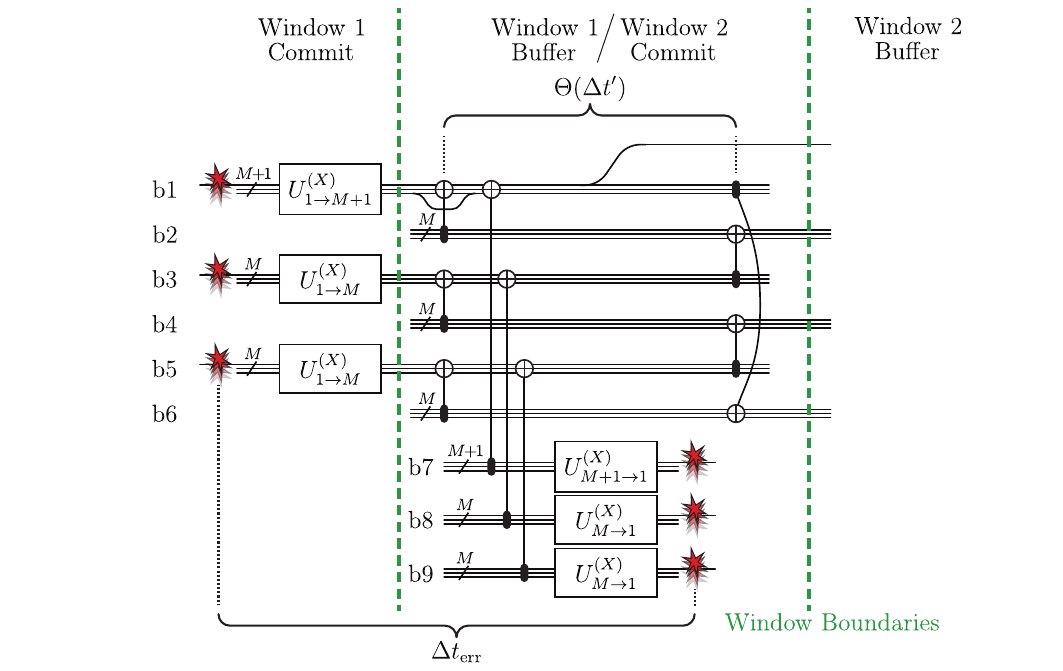}
    \caption{A circuit $\overline{\mathcal{C}}$ with $9M+2$ logical qubits in which an error which scales less than linearly in $d$ can cause a logical error when decoded by the windowed-\textsc{lom} decoder. We represent the logical qubits in nine blocks of $M$ or $M+1$ qubits labeled b1--b9. CNOT gates between the blocks represent independent logical CNOTs between the $m$th qubit in the first and the $m$th qubit in the second block. The $U$-gates are defined in \cref{eq:UX_defn}. Note that we also split off individual logical qubits from block 1 in some of the CNOTs, either from the top of the block or from the bottom. $X$-data-qubit-errors occur in 6 circuit locations, the coordinates of these errors are shown in detail in \cref{fig:window_snake_6}.}\label{fig:window_snake_5}
\end{figure*}

We begin by showing our relatively simple example of how a time-like ``snake'' can allow a low-weight error to lead to a logical decoding error. In particular, in \cref{fig:window_snake_1,fig:window_snake_2,fig:window_snake_3,fig:window_snake_4} we show how an error of weight $2d/5+6$ can lead to a logical error in the windowed-\textsc{lom} decoder. In this example there are multiple minimum-weight corrections that each single-\textsc{lom} decoder can make, one such choice in each single-\textsc{lom} leads to the logical decoding error \footnote{Stronger examples where the \textit{unique} minimum-weight correction in each single-\textsc{lom} leads to a logical error, can also be constructed, but just require another $O(1)$ physical errors.}. For simplicity, we moreover only consider $X$ errors that are spread using CNOT gates, which requires 7 logical qubits in this example; although nothing precludes the possibility that similar structures could be constructed with fewer logical qubits by making use of both $X$- and $Z$-errors spread using $S$ gates.

In this example we consider two windows of decoding, where the artificial defects left behind from the first window of decoding cause a logical Pauli flip in one of the single-\textsc{lom}s in the second window. There are three features of the circuit and error pattern that make this possible:
\begin{enumerate}
    \item the ``proliferation'' of defects: multiple single-\textsc{lom}s in the first window {\em see} the same set of $\Delta y$ spatial errors, i.e.~meaning the defects of these errors occur in the decoding subgraphs of these single-\textsc{lom}s, each leaving behind a pair of artificial defects. This means that in the second window, a correction of weight $>\Delta y$ is required to match the observed defects.
    \item a time-like ``snake'', i.e.~the logical structure of the circuit in the second window implies that the proliferation of defects is seen by a next single-\textsc{lom} decoder. This decoder has the opportunity to match these defects in novel ways, in particular partially using time-like edges (measurement errors). Specifically, we require that any time-like correction that is matching between different pairs of artificial defects (left by different previous decoders) are also long (but not \emph{too} long, as discussed below).
    \item a ``descending staircase'' pattern of defects in the second window, i.e.~some of the artificial defects are left behind at different spatial coordinates, causing a pattern of defects where the bad correction is always as short as the shortest good correction.
\end{enumerate}

We show the structure of the circuit in \cref{fig:window_snake_1} and the placement of the errors in the decoding hypergraph in \cref{fig:window_snake_2}. \cref{fig:window_snake_3} shows how the proliferation of defects occurs through each of the single-\textsc{lom}s in the first window. Note also that each single-\textsc{lom} applies a different correction that will create the descending staircase pattern of defects in the next window. \cref{fig:window_snake_4} then shows the snake in the second window. The key feature of the snake is that the shortest time-like path between pairs of left-over defects in the single-\textsc{lom} shown in \cref{fig:window_snake_4}(a,b) is $\Theta(d)$. To see why this is necessary, we must consider the possible good and bad corrections that the single-\textsc{lom} may make as shown in \cref{fig:window_snake_4}(c). If the snake is too short, then $\Delta\tilde{t}=6\Delta t+9$ will be small, meaning that good corrections 2, 3 and 4 will be shorter than the bad correction. Meanwhile if the snake is too \textit{long} and $\Delta\tilde{t}$ too large, then good correction 1 will be shorter than the bad correction. Therefore only when $\Delta\tilde{t}$ is intermediate in size---roughly $3d/5$---can the bad correction be shorter in length than all the good corrections.

It is an interesting question to determine the smallest error of weight $<d/2$ that can cause a logical decoding error. For the example shown in \cref{fig:window_snake_1,fig:window_snake_2,fig:window_snake_3,fig:window_snake_4}, this is $d=65$, which works with the parameters $\Delta y=26$ and $\Delta t=5$ such that the error has weight $\Delta y+6=32<65/2$ and all the corrections in \cref{fig:window_snake_4}(c) have weight 78. It is therefore unlikely that this particular example will cause practical problems in the surface code, and it may even be the case that below some minimum distance it is impossible for there to be a snake that causes such a problem. However, we have no proof that more elaborate examples do not exist for smaller distances, and it is therefore desirable to find a solution to this problem which can be tested on practical examples later.

\subsection{An asymptotically worse example}\label{sec:complicated_snake}

To construct an example where the weight of the error which causes a logical error scales slower than $\Theta(d)$, we use the more elaborate circuit and set of errors shown in \cref{fig:window_snake_5,fig:window_snake_6,fig:window_snake_7}, which represent a generalization and modification of the smaller circuit analyzed above, in \cref{sec:simple_snake}. The three features of the example in \cref{sec:simple_snake}---the proliferation of defects, the time-like snake, and the descending staircase of defects---are still present. Just like in \cref{sec:simple_snake}, there are two windows in the circuit that are relevant, where the artificial defects from the first window of decoding cause a logical error in the second window. The weight of the physical error depends on the function $f(t)$ that determines the maximum spreading of a Pauli operator in a circuit of depth $t$ as discussed in \cref{sec:window}. In the unrealistic worst-case where there are no spatial restrictions on the gates that can be implemented with $f(t)=2^t$, we find that the weight of the physical error which creates a logical error can scale as slowly as $\Theta\big(\log_{2}(d)\big)$; however, in realistic practical scenarios where $f(t)=\Theta\big(t^{D}\big)$ for some integer $D$, the weight of the physical error instead can scale as $\Theta\big(d^{1/(D+1)}\big)$.

The circuit itself is shown in \cref{fig:window_snake_5} and involves nine {\em blocks} of logical qubits, each of which contains $M$ or $M+1$ logical qubits. Here $M$ is a parameter that will scale as some function of $d$ that we will determine later. The location of the spatial errors and the corrections made by the single-\textsc{lom}s in the first window are shown in \cref{fig:window_snake_6}, and the correction that causes a logical  error in the second window is shown in \cref{fig:window_snake_7}.

To understand how this works, we explain how this circuit enables the proliferation of defects, a time-like snake, and a descending staircase of artificial defects to cause a logical error. The proliferation of defects occurs in the first window of decoding, where the spatial errors on the first qubit in, for example, blocks 1 and 7 are seen by all $M+1$ single-\textsc{lom}s derived from $\overline{Z}$-operators in block 1. Key to this is the formal circuit $U^{(X)}_{1\rightarrow N}$, which we define for any integer $N$ as any Clifford circuit such that
\begin{equation}\label{eq:UX_defn}
   U^{(X)}_{1\rightarrow N} \overline{X}^{\vphantom{(X)}}_{1}U^{(X)\dag}_{1\rightarrow N}=\overline{X}^{\otimes N}.
\end{equation}
We likewise define $U^{(X)}_{N\rightarrow 1}=U^{(X)\dag}_{1\rightarrow N}$. Importantly, any $\overline{Z}_{n}$ operator with $n=1,\dots,N$ anticommutes with $X^{\otimes N}$, and therefore $U^{(X)}_{1\rightarrow N}\overline{Z}_{n}U^{(X)\dag}_{1\rightarrow N}$ anticommutes with $\overline{X}_{1}$. In block 1 with $N=M+1$, this means that every single-\textsc{lom} derived from logical observable $\overline{Z}_{n}$ in block 1 at the center of the first window includes $\overline{Z}_{1}$ prior to the $U^{(X)}_{1\rightarrow M+1}$ gate, and thus {\em sees} the defects produced by the spatial errors on the first logical qubit in block 1. The same argument applied forwards means that every such single-\textsc{lom} also contains the $Z$-detectors for $\overline{Z}_{1}$ at the end of block 7, as shown in \cref{fig:window_snake_6}(a). The depth of the circuit $U_{1\rightarrow N}^{(X)}$ depends on $f(t)$ and, as we will see later, will govern the minimum weight of the physical error that is possible in this example.

Each of the $M+1$ single-\textsc{lom}s in block 1 thus sees the same set of spatial errors, which is shown in \cref{fig:window_snake_6}(b). Of course, each single-\textsc{lom} also has support on other parts of the circuit, but since no errors occur in those locations they are not relevant for us. Of particular importance is the number of QEC rounds between the two sets of spatial errors $\Delta t_{\text{err}}$, shown in \cref{fig:window_snake_6}(a), which depends on
the operator-spreading structure of the circuit $U^{(X)}_{1\rightarrow M+1}$ and hence on $f(t)$. By inspection of \cref{fig:hyperedges}, the number of time-edges separating the two spatial errors in any given single-\textsc{lom} may be larger than the number of QEC rounds separating them, but is always bounded by $\Delta t_{\text{err}} \leq \Delta t_{\textsc{slom}}\leq 2\Delta t_{\text{err}}$. The error shown in \cref{fig:window_snake_6}(b) has weight $2\Delta y$, while the correction has weight $2\Delta t_{\textsc{slom}}$. In order for the indicated correction (in purple) to be chosen in each single-\textsc{lom}, we therefore require $2\Delta y\geq 2\Delta t_{\textsc{slom}}$ for each single-\textsc{lom} in block 1. Alternatively, a simpler sufficient condition is $\Delta y\geq 2\Delta t_{\text{err}}$.
The same construction is repeated for the single-\textsc{lom}s in blocks 3 and 5 as shown in \cref{fig:window_snake_6}(c--d), albeit with the errors happening at different spatial coordinates. In total, the construction uses an error of weight $4\Delta y$, and a sufficient condition for the purple correction to be chosen is \begin{equation}\label{eq:Delta_y_sufficient_condition}
    \Delta y\geq 4\Delta t_{\text{err}},
\end{equation}
where the constant factor of 4 instead of 2 arises because the height of the loops in \cref{fig:window_snake_6}(c--d) is only $\Delta y/2$.

The next ingredient that is required is the presence of a time-like snake in a single-\textsc{lom} in the following window, which we show in \cref{fig:window_snake_7}(a). The most fool-proof way of constructing such a snake is to consider a single-\textsc{lom} that is in the second step of the two-step windowed-\textsc{lom} decoder. The single-\textsc{lom} is propagated backwards from a multi-logical-qubit Pauli operator at the center of window 2 given by a product of $\overline{Z}$ operators on the first qubit of the first block and all qubits in blocks 2, 4, and 6, as shown in \cref{fig:window_snake_7}(a). Such a second-step single-\textsc{lom} can arise when dealing with the fragile observables arising from fast resets and/or measurements in the commit region of window 2 (not shown in explicitly \cref{fig:window_snake_5,fig:window_snake_6,fig:window_snake_7}). We suspect that examples could also be constructed in the basic windowed-\textsc{lom} decoder where all single-logical observables are reliable, by including a unitary gate $H^{\otimes (3M+1)}U^{(X)}_{3M+1\rightarrow 1}H^{\otimes (3M+1)}$ on the first qubit of block 1 and all qubits in blocks 2, 4, and 6 immediately before the center of window 2. In this case, however, it is more complicated to prove \cref{lem:descending_staircase}, even though we suspect the lemma would still be true.

The snake, shown in \cref{fig:window_snake_7}(a), wriggles its way through the different blocks of the circuit, from the first qubit in blocks 1, 2, 3, 4, 5, and then 6, and then back to the second qubit in blocks 1, 2, 3, and so on, until it terminates on the last qubit of block 1. Moreover, the single-\textsc{lom} sees a pair of left-over defects for each of the single-\textsc{lom}s in \cref{fig:window_snake_6}, which are equally-spaced throughout the snake region. Importantly, the shortest time-like matching between defects in left-over defects is through the snake and therefore scales with $\Delta t'$ and $M$.

Finally, the descending staircase of defects can be seen in the flattened representation of the single-\textsc{lom} shown in \cref{fig:window_snake_7}(b), with two example good corrections and one example bad correction to this pattern of defects. We choose $\Delta t'=\Delta y/2$ as indicated in \cref{fig:window_snake_7}(b). This way, the two good corrections shown in \cref{fig:window_snake_7}(b) have the same weight. In fact, the structure of the descending staircase allows us to prove that at least one of three corrections shown in \cref{fig:window_snake_7}(b) is a minimum-weight correction to the defects.

\begin{lemma}\label{lem:descending_staircase}
    For the descending staircase of defects in \cref{fig:window_snake_7}(b) with $\Delta t'=\frac{1}{2}\Delta y$ and $\Delta y<d/2$, every correction to the defects by the multi-Pauli single-\textsc{lom} at the center of window 2 in \cref{fig:window_snake_7}(a) has weight $w$ lower-bounded as
    \begin{equation}
        w\geq \mathrm{min}\big((4M+2)\Delta y/2, d+(3M-2)\Delta y/2\big).
    \end{equation}
\end{lemma}
\begin{proof}
    We will structure this proof by considering the number of times $n_{\text{bdy}}$ the correction matches to the boundary (either the top or bottom spatial boundary). 
    All other edges in the matching must be between the remaining $6M+2-n_{\text{bdy}}$ non-boundary vertices, which is only possible if $n_{\text{bdy}}$ is even. The non-boundary matching graph of these non-boundary vertices $G_{\text{match}}=(V_{\text{match}},E_{\text{match}})$, which contains only the vertices that correspond to (artificial) defects, is shown in \cref{fig:window_snake_7}(c). Note that the number of vertices in $G_{\text{match}}$ is $|V_{\text{match}}|=6M+2$.

    We begin by considering the case where $n_{\text{bdy}}\neq 0$. Since the closest defect to the boundary is $(d-\Delta y)/2$, the weight of each edge in the matching connected to the boundary is at least $(d-\Delta y)/2$. Moreover, each edge in $G_{\text{match}}$ (i.e.~every edge that's \textit{not} connected to the boundary) has weight at least $\frac{1}{2}\Delta y$. Therefore, by counting the number of edges that are and are not connected to the boundary we have
    \begin{equation}\label{eq:w_lb_with_n_bdy}
        w\geq n_{\text{bdy}}d/2+\bigg(3M+1-\frac{3}{2}n_{\text{bdy}}\bigg)\Delta y/2.
    \end{equation}
    If $n_{\text{bdy}}\geq 2$, we can use $\Delta y<d/2$ to write
    \begin{subequations}
    \begin{equation}
        \begin{split}
        w\geq (n_{\text{bdy}}-2)d/2&+\bigg(3M+1-\frac{3}{2}(n_{\text{bdy}}-2)\bigg)\Delta y/2\\&\hspace{2.5 cm}+\bigg(d-\frac{3}{2}\Delta y\bigg)\\
        \geq(n_{\text{bdy}}-2)d/2&+\bigg(3M+1-\frac{3}{2}(n_{\text{bdy}}-2)\bigg)\Delta y/2,
    \end{split}
    \end{equation}
    and therefore we can substitute $n_{\text{bdy}}=2$ into \cref{eq:w_lb_with_n_bdy}, giving
    \begin{equation}
        w\geq d+(3M-2)\Delta y/2,\quad\text{if }n_{\text{bdy}}\geq 2.\label{eq:w_lb_with_bdy}
    \end{equation}
    \end{subequations}

    Now, if $n_{\text{bdy}}=0$, then the entire correction must use only the edges shown in \cref{fig:window_snake_7}(c). Consider each ``descending staircase'' of vertices $V_{\text{stair},m}\subset V_{\text{match}}$ for $m=1,\dots,M$, shown in the dotted regions. Edges of weight $\Delta y/2$ connect pairs of vertices in the same staircase $V_{\text{stair},m}$; to match vertices in different staircases requires a ``thick'' edge of weight $\Delta y$. For each $V_{\text{stair},m}$, there are thick edges to the left, connecting it with $V_{\text{stair},m-1}$, and thick edges to the right connecting it with $V_{\text{stair},m+1}$. Since there are an even number of vertices in each staircase, the number of thick edges to the left and right that are included in the matching must be either both even, or both odd. Moreover, looking at the first descending staircase $V_{\text{stair},1}$, in order for the matching to match the bottom-left vertex there must be an odd number of thick edges to the left that are included in the matching. This means that there must be an odd number of thick edges between every pair of staircases in $G_{\text{match}}$. Therefore, there are at least $M+1$ thick edges in the matching. Since from before there are at least of $3M+1$ edges in the matching, we have
    \begin{equation}\label{eq:w_lb_no_bdy}
        w\geq (M+1)\Delta y+2M\Delta y/2=(2M+1)\Delta y,\quad\text{if }n_{\text{bdy}}=0.
    \end{equation}
    Combining \cref{eq:w_lb_with_bdy,eq:w_lb_no_bdy} proves the result.
\end{proof}

\begin{figure*}[htb!]
    \includegraphics[width=0.9\textwidth]{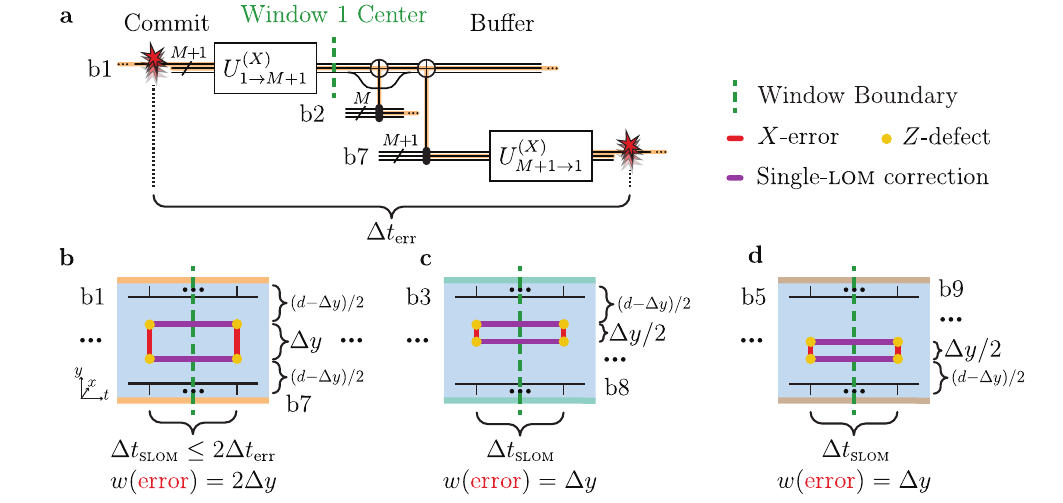}
    \caption{The first window of decoding the error from \cref{fig:window_snake_5}. Each $\overline{Z}_{i}$ single-\textsc{lom} with $i$ representing a qubit in blocks 1, 3, or 5 are relevant for our example. (a) The propagation of an arbitrary block-1-single-\textsc{lom} for $\overline{Z}_n$, $n=1,\ldots M+1$ through the circuit. Similar structures are followed for the other single-\textsc{lom}s in this window. (b--d) The decoding subgraphs of single-\textsc{lom}s in blocks 1, 3, and 5 (respectively), along with the errors and corrections in each single-\textsc{lom}. The artificial defects from each single-\textsc{lom} appear in the second window and cause a logical error in decoding, as shown in \cref{fig:window_snake_7}.}\label{fig:window_snake_6}
\end{figure*}

With these three ingredients in place, we can now determine the minimum asymptotic scaling of the weight of the physical error that causes the logical error shown in \cref{fig:window_snake_5,fig:window_snake_6,fig:window_snake_7}. From \cref{lem:descending_staircase}, the bad correction in \cref{fig:window_snake_7}(b)(iii) will be a minimum-weight correction to the error syndrome if
\begin{equation}
    d\leq (M+4)\Delta y/2,
\end{equation}
or, in asymptotic terms, we require $M\Delta y=\Omega(d)$. Moreover, from the corrections in the first window, we require
\begin{equation}
    \Delta y\geq 4\Delta t_{\text{err}}.\tag{\ref{eq:Delta_y_sufficient_condition}}
\end{equation}
However, $\Delta t_{\text{err}}$ is lower-bounded by the depth of the $U^{(X)}_{1\rightarrow N}$ circuits. In particular, in \cref{sec:window}, we defined the function $f(t)$ as the maximum number of qubits that a single-qubit Pauli operator can be propagated to in a circuit of depth $t$. In general, $f(t)=2^{t}$; however, in a physical setting, locality and efficiency requirements will limit this to $f(t)=\mathrm{poly}(t)$ for some polynomial in $t$. Therefore, the depth of $U^{(X)}_{1\rightarrow N}$ is lower-bounded by $f^{-1}(N)$. Therefore there exist circuits $U^{(X)}_{1\rightarrow M}$ such that $\Delta t_{\text{err}}$ scales as $O\big(f^{-1}(M)\big)$.

\begin{figure*}
    \includegraphics[width=0.9\textwidth]{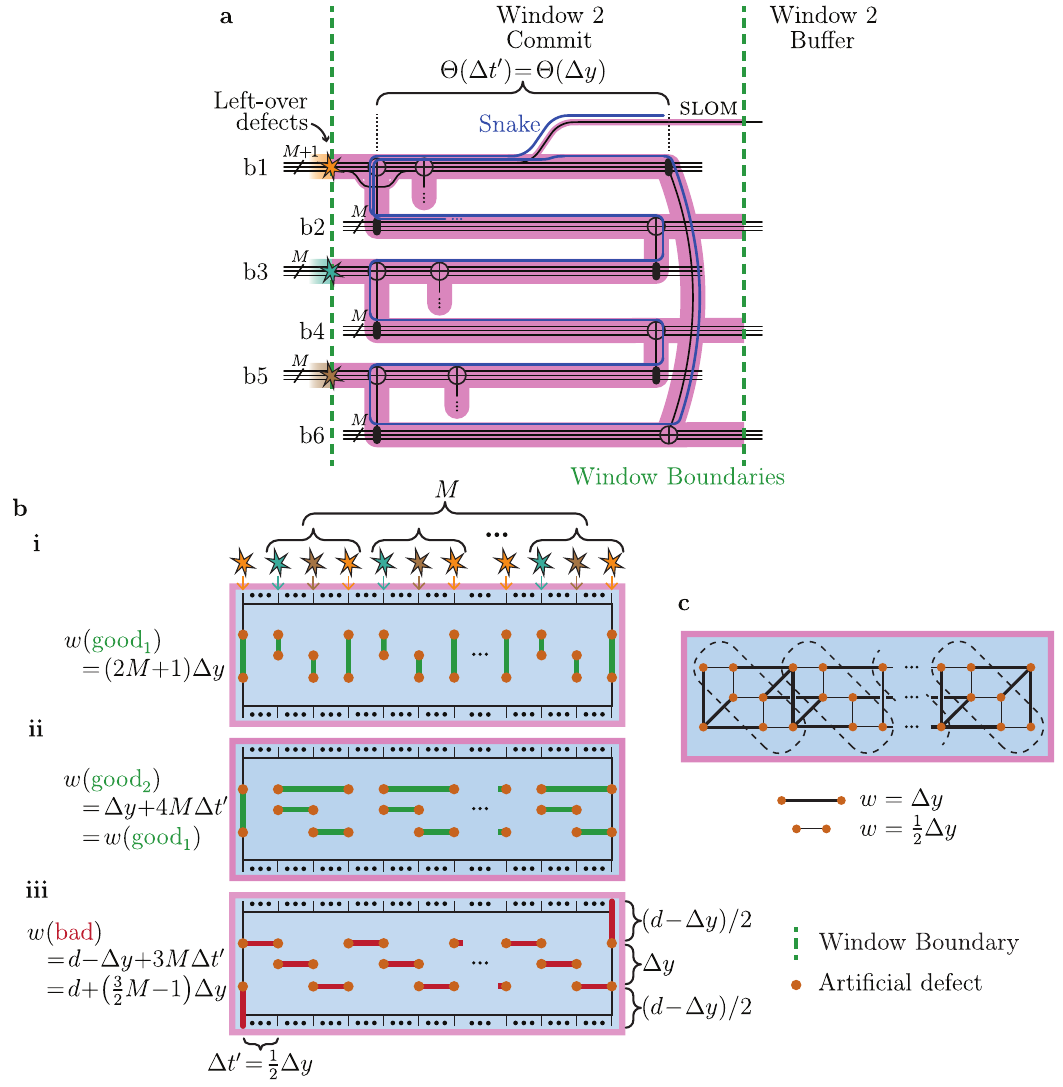}
    \caption{The second window of decoding the left-over defects from the errors and corrections in \cref{fig:window_snake_5,fig:window_snake_6}. We are only interested in the single-\textsc{lom} corresponding to $\overline{Z}_{1}$, shown shaded in (a). We also show in dark blue the structure of the snake, which wriggles from the first qubit in block b1, then to the first qubit in each of the blocks in turn down to b6, before returning to the second qubit in block b1, and so on. On blocks b1, b3, and b5 artificial defects are created at different spatial locations according to \cref{fig:window_snake_6}(b--d). (b) Some good and bad corrections. (c) The matching graph of the artificial defects, with edges to the boundaries removed, used in the arguments in \cref{lem:descending_staircase}. The weight of each edge is indicated by the thickness of the line and follows the Manhattan distance. The dashed ovals indicate each ``descending staircase'' of defects, as used in \cref{lem:descending_staircase}.}\label{fig:window_snake_7}
\end{figure*}

First, consider the worst-case unphysical situation where $f(t)=2^{t}$. In this case we can set
\begin{equation}
    \begin{split}
    M&=\Theta(d), \\ \Delta y&=\Theta(\Delta t_{\text{err}})=\Theta\big(f^{-1}(M)\big)=\Theta\big(\log_{2}(d)\big),
\end{split}
\end{equation}
such that the weight of the error is $\Theta\big(\log_{2}(d)\big)$ while the scaling still satisfies $M\Delta y=\Omega(d)$. In the more physical case where $f(t)=\Theta(t^{D})$ for some dimension $D$, we can likewise set
\begin{equation}
    \begin{split}
    M&=\Theta(d^{D/(D+1)}),\\\Delta y&=\Theta(\Delta t_{\text{err}})=\Theta\big(f^{-1}(M)\big)=\Theta\big(d^{1/(D+1)}\big).
\end{split}
\end{equation}
In both cases, the weight of the error $\Delta y$ which can cause a logical error as in \cref{fig:window_snake_7}, scales sub-linearly with $d$.


\subsection{Short-cut edges}\label{sec:short-cut_edges}


We provide a precise definition of the short-cut edges in either the basic or two-step windowed-\textsc{lom} decoders for which we introduce the following notation. We write the full decoding hypergraph as $\mathcal{G}=(\mathcal{V},\mathcal{H})$, and the decoding subgraph for a single-\textsc{lom} as $G=(V,E)$. Moreover, every basic error $h\in\mathcal{H}$ that overlaps with the single-\textsc{lom} can be projected down to an edge $e\in E$. A subset of edges $s\subseteq E$ has a \textit{vertex boundary} $\partial s\subseteq V$ that includes all vertices that are the endpoints of an odd number of edges in $s$. We write $w(s)=|s|$ for the weight (i.e.~size) of the edge subset $s$.
If a non-empty subset of edges has no more than two endpoints we call it an edge \textit{string}.

Recall that each detector $d(t,j,x,y)$ is labeled with a set of coordinates as in Table~I from SM~\cite{supp}. Here, $t\in\mathbb{Z}$ represents the time of the detector, $j$ represents the logical qubit and $x$ and $y$ are spatial coordinates. $X$-detectors have spatial coordinates satisfying $x\in\mathbb{Z}$ and $y\in\mathbb{Z}+\frac{1}{2}$ while for $Z$-detectors we have $x\in\mathbb{Z}+\frac{1}{2}$ and $y\in\mathbb{Z}$, with $0\leq x,y\leq d-1$.
In the statements below we will find it useful to have a different labeling of vertices $v(t,k,z,\tilde{z})$, given by $v(t,2j-1,x,y+1/2)\equiv d(t,j,x,y)$ for $X$-detectors and $v(t,2j,y,x+1/2)\equiv d(t,j,x,y)$ for $Z$-detectors. This way, $z,\tilde{z}\in\{0,1,\dots,d-1\}$ regardless of whether the vertex represents an $X$- or $Z$-detector.

Now, consider any single-\textsc{lom} instance in the execution of the windowed-\textsc{lom} decoder. In particular, the single-\textsc{lom} decoders runs \textsc{mwpm} on a decoding subgraph $G=(V,E)$. We can moreover assume without loss of generality that $G$ is connected; after all, if $G$ is \textit{not} connected, then we can decode it in two independent single-\textsc{lom}s. For convenience, in each single-\textsc{lom} we introduce a set of \textit{boundary} vertices $V_{\text{bdy}}$. In particular, for each edge in the single-\textsc{lom} that is connected to the boundary vertex $v_{\text{bdy}}$ we define a separate boundary vertex that has its own unique set of $(t,k,z,\tilde{z})$ coordinates. We define the set of boundary vertices as $V_{\text{bdy}}$, which itself is partitioned into three subsets: $V_{\text{top}}$, $V_{\text{bottom}}$ and $V_{\text{time}}$ representing the top spatial boundary ($\tilde{z}=d$), bottom spatial boundary ($\tilde{z}=0$), and open time boundaries respectively. To obtain a valid correction, the correction subset of hyperedges $c\subseteq\mathcal{H}$ must have vertex boundary $\partial c$ that matches the observed defects everywhere except on the boundary vertices $V_{\text{bdy}}$. For the purposes of running \textsc{mwpm} one needs to connect all these boundary vertices to each other by weight-zero edges; however, for the purposes of our proof, we instead equivalently assume that there are no edges between these vertices. In this picture, the boundary vertices are the only vertices of degree one in the decoding graph.

With this nomenclature, all of the standard edges in the decoding subgraph, assuming basic errors, fall in one of the four categories:
\begin{itemize}
    \item $z$-type edges between a pair of vertices $\{v(t,k,z,\tilde{z}),v(t,k,z\pm 1,\tilde{z})\}$;
    \item $\tilde{z}$-type edges between $\{v(t,k,z,\tilde{z}),v(t,k,z,\tilde{z}\pm 1)\}$;
    \item $t$-type edges between $\{v(t,k,z,\tilde{z}),v(t+1,k',z,\tilde{z})\}$ for some pair of indices $k,k'$ that may or may not be equal; or
    \item $k$-type edges between $\{v(t,k,z,\tilde{z}),v(t,k',z,\tilde{z})\}$ for $k\neq k'$.
\end{itemize}
Note that by definition, $z$-type, $\tilde{z}$-type and $k$-type edges connect vertices whose coordinates differ only in a single coordinate (the $z$, $\tilde{z}$ and $k$ coordinates respectively), while $t$-type edges contain vertices that can differ in both their $t$- and $k$-coordinate values. We will also occasionally refer to the $z$- and $\tilde{z}$-type edges collectively as \textit{space-like} edges, and the $t$- and $k$-type edges as \textit{time-like} edges.
The observing edge set in any given single-\textsc{lom} contains only $\tilde{z}$-type edges between a pair of vertices $\{v(t,k,z,d-1),v(t,k,z,d)\}$, where the latter vertex is on the top boundary $v(t,k,z,d)\in V_{\text{top}}$ and the time coordinate $t$ is in the commit region of the single-\textsc{lom}.

With this, we define the short-cut edges that we add to each single-\textsc{lom} during the execution of the circuit as follows:
\begin{definition}[Short-cut edges]
    \,
\begin{itemize}
    \item \emph{($k$-type short-cut edges:)} for each $t,z,\tilde{z}$, add edges $\{v(t,k,z,\tilde{z}),v(t,k',z,\tilde{z})\}$ for all $k,k'$ that are in the single-\textsc{lom} at time $t$ and such that neither vertex is a boundary vertex (i.e.~$v(t,k,z,\tilde{z}),v(t,k',z,\tilde{z})\notin V_{\text{bdy}}$), and
    \item \emph{($t$-type short-cut edges:)} for each $t,z,\tilde{z}$, add edges $\{v(t,k,z,\tilde{z}),v(t+1,k',z,\tilde{z})\}$ for every $k$ that is in the single-\textsc{lom} at time $t$ and for every $k'$ that is in the single-\textsc{lom} at time $t+1$ and such that neither vertex is a boundary vertex.
\end{itemize}
\end{definition}
In words, this means that for every single-\textsc{lom}, there is a (real or short-cut) $k$-type edge between \textit{every} pair of vertices in the single-\textsc{lom} that has the same time and space coordinates, even if they differ by logical qubit and/or the Pauli type of the detector; \textit{and} there is a (real or short-cut) $t$-type edge between every pair of vertices that has the same space coordinates and time coordinates that differ by one.

The short-cut edges are important when running \textsc{mwpm} on the single-\textsc{lom}, since without them low-weight errors can cause logical decoding errors as discussed in \cref{sec:snakes}. However, the short-cut edges do not directly change what is committed to in an individual single-\textsc{lom}. For one, since none of the short-cut edges have any spatial component, they are not in the observing edge set and therefore cannot directly flip the inference of a logical outcome. And second, they also do not directly lead to any artificial defects. To understand this second point, note that any pair $\{v_{1},v_{2}\}$ of vertices that are the endpoints of a short-cut edge can always be matched using regular $t$-type and/or $k$-type edges. Whenever the correction uses a short-cut edge, we can therefore interpret the short-cut edge instead as a matching with regular edges. However, this matching is designed to be of lower weight than when regular edges are used, hence they represent short-cuts. Moreover, due to the structure of any first-step single-\textsc{lom}, this matching can never traverse the center of the single-\textsc{lom} window, so it therefore does not directly change the artificial defects in the following window.

The short-cut edges are necessary for the following.
\begin{conjecture}
    Assuming a basic error model, a basic or two-step windowed-\textsc{lom} decoder defined in \cref{sec:window} and adapted to include the short-cut edges, applied to circuits defined in \cref{sec:ft-surface} with synchronized resets and measurements and implemented in the unrotated surface code, can correct any basic error of weight $<d/2$.
    \label{conj:FT}
\end{conjecture}

Intuitively, the reason we believe this works is because short-cut edges remove the possibility of snakes occurring in any given single-\textsc{lom}. In particular, recall that in \cref{sec:simple_snake,sec:complicated_snake}, we explained that for the counterexamples to work, the snake had to have a particular length in order for the bad correction to be shorter than the good corrections. If the snake was too short, or too long, then at least one of the good corrections would be shorter than the bad correction. The short-cut edges work by essentially making sure the snake is as small as possible, effectively removing it from the single-\textsc{lom}. For example, in \cref{fig:window_snake_4,fig:window_snake_7}, the different pairs of left-over defects are separated by a set of $\Theta(d)$ time-like edges in the single-\textsc{lom} in the second window. The $k$-type short-cut edges here act to reduce this time-like distance to 1, so that good corrections 2, 3, and 4 in \cref{fig:window_snake_4}(c), and good correction 2 in \cref{fig:window_snake_7}(b), are always the minimum-weight correction. On the other hand, one might worry that the short-cut edges might ``help'' cause a logical error that is of weight $<d/2$. However, this does not happen because \textit{none} of the short-cut edges have any spatial component; that is, they always connect vertices with the same $z$ and $\tilde{z}$ coordinates. We hope to include a proof of \cref{conj:FT} either in a later version of this manuscript or in a future publication.

One helpful lemma that highlights the importance of the short-cut edges is the following:
\begin{lemma}\label{lem:short-cut_edge_metric}
    For any decoding subgraph $G=(V,E)$ for a single-\textsc{lom}, and for any pair of vertices $v(t,k,z,\tilde{z}),v(t',k',z',\tilde{z}')\in V\setminus V_{\text{bdy}}$, the minimum-weight edge string $s$ with vertex boundary $\partial s=\{v(t,k,z,\tilde{z}),v(t',k',z',\tilde{z}')\}$ is of weight
    \begin{equation}\label{eq:short-cut_edge_metric}
        w(s)= |z-z'|+|\tilde{z}-\tilde{z}'|+|t-t'| +\delta_{t,t'}(1-\delta_{k,k'}).
    \end{equation}
    Moreover, the shortest edge string $s$ with vertex boundary $\partial s=\{v(t,k,z,\tilde{z}),v_{\text{top}}\}$ for some $v_{\text{top}}\in V_{\text{top}}$ has weight $w(s)=d-\tilde{z}$, and the shortest edge string $s$ with $\partial s=\{v(t,k,z,\tilde{z}),v_{\text{bottom}}\}$ for some $v_{\text{bottom}}\in V_{\text{bottom}}$ has $w(s)=\tilde{z}$.
\end{lemma}
Importantly, the weight \cref{eq:short-cut_edge_metric} does \textit{not} depend on the logical structure of the single-\textsc{lom}. Therefore, if the same pair of defects arise in two different single-\textsc{lom}s, they can be matched with a correction of the same weight in each single-\textsc{lom}.
\begin{proof}
    We prove \cref{lem:short-cut_edge_metric} as follows. We begin with the edge string between a pair of vertices. If $t\neq t'$ then assume without loss of generality $t<t'$. Since every (real or short-cut) edge has weight 1 and increments at most one coordinate $z,\tilde{z}$ or $t$ by 1, it is clear that every such edge string has weight lower-bounded by $|z-z'|+|\tilde{z}-\tilde{z}'|+|t-t'|$. We can construct such a string beginning with a string of $t$-type edges $s_{\text{time}}=\{v(t_{\text{int}},k(t_\text{int}),z,\tilde{z}),v(t_{\text{int}}+1,k(t_\text{int}+1),z,\tilde{z})\}$ for $t_{\text{int}}=t,t+1,\dots,t'-1$. This must have the property that $k(t)=k$, $k(t')=k'$, and $k(t_\text{int})$ is in the single-\textsc{lom} at time $t_\text{int}$---such a choice of $k(t)$ is always possible since there is at least one value of $k$ in the single-\textsc{lom} at every time-step. This string of $t$-type edges $s_{\text{time}}$ is then composed with the space-like edge string $s_{\text{space}}$ from $v(t',k',z,\tilde{z})$ to $v(t',k',z',\tilde{z}')$ which is of weight $|z-z'|+|\tilde{z}-\tilde{z}'|$, so that in total $w(s)=w(s_{\text{time}})+w(s_{\text{space}})=|z-z'|+|\tilde{z}-\tilde{z}'|+|t-t'|$.
    
    If $t=t'$ and $k=k'$ then the statement is clear simply by constructing a space-like string. If $k\neq k'$ then clearly the minimum weight of such a string must be at least $|z-z'|+|\tilde{z}-\tilde{z}'|+1$, since we must use at least one $t$-type or $k$-type edge to move from $k$ to $k'$. We can construct such an edge string $s$ with the short-cut edge $\{v(t,k,z,\tilde{z}),v(t,k',z,\tilde{z})\}$ and the space-like edge string from $v(t,k',z,\tilde{z})$ to $v(t,k',z',\tilde{z}')$ which is of weight $|z-z'|+|\tilde{z}-\tilde{z}'|$, so that in total $w(s)=|z-z'|+|\tilde{z}-\tilde{z}'|+1$.

    The minimum weight strings connecting a vertex with the top or bottom spatial boundary with the correct weight are clear since we can simply construct a space-like string that does not contain any $t$- or $k$-type edges.
\end{proof}

\FloatBarrier


\clearpage
\section*{Supplemental Material for ``Decoding across transversal Clifford gates in the surface code''}

\section{Proof of Lemma 1}
\label{sec:proppauli}

Here, we prove Lemma~1; that is, that $S^{\rightarrow}$ and $O$ anticommute if and only if $S$ and $O^{\leftarrow}$ anticommute.

\begin{proof}
We first define the observing region backpropagated only to time $t$, denoted $O^{t\leftarrow}$. More specifically, $O^{t\leftarrow}$ contains an $X$ (resp.~$Z$) before every $Z$-measurement ($X$-measurement) included in the observable $O$, \textit{and} these Pauli operators are backpropagated using rule 2 only for times $t'\geq t$. Note that this means that $O^{t\leftarrow}$ may have non-trivial components earlier than time $t$ if there are measurements in $O$ that occur before that time $t$. Under this definition, $O^{\infty\leftarrow}=O$, while $O^{0\leftarrow}=O^{\leftarrow}$. We similarly define the reset stabilizing region propagated forwards to time $t$ and denote it $S^{\rightarrow t}$.

With these definitions, we now wish to show that, for any $t$, $S^{\rightarrow t}$ and $O^{t\leftarrow}$ anticommute if and only if $S^{\rightarrow (t+1)}$ and $O^{(t+1)\leftarrow}$ anticommute.
Consider the propagation of the operator $O^{(t+1)\leftarrow}$ backwards one step in time to $O^{t\leftarrow}$: the components of its Pauli regions don't change for any circuit locations $(t',j)$ with $t'\neq t,t+1$; that is, $O^{(t+1)\leftarrow}_{(t',j)}=O^{t\leftarrow}_{(t',j)}$ for all $t'\neq t,t+1$. Let the set of qubits that do not participate in a reset or measurement at time $t+1/2$ be denoted $J$, and let the corresponding $|J|$-qubit Clifford gate applied to those qubits be denoted $C_{J}$. For all qubits $j\notin J$, $O^{t\leftarrow}_{(t,j)}=O^{(t+1)\leftarrow}_{(t,j)}$ and $O^{t\leftarrow}_{(t+1,j)}=O^{(t+1)\leftarrow}_{(t+1,j)}$ since rule 2 of the backpropagation algorithm doesn't apply to those circuit locations\added{; in other words, measurements and resets on qubits outside $J$ do not affect the anti-commutation relations during time-step propagation since these qubits are not involved in the Clifford operation}. The same is true when comparing $S^{\rightarrow t}$ and $S^{\rightarrow(t+1)}$. With this, we can trivially say that for every circuit location $(t',j)$ with either $t'\notin\{t,t+1\}$ or $j\notin J$, $[O^{t\leftarrow}_{(t',j)},S^{\rightarrow t}_{(t',j)}]=[O^{(t+1)\leftarrow}_{(t',j)},S^{\rightarrow (t+1)}_{(t',j)}]$. Now, note that $S^{\rightarrow t}_{(t+1,j)}=I$ for all $j\in J$, and likewise $O^{(t+1)\leftarrow}_{(t,j)}=I$ for all $j\in J$. Moreover, by definition we have $O^{t\leftarrow}_{(t,J)}=C_J^{\dag}O^{(t+1)\leftarrow}_{(t+1,J)}C_J$ and $S^{\rightarrow t}_{(t,J)}=C_J^{\dag}S^{\rightarrow(t+1)}_{(t+1,J)}C_J$. So, the number of circuit locations on which $O^{t\leftarrow}$ and $S^{\rightarrow t}$ anticommute must be the same as the number of circuit locations on which $O^{(t+1)\leftarrow}$ and $S^{\rightarrow (t+1)}$ anticommute, from which the claim follows.
\end{proof}

\section{Detector frames}
\label{sec:frames-errors}

In this appendix we provide a precise mathematical definition of the fold-transversal gates and the pre-gate and post-gate frames.

\subsection{Qubit coordinates and fold-transversal gates}
\label{sec:coordinates}

We label each qubit in the unrotated surface code, see Fig.~1 for $d=3$, with spatial coordinates $(x,y)$ with $x,y\in\frac{1}{2}\mathbb{Z}$ and $0\leq x,y\leq d-1$. Data qubits have $x+y\in\mathbb{Z}$, ancilla qubits measuring $X$-stabilizers have $x\in\mathbb{Z}$ and $y\in\mathbb{Z}+\frac{1}{2}$, and $Z$-ancilla qubits have $x\in\mathbb{Z}+\frac{1}{2}$ and $y\in\mathbb{Z}$.

Recall that the bare circuit $\mathcal{C}$ is made up of gate layers occurring at half-integer time steps such that we can define a space-time location $(t,j)$ of $\mathcal{C}$ with $t,j\in\mathbb{Z}$. Then, in the encoded circuit $\overline{\mathcal{C}}$, a QEC round is implemented \textit{after} each logical gate at half-integer time-steps. With this, each physical qubit of $\overline{\mathcal{C}}$ is described by a tuple $(j,x,y)$ with $j\in\mathbb{Z}$ and $x,y\in\frac{1}{2}\mathbb{Z}$. 

Using the notation $U_{(j,x,y)}$ to denote the single-qubit gate $U$ acting on the qubit $(j,x,y)$ (and similarly for two-qubit gates), we can describe the exact implementation of the fold-transversal gates:
\begin{subequations}\label{eq:fold-transversal_gates}
\begin{gather}
    \overline{H}_{j}=\prod_{\substack{x,y\in\mathbb{Z}/2,\\x+y\in\mathbb{Z},\\0\leq x,y\leq d-1}}H_{(j,x,y)}\prod_{\substack{x,y\in\mathbb{Z}/2,\\x+y\in\mathbb{Z},\\0\leq x<y\leq d-1}}\mathrm{SWAP}_{(j,x,y),(j,y,x)}\\
    \overline{S}_{j}=\prod_{\substack{x\in\mathbb{Z},\\0\leq x\leq d-1}}S_{(j,x,x)}\prod_{\substack{x\in\mathbb{Z}+1/2,\\0\leq x\leq d-1}}S^{\dag}_{(j,x,x)}\prod_{\substack{x,y\in\mathbb{Z}/2,\\x+y\in\mathbb{Z},\\0\leq x<y\leq d-1}}CZ_{(j,x,y),(j,y,x)}\\
    \overline{\mathrm{CNOT}}_{j_{1},j_{2}}=\prod_{\substack{x,y\in\mathbb{Z}/2,\\x+y\in\mathbb{Z},\\0\leq x,y\leq d-1}}\mathrm{CNOT}_{(j_{1},x,y),(j_{2},x,y)}.
\end{gather}
\end{subequations}
\added{Note that for the fold-transversal $H$ and $S$ gates, there are SWAPs and CZs between physical qubits in coordinates $(x,y)$ and $(y,x)$ due to the 'folding' nature of these logical gates, see Fig. 1 from the main text. }


\subsection{The pre-gate frame}
\label{sec:pregate}

In this section, we discuss various ``frame'' definitions of detectors in the presence of logical gates, with a particular focus on the \textit{pre-gate} frame that we use throughout the main text. For a memory experiment without logical gates, all frames are equivalent and the detectors are defined in the standard way, i.e. 
\begin{equation}
    d(t,j,x,y)=\big\{m(t-1/2,j,x,y),m(t+1/2,j,x,y)\big\},
\end{equation}
where $m(t-1/2,j,x,y)$ refers to the measurement of the ancilla qubit $(x,y)$ in the $j$th logical qubit at the end of the QEC round at time $t-1/2$. Here it is implied that the detector takes the XOR of the measurement bits on which it depends.

We also define additional detectors for each reset and measurement in the standard way. In particular, for a $\ket{\overline{0}}$-reset occurring at location $(t-1/2,j)$ (in the bare circuit), we define additional $Z$-detectors $d(t,j,x,y)$ that each contain only a single measurement $m(t+1/2,j,x,y)$.

Meanwhile, for a $\overline{Z}$-measurement occurring at $(t+1/2,j)$ (in the bare circuit), there are physical $Z$-measurements on all data qubits $(x,y)$ that we label as $m(t+1/2,j,x,y)$. Then, we define the additional $Z$-detectors $d(t+1,j,x,y)$ containing the data qubit measurements $m(t+1/2,j,x\pm1/2,y\pm1/2)$~\footnote{Technically, to account for $Z$-stabilizers on the spatial boundaries of the surface code, we only include the qubit measurements $m(t+1/2,j,x',y')$ that satisfy $0\leq x'=x\pm1/2\leq d-1$, $0\leq y'=y\pm1/2\leq d-1$.} and the ancilla qubit measurement $m(t-1/2,j,x,y)$. The data qubit measurements here simply represent the support of the $Z$-stabilizer that is being measured by the corresponding ancilla qubit. We define analogous detectors for $\ket{\overline{+}}$-reset and $\overline{X}$-measurements.

We have assumed that $\ket{\overline{T}}$ magic states are initialized fault-tolerantly such that both its $X$ and $Z$ stabilizers are $+1$ with high probability. Therefore, for each $\ket{\overline{T}}$-reset occurring at $(t-1/2,j)$ in the bare circuit, the first round of both $X$- and $Z$-detectors $d(t,j,x,y)$ contains only one measurement $m(t-1/2,j,x,y)$ since this should be $+1$ in the absence of errors.

In the presence of non-trivial logical gates, the pre-gate frame is defined such that every detector $d(t,j,x,y)$ 
depends on
\begin{enumerate}
    \item one measurement $m(t-1/2,j,x,y)$ at time $t$, and
    \item any number of measurements at time $t+1/2$ that guarantee that the detecting region only spans two consecutive QEC rounds.
\end{enumerate}
The detailed definition of the detectors is given in Table~\ref{tab:detectors}.

\begin{table*}
\caption{Definition of detectors around a fold-transversal logical gate executed at time $t+1/2$. Note that $X$-detectors have spatial coordinates $(x,y)$ that satisfy $x\in\mathbb{Z}$ and $y\in\mathbb{Z}+\frac{1}{2}$, while $Z$-detectors satisfy $x\in\mathbb{Z}+\frac{1}{2}$ and $y\in\mathbb{Z}$.}\label{tab:detectors}
    \begin{tabular}{| c | c |}
        \hline
        Logical gate at $t{+}\frac{1}{2}$&Detector $d(t,j,x,y)$\\
        \hline
        $\overline{I}_{j}$&$\{m(t{-}\frac{1}{2},j,x,y),m(t{+}\frac{1}{2},j,x,y)\}$\\\hline
        $\overline{H}_{j}$&$\{m(t{-}\frac{1}{2},j,y,x),m(t{+}\frac{1}{2},j,x,y)\}$\\\hline
        $\overline{S}_{j}$&$\begin{array}{r c}X\text{-detectors: }&\{m(t{-}\frac{1}{2},j,x,y),m(t{+}\frac{1}{2},j,x,y),m(t{+}\frac{1}{2},j,y,x)\}\\Z\text{-detectors: }&\{m(t{-}\frac{1}{2},j,x,y),m(t{+}\frac{1}{2},j,x,y)\}\end{array}$\\\hline
        $\mathrm{CNOT}_{j_{1},j_{2}}$&$\begin{array}{r c}X_{j_{1}}\text{-detectors: }&\{m(t{-}\frac{1}{2},j_{1},x,y),m(t{+}\frac{1}{2},j_{1},x,y),m(t{+}\frac{1}{2},j_{2},x,y)\}\\Z_{j_{1}}\text{-detectors: }&\{m(t{-}\frac{1}{2},j_{1},x,y),m(t{+}\frac{1}{2},j_{1},x,y)\}\\X_{j_{2}}\text{-detectors: }&\{m(t{-}\frac{1}{2},j_{2},x,y),m(t{+}\frac{1}{2},j_{2},x,y)\}\\Z_{j_{2}}\text{-detectors: }&\{m(t{-}\frac{1}{2},j_{2},x,y),m(t{+}\frac{1}{2},j_{1},x,y),m(t{+}\frac{1}{2},j_{2},x,y)\}\end{array}$\\\hline
    \end{tabular}
\end{table*}

In contrast, the post-gate frame is defined such that every detector $d(t,j,x,y)$ only contains a single measurement at time $t+1/2$, but any number of measurements at time $t-1/2$.

The motivation for using this terminology for the detector frames is that incoming independent $X$ and $Z$ errors before (resp.~after) a logical gate lead to edges in the pre-gate frame (resp.~post-gate frame). 

Both the pre-gate and post-gate frames could be called ``co-moving'' frames because one associates a detector with a gate-transformed stabilizer generator immediately before (resp.~after) the logical gate, hence the names pre- and post-gate, respectively. Choosing a ``stationary'' frame in which one associates a detector with each stabilizer generator at the beginning of the circuit---thus requiring one to express the value of this generator in terms of multiple measured stabilizer outcomes at both times $t\pm 1/2$---leads to a detecting region which is local in time but possibly non-local in space. Choosing an ``uninformed'' frame would be to define a detector directly as the XOR of the bits measured for a given parity check at two consecutive times $t\pm 1/2$, which guarantees that measurement errors create only two defects, but leads to detecting regions that can be nonlocal both in space as well as in time, meaning that data qubit errors trigger a number of defects growing with the depth of the circuit.



\section{$G_O$ for an observable $O$ is a graph}
\label{sec:graph}

Here we provide an intuitive, general, argument why $G_O$ defined in Section~III.A from the main text is a graph. Near each logical gate we can consider so-called ``space-time stabilizers'' of the decoding hypergraph---a set of hyperedges that doesn't flip a detector \textit{or} any representation of a logical operator. A simple known way to construct a space-time stabilizer in a memory experiment is to take two time-like edges and two space-like edges in a square; this corresponds to an incoming qubit error into a QEC round, errors on the parity check measurements which would detect the errors so the error is not observed, and then a canceling qubit error after the QEC round. The generalization of this error pattern when we have a time-like hyperedge is to take two time-like hyperedges (corresponding to measurement errors) and three space-like edges (qubit errors on the edges between the vertices of the hyperedges). If the space-time stabilizer is located at a spatial boundary of the decoding hypergraph, then we only need one time-like hyperedge and three space-like edges, as shown in \cref{fig:space-time_stabilizer}. Because space-time stabilizers, by definition, do not flip any logical operators, the space-time stabilizer must have even overlap with any observing hyperedge region (which only contains space-like edges). Therefore, whenever there is a weight-3 time-like hyperedge $h \in \mathcal{H}$, then either zero or two of its endpoints can lie in $V_{O}$, i.e. $|h \cap V_{O}|=0 \mod 2$, so that the corresponding space-time stabilizer has even overlap with the observing hyperedge region $H_O$. This implies that $E_{O}$ contains no hyperedges. \added{An example of $G_{O}$ and its corresponding $E_O$ is shown in Fig.~4(c) from the main text.} 

\begin{figure}
    \begin{center}
        \includegraphics{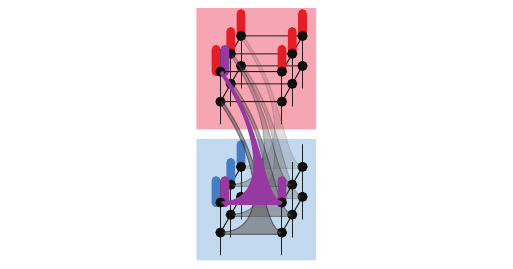}
    \end{center}
    \caption{An example of a space-time stabilizer (purple) that has even overlap with an observing hyperedge region (red and blue). This particular example represents the decoding graph around a logical $\overline{S}$-gate and is a part of the decoding graph in Fig.~4(b) in the main text.}\label{fig:space-time_stabilizer}
\end{figure}

\section{Circuit distance and hook errors} \label{sec:d_circ}

Let $\mathcal{F}$ be the set of individual errors that can occur in a circuit $\mathcal{C}$. The errors in $\mathcal{C}$ can be expressed by a vector $\vec{e} \in \mathbb{F}_2^{|\mathcal{F}|}$, with $e_i = 0$ if error $i$ did not occur. We label the number of detectors in circuit $\mathcal{C}$ as $n_D$ and the number of reliable logical observables as $n_L$. Following the same vector representation, the detectors that a set of errors $\vec{e}$ triggers, are given by $H_{\mathcal{C}}\vec{e}$ with $H_{\mathcal{C}} \in \mathbb{F}_2^{n_D \times |\mathcal{F}|}$ the parity-check matrix of circuit $\mathcal{C}$. 
Similarly, the reliable logical observables that $\vec{e} \in \mathbb{F}_2^{|\mathcal{F}|}$ flips are given by $L_{\mathcal{C}}\vec{e}$ with $L_{\mathcal{C}} \in \mathbb{F}_2^{n_L \times |\mathcal{F}|}$ the logical observable matrix of the circuit $\mathcal{C}$. 
The circuit distance of a circuit $\mathcal{C}$ with observables is 
\begin{equation}
	d_{\mathrm{circ}} := \min_{\vec{e} \in \mathbb{F}_2^{|\mathcal{F}|} : H_{\mathcal{C}}\vec{e} = \vec{0}, L_{\mathcal{C}}\vec{e} \neq \vec{0}} |\vec{e}|,
\end{equation}
that is, the minimum-weight error which does not trigger any detectors but has a logical effect. 
The circuit distance values reported in the main text have been obtained by solving the equivalent MaxSAT problem~\cite{gidney2021stim}, whose description for each experiment is generated using Stim~\cite{gidney2021stim}, in particular,
\begin{verbatim}
    stim.Circuit.shortest_error_sat_problem
\end{verbatim}
Then, the problem is solved using the RC2 (relaxable cardinality constraints) algorithm from the \texttt{PySAT} Python package. We have written a wrapper for computing the circuit distance directly from a Stim circuit and it can be found in~\cite{qec-util}. 

Because of the hardness of this problem, we have only computed the circuit distances for experiments involving $d=3$ and $d=5$ surface codes. 


\subsection{Distance-reducing hook errors for the \textsc{lom} decoder}
\label{sec:hook}
In this section we discuss a circuit-level error of weight $d-1$ which can occur in the repeated-$\overline{S}$ experiment and which is undetectable by the \textsc{lom} decoder, possibly affecting the numerical results in Section~IV.D.1 from the main text. 
Fig.~\ref{fig:bad_error_s_gate} shows such weight-4 error for the distance-5 unrotated surface code. The error consists of three incoming errors and a hook error, see Fig.~\ref{fig:bad_error_s_gate}(b). Note that a hook error is not included in the basic error model, i.e.~it requires more than 1 error in the basic model. 
Also, the error is not a logical Pauli (up to stabilizers) as the circuit distance of the experiment is $d_{\mathrm{circ}} = d = 5$. It corresponds to a $\overline{Z}$ error plus two $X$ data-qubit errors, which are not detected by the \textsc{lom} decoder because the triggered $Z$-detectors have been removed in the decoding subgraph, as drawn in Fig.~\ref{fig:bad_error_s_gate}(c). 
For a distance-$d$ unrotated surface code, this error is built using $d-2$ incoming errors along the code boundary and a hook error located at the corner of the fold for the $\overline{S}$ gate, see Fig.~\ref{fig:physical_circuits}. \added{The effect of these distance-reducing errors on the logical performance can be observed in Fig.~\ref{fig:threshold_repeated_scaling}(c).}

\begin{figure*}[tb]
    \centering
    \includegraphics[width=0.99\linewidth]{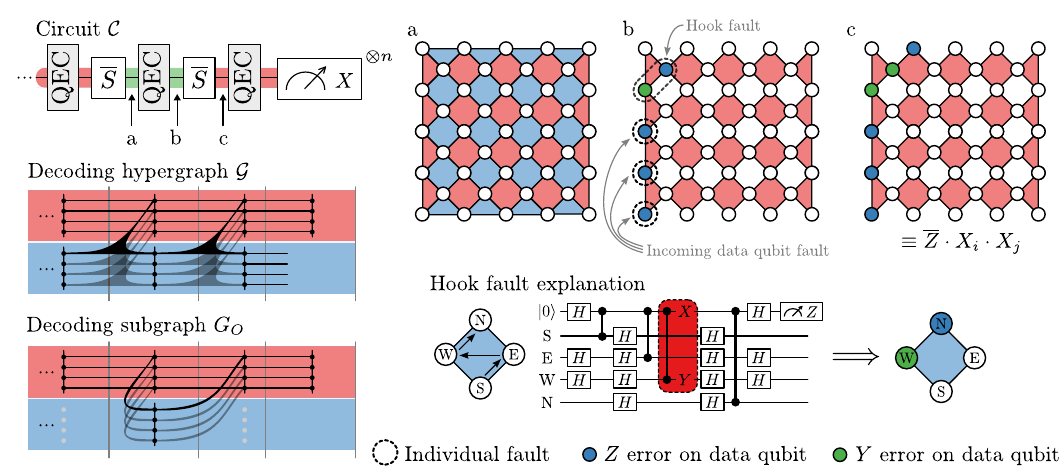}
    \caption{Example of a weight-4 error which is undetectable when decoding the sequence of repeated $\overline{S}$ gates by the single-\textsc{lom} decoder which only considers the subgraph $G_O$ of the full hypergraph $G$. 
    The four errors are surrounded by a dashed line in panel b and happen between time-points a and b: there are three errors on data qubits after the QEC round and one weight-2 XY error on the CZ inside the QEC round which spreads. Due to $\overline{S}$ gate the error transforms into a $\overline{Z}$ with two $X$ errors on top, but the latter errors are not detected by the detectors included in $G_O$. The noise model in the circuit is circuit-level depolarizing noise although only the basic error model is depicted in the decoding graphs for visual simplification. The detectors are defined in the post-gate frame to simplify the error pattern. Similar undetectable error patterns can be generated for any code distance and valid CZ ordering in the QEC round. }
    \label{fig:bad_error_s_gate}
\end{figure*}

\section{Circuit implementation}
\label{sec:schedule}

\begin{figure*}[tb]
    \centering
    \includegraphics[width=0.99\textwidth]{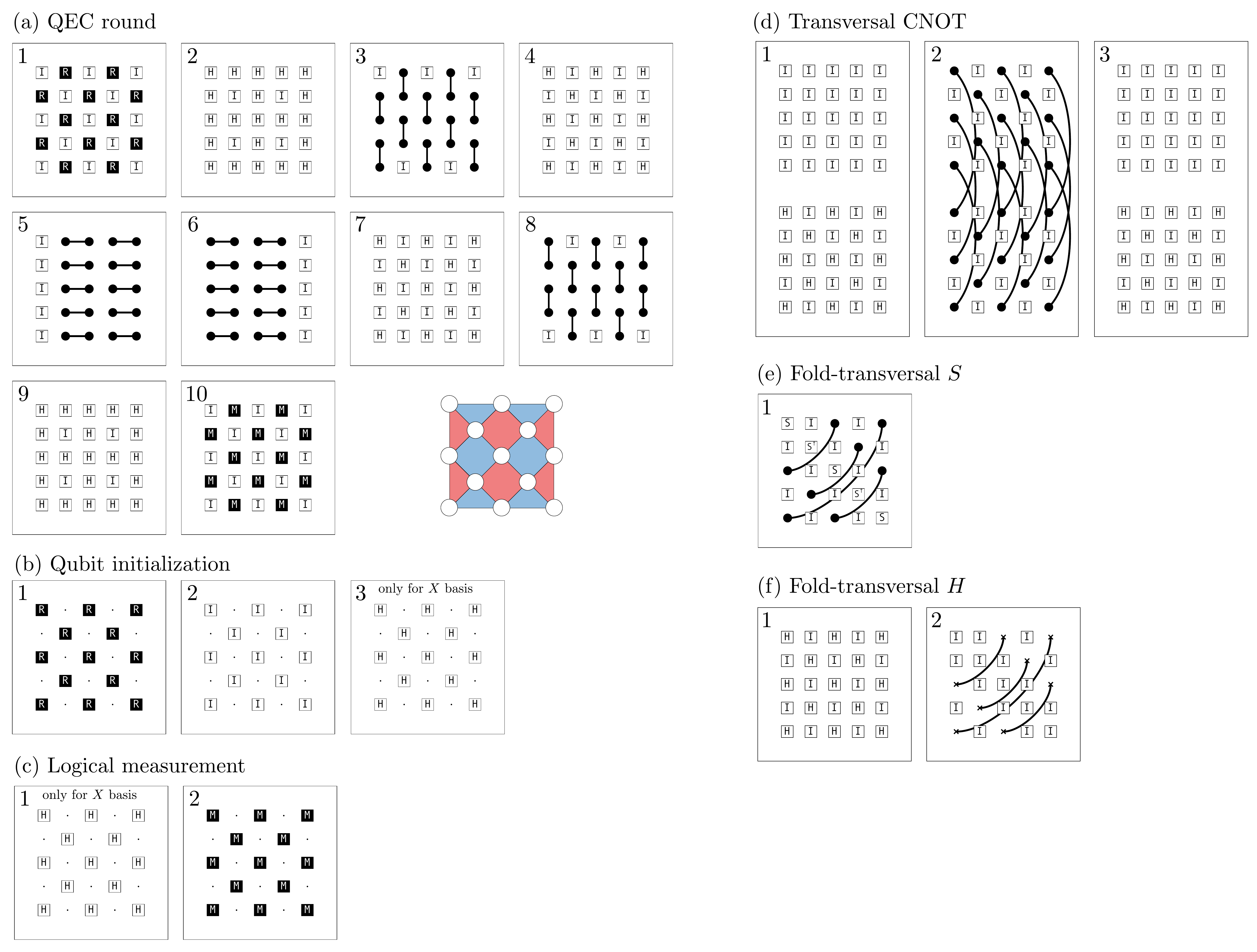}
    \caption{Physical circuits for the $d=3$ unrotated surface code used in the numerical experiments. The circuits are represented as gate layers that are applied in order. Layers labeled with ``only for $X$ basis'' are only applied if the experiment is run in the $X$ basis. The boxes with ``R'' and ``M'' represent a reset to $\ket{0}$ and a measurement in the $Z$-basis, respectively. \added{To benchmark our decoder in conditions similar to future experimental realizations of logical circuits, the CNOT gates and $X$-basis measurements have been decomposed into operations from the primitive gate set of superconducting qubits and neutral atoms. This decomposition is consistent with the noise and operations available in the SI1000 noise model.} }
    \label{fig:physical_circuits}
\end{figure*}

This section describes the physical implementation of the individual circuit blocks described in the main text and used in the numerical simulations. The circuits for the $d=3$ experiments are shown in Fig.~\ref{fig:physical_circuits}. Note that in the qubit initialization, the second layer is just qubit idling. We have added this idling layer to model the noisier preparation in $\ket{1}$ versus $\ket{0}$. 

The CNOT gates in the $\overline{\mathrm{CNOT}}$ implementation have been compiled down to $H$ and CZ to use a primitive gate set of, say, neutral atoms. As neutral atoms or trapped ions in some set-ups can be moved and rearranged, we have considered that SWAP gates are part of the primitive gate set, or, alternatively, one works with a folded surface such that the SWAP acts very locally. Therefore, the SWAP gates in the $\overline{H}$ gate have not been decomposed. 

\section{Additional numerical results} \label{sec:additional_results}

This section shows the scaling of the logical error probability for the repeated-gate experiments under circuit-level noise for a larger range of physical error probabilities $p$ than in Fig.~12(f)--(j), see Fig.~\ref{fig:threshold_repeated_scaling}. The worse scaling of the logical error probability for the repeated-$\overline{S}$ experiment in the $\overline{X}$ basis described in Section~IV.B.1 from the main text and Section~\ref{sec:hook} in this SM is visible for the $d=3$ and 5 unrotated surface codes. Finally, Table~\ref{tab:thresholds} also includes the computed thresholds for the experiments, which have been obtained using the method from Ref.~\cite{hillmann2024single}. \added{Note that the threshold for memory experiments under SI1000 noise is lower than the standard one under circuit-level noise (at around 0.75\%-1\%~\cite{Dennis_2002}) due to the differences between the noise models. Therefore, to correctly interpret the \textsc{lom} performance when decoding different logical gates or circuits, one has to do relative comparisons with respect to the SI1000-memory-experiment threshold.}

\begin{table*}[tb]
\caption{Threshold values (in \%) for the different experiments, noise models, and basis obtained using the method from Ref.~\cite{hillmann2024single}. Note that the repeated-$\overline{I}$ experiment corresponds to a memory experiment. }
\label{tab:thresholds}
\begin{tabular}{ccccccc}
\hline
\multicolumn{7}{c}{\textbf{SI1000 circuit-level noise}}                                               \\ \hline
        & \multicolumn{5}{c|}{\textbf{Repeated-gate experiments}}                & \textbf{Two-qubit}            \\
        & $\overline{I}$     & $\overline{H}$     & $\overline{S}$     & $\overline{\mathrm{CNOT}}$  & \multicolumn{1}{c|}{alt-$\overline{\mathrm{CNOT}}$} & \textbf{Clifford experiments} \\ \hline
$\overline{X}$ basis & 0.4508\added{(4)} & 0.4200\added{(4)} & 0.3236\added{(4)} & 0.3398\added{(3)} & 0.3337\added{(3)}                         & 0.4963\added{(1)}                \\
$\overline{Z}$ basis & 0.4360\added{(4)} & 0.4214\added{(4)} & 0.4156\added{(4)} & 0.3275\added{(3)} & 0.3218\added{(3)}                         & 0.5036\added{(1)}                \\ \hline
\multicolumn{7}{c}{\textbf{Phenomenological   depolarizing noise}}                                      \\ \hline
        & \multicolumn{5}{c|}{\textbf{Repeated-gate experiments}}                & \textbf{Two-qubit}            \\
        & $\overline{I}$     & $\overline{H}$     & $\overline{S}$     & $\overline{\mathrm{CNOT}}$  & \multicolumn{1}{c|}{alt-$\overline{\mathrm{CNOT}}$} & \textbf{Clifford experiments} \\ \hline
$\overline{X}$ basis & 2.247\added{(2)} & 2.245\added{(2)} & 1.625\added{(2)} & 1.786\added{(2)} & 1.736\added{(1)}                         & 2.8212\added{(2)}                \\
$\overline{Z}$ basis & 2.247\added{(2)} & 2.250\added{(2)} & 2.248\added{(2)} & 1.786\added{(2)} & 1.737\added{(1)}                         & 2.9576\added{(2)}                \\ \hline
\end{tabular}
\end{table*}

\begin{figure*}[tb]
    \centering
    \includegraphics[width=0.99\textwidth]{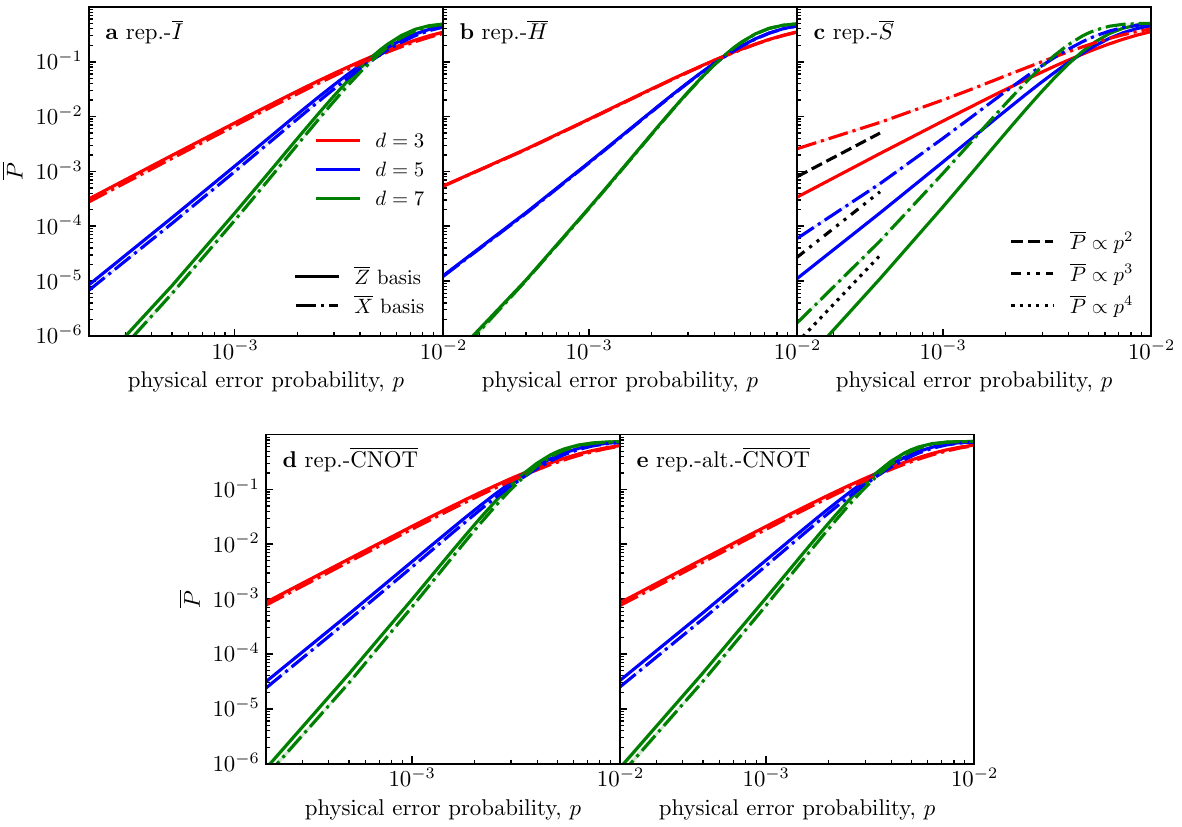}
    \caption{Extended version of Fig.~12(f)--(j) showing the scaling of the logical error probability of the \textsc{lom} decoder for repeated-gate experiments under circuit-level noise in the pre-gate frame. 
    Dashed/dotted black lines are included for the repeated-$\overline{S}$ experiment to emphasize that the scaling in the $\overline{X}$ basis is not $p^{\lfloor(d+1)/2\rfloor}$. 
    \added{The 95\% confidence intervals are represented as shaded regions, but they are too small to be visible.} 
    }
    \label{fig:threshold_repeated_scaling}
\end{figure*}

\FloatBarrier
\clearpage
\bibliography{references}

\end{document}